\pdfoutput=1
\documentclass[prd,superscriptaddress,twocolumn,preprintnumbers,nofootinbib,%
  amsmath,amssymb,longbibliography]{revtex4-1}

\usepackage[breaklinks,colorlinks=true]{hyperref}
\usepackage[T1]{fontenc}
\usepackage[utf8]{inputenc}
\usepackage{mathtools}
\usepackage{amsmath,amsthm,amssymb}
\usepackage[dvipsnames]{xcolor}
\usepackage{bm}
\usepackage{bbold}
\usepackage{dsfont}
\allowdisplaybreaks
\usepackage{makecell}
\usepackage{booktabs}

\newcommand\myshade{85}
\colorlet{mylinkcolor}{BrickRed}
\colorlet{mycitecolor}{NavyBlue}
\colorlet{myurlcolor}{Aquamarine}

\hypersetup{
  linkcolor  = mylinkcolor!\myshade!black,
  citecolor  = mycitecolor!\myshade!black,
  urlcolor   = myurlcolor!\myshade!black,
  colorlinks = true,
}

\begin{document}

\title{Shannon observables: A scalable approach to information decomposition}

\author{Aaron J. Gutknecht}
\email{agutkne@uni-goettingen.de}
\affiliation{Göttingen Campus Institute for Dynamics of Biological Networks, Georg-August University Göttingen}

\author{Fernando E. Rosas}
\email{f.rosas@sussex.ac.uk}
\affiliation{Sussex Centre for Consciousness Science and Sussex AI, Department of Informatics, University of Sussex}
\affiliation{Center for Psychedelic Research and Centre for Complexity Science, Department of Brain Science, Imperial College London}
\affiliation{Center for Eudaimonia and Human Flourishing, University of Oxford}
\affiliation{Principles of Intelligent Behaviour in Biological and Social Systems (PIBBSS)}

\author{David A. Ehrlich}
\affiliation{Göttingen Campus Institute for Dynamics of Biological Networks, Georg-August University Göttingen}
\affiliation{Max Planck Institute for Dynamics and Self-Organization, Göttingen}

\author{Abdullah Makkeh}
\affiliation{Göttingen Campus Institute for Dynamics of Biological Networks, Georg-August University Göttingen}
\affiliation{Max Planck Institute for Dynamics and Self-Organization, Göttingen}

\author{Pedro A. M. Mediano}
\affiliation{Department of Computing, Imperial College London}
\affiliation{Division of Psychology and Language Sciences, University College London}

\author{Michael Wibral}
\affiliation{Göttingen Campus Institute for Dynamics of Biological Networks, Georg-August University Göttingen}
\affiliation{Max Planck Institute for Dynamics and Self-Organization, Göttingen}

\newcommand{\equalcontrib}{These authors contributed equally to this work.}

\newtheorem{definition}{Definition}
\newtheorem{conjecture}{Conjecture}
\newtheorem{theorem}{Theorem}
\newtheorem{lemma}{Lemma}
\newtheorem{proposition}{Proposition}
\newtheorem{corollary}{Corollary}
\newtheorem{example}{Example}
\newtheorem{remark}{Remark}

\begin{abstract}
\noindent
Distributed systems, such as biological and artificial neural networks, process information via complex interactions engaging multiple subsystems, resulting in high-order structure of the information within the system with distinct properties across scales. Investigating how these systems process information remains challenging due to
difficulties in defining appropriate multivariate metrics and ensuring their scalability to large systems. To address these challenges, we introduce a novel framework based on what we call \textit{Shannon observables}: quantities that capture essential properties of high-order information processing, yet can be computed directly from Shannon entropies and scale efficiently to large systems. Our theoretical results demonstrate how Shannon observables can be used to resolve long-standing ambiguities regarding the interpretation of widely used multivariate information-theoretic measures.
On the practical side, we demonstrate the framework's versatility across several domains. An analytical study of error-correcting codes shows how Shannon observables trace out distinct encoding regimes as channel noise varies; analyses of deep classifiers and convolutional autoencoders reveal characteristic information-processing signatures across layers and over training; and applied to magnetoencephalography recordings from the human brain, the framework yields a cortical atlas exposing pronounced regional structure in directed information transfer. Overall, our framework resolves fundamental limitations in analysing high-order information phenomena and offers broad opportunities for theoretical developments and empirical analyses.

\end{abstract}

\maketitle
\begingroup
\renewcommand\thefootnote{}\footnotetext{\equalcontrib}
\addtocounter{footnote}{-1}
\endgroup

\section{Introduction}

Understanding how information is processed and distributed across systems is a fundamental challenge in multiple scientific domains, from genetics and neuroscience to machine learning and artificial intelligence~\cite{waldrop1993complexity,gleick2011information,bengio2024machine}. Advances in our understanding of distributed information phenomena have led to powerful new insights into, for example, the inner workings of natural evolution~\cite{rajpal2025information}, genetic information flow~\cite{cang2020inferring,park2021higher}, psychometric interactions~\cite{marinazzo2022information,varley2022untangling}, cellular automata~\cite{rosas2018information,orio2023dynamical},
and polyphonic music~\cite{scagliarini2022quantifying}.
Analyses based on distributed information phenomena have been particularly useful for studying biological~\cite{gatica2021high,luppi2022synergistic,varley2023information,varley2023multivariate,herzog2024high,pope2024time,varley2024emergence} and artificial~\cite{tax2017partial,proca2024synergistic,kaplanis2023learning,kong2023interpretable} neural systems, shedding light on the distinct roles of qualitatively different types of information in neural dynamics~\cite{luppi2024information}.

The \emph{partial information decomposition} (PID) framework provides a comprehensive theoretical solution to the challenge of understanding distributed information processing by decomposing information carried by multiple source variables about a target into redundant, synergistic, and unique contributions~\cite{williams2010nonnegative,lizier2018information,mediano2025toward, gutknecht2021bits, lizier2013towards}.
PID has been used, for example, to investigate how artificial neural networks encode and distribute information across multiple nodes or layers~\cite{proca2024synergistic,ehrlich2023a}, to develop new types of self-organising AI systems that explicitly optimise specific PID-defined objectives~\cite{schneider2025should,makkeh2025general, wibral2017partial}, to understand how features should be selected
information-theoretically in machine learning \cite{wollstadt2023rigorous}, and to explore how redundancies and synergies complement each other in balancing robustness and flexibility~\cite{orio2023dynamical,varley2024evolving,luppi2024information}.

Despite its promise, PID suffers from important shortcomings related to its operationalisation and scalability. In fact, there are multiple ways to measure redundancy and synergy \cite{williams2010nonnegative, rosas2020operational, finn2018pointwise, makkeh2021isx, Kolchinsky2022, van2023pooling}, and it is not always straightforward to pick the optimal measure for a specific analysis (for a comprehensive overview and categorisation of these measures, see \cite{Gutknecht2025babel, liardi2026mathematical}). Moreover, the computational complexity of PID grows super-exponentially with the number of variables involved, rendering the full decomposition computationally infeasible beyond a handful of variables~\cite{jansma2025fast, gutknecht2021bits}.
Together, these limitations seriously hinder the capability of PID to address practical applications.

Here we address these challenges by introducing a novel approach to information decomposition inspired by an analogy to statistical mechanics. In statistical mechanics, systems with an enormous number of particles are analysed by focusing on global observables, such as their average energy, which can be calculated without knowledge of microscopic details~\cite{jaynes1957information,rosas2024software}. Similarly, in the context of PID, one deals with a vast number of `information atoms' which are very hard to estimate for large systems. We demonstrate that certain global properties, which are of high relevance to understanding distributed information processing, can be computed without resolving the full PID. We call these properties \emph{Shannon observables}, given that they can be calculated from Shannon entropies alone --- circumventing the need to specify a particular PID measure. We introduce two such observables: the \textit{average degree of redundancy}, which captures the extent to which information can be accessed through multiple variables, and the \textit{average degree of vulnerability}, which reflects how information becomes inaccessible when specific sources are removed. Calculating these two observables requires evaluating a number of entropies that grows only linearly with the number of variables, avoiding the super-exponential growth associated with a full information decomposition.

We illustrate the benefits of these two Shannon observables in the theoretical and practical domains. On the theoretical side, we show how the degree of redundancy provides an explanation of what is being calculated by the `redundancy-synergy index'~\cite{Chechik2002_group}, a well-known metric in the computational neuroscience literature whose interpretation has nevertheless remained elusive~\cite{timme2014synergy,rosas2024characterising}. We also show how the degree of vulnerability leads to a novel quantity, which we call `dual redundancy-synergy index'. On the practical side, we demonstrate the versatility of Shannon observables across several domains. First, an analysis of error-correcting codes reveals how redundancy and vulnerability trace out a rich landscape of encoding regimes as noise varies, providing concrete intuition for the distinct aspects of information structure captured by the two measures. Furthermore, we show how Shannon observables allow us to analyse deep learning systems, including a feed-forward image classifier and an autoencoder. For these systems, our results show that Shannon observables track distinct representational properties between layers and also distinct learning dynamics, providing useful tools for AI interpretability. Finally, we extend the framework to directed information flow and apply it to magnetoencephalography (MEG) recordings from the human brain, obtaining a cortical atlas of incoming information transfer that exhibits pronounced regional structure in both redundancy and vulnerability.

\section{Results}

\subsection{Shannon observables are  scalable signatures of multivariate information}

Consider $n$ source variables $\bm X = (X_1,\dots,X_n)$ that jointly carry information about a target variable $Y$, as quantified by the mutual information $I(\bm X;Y)$. PID decomposes this quantity into non-overlapping \emph{information atoms} corresponding to qualitatively distinct modes of encoding (see Methods \ref{sec:PID}).

To illustrate, consider the simplest case of two sources (Fig.~\ref{fig:overview}a). PID decomposes the joint information as
\begin{equation}\label{eq:pid_n2}
    I(X_1,X_2;Y) = \Pi_\text{red} + \Pi_\text{unq1} + \Pi_\text{unq2} + \Pi_\text{syn},
\end{equation}
where $\Pi_\text{red}$ is the information \emph{redundantly} available from both $X_1$ and $X_2$, $\Pi_\text{unq1}$ and $\Pi_\text{unq2}$ are the information \emph{unique} to $X_1$ and $X_2$ respectively, and $\Pi_\text{syn}$ is \emph{synergistic} information accessible only when both sources are observed together.
These four atoms exhaust everything the two sources jointly reveal about $Y$, and each captures a qualitatively different way in which the sources relate to the target \cite{gutknecht2021bits}. However, for larger systems the number of such atoms grows super-exponentially with $n$, and their precise values depend on which of several competing PID definitions one adopts (see Methods~\ref{sec:PID} for the general formalism). Together, these issues seriously hinder the applicability of the full decomposition.

We now show that these challenges can be circumvented by considering certain aggregates of information atoms. Concretely, we call a quantity a \emph{Shannon observable} if it is a function of PID atoms whose value is determined by Shannon entropies alone.
Thus, Shannon observables are collective properties of the information decomposition that are oblivious to how individual atoms are quantified. For technical details, see Methods \ref{sec:shannon_observables_formal}. Below, we introduce two Shannon observables---the average degree of redundancy and the average degree of vulnerability---that require only a number of entropy evaluations that grows linearly with $n$.

\begin{figure*}[t]
    \centering
    \includegraphics[width=0.9\textwidth]{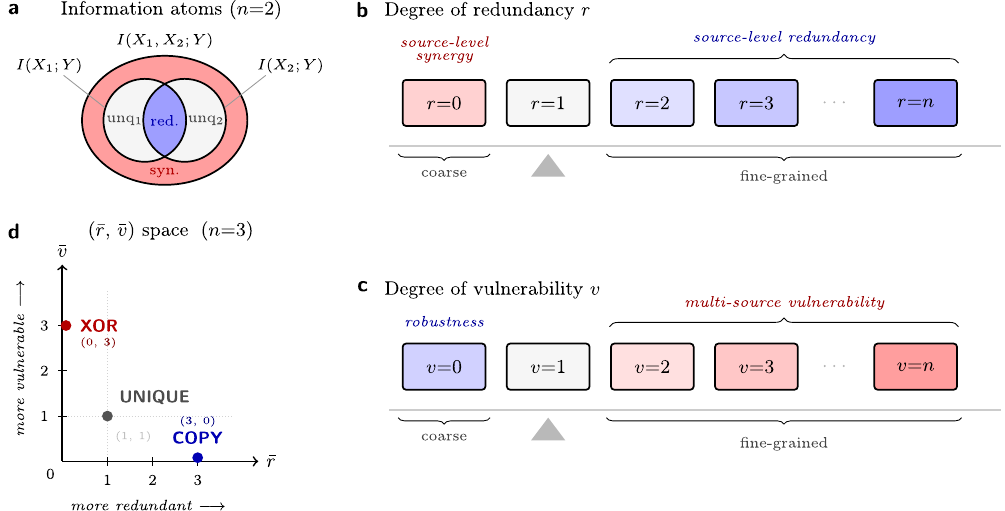}
\caption{\textbf{Two complementary perspectives on how information is distributed across source variables.}
\textbf{(a)}~For two sources, PID decomposes the joint information $I(X_1,X_2;Y)$ (outer ellipse) into information redundantly available from each of the sources (blue), information unique to each source (grey), and synergistic information accessible only from the joint observation (red).
\textbf{(b)}~Each information atom in PID has a \emph{degree of redundancy} $r$, counting how many individual sources can provide it. Atoms with $r{=}0$ are source-level synergistic (not accessible from any single source); atoms with $r{\geq}2$ are source-level redundant. The average degree of redundancy $\bar{r}$ is the weighted average of $r$ across all atoms. The grey triangle at $r{=}1$ acts like the balance point of an asymmetric seesaw.
\textbf{(c)}~Dually, each atom has a \emph{degree of vulnerability} $v$, counting how many sources it critically depends on. Atoms with $v{=}0$ are robust (no single source is critical); atoms with $v{\geq}2$ have multi-source vulnerability. The same seesaw structure applies in mirror.
\textbf{(d)}~The two measures span complementary axes of information structure, illustrated for three canonical systems with $n{=}3$ binary sources each: COPY (maximal redundancy), XOR (maximal synergy), and UNIQUE (purely unique information).}
\label{fig:overview}
\end{figure*}

\paragraph*{\textbf{Average degree of redundancy}.}
Each information atom can be characterised by how many individual source variables can, on their own, provide access to that information. We call this the atom's \emph{degree of redundancy}, denoted $r$ (Fig.~\ref{fig:overview}b; see Methods \ref{sec:shannon_observables_formal}). For example, an atom that can be read off from three different individual sources has $r=3$; in contrast, a synergistic atom cannot be obtained from any single source, and hence has $r=0$.

The \emph{average degree of redundancy}, $\bar{r}$, is the weighted average of $r$ across all atoms, where each atom is weighted by its share of the total information $I(\bm X;Y)$. Intuitively, $\bar{r}$ answers the question: ``for a typical piece of information about $Y$, how many individual sources can provide it?'' We show (Proposition~\ref{prop:r_bar_shannon_observable}; proof in Methods \ref{sec:shannon_observables_formal}) that $\bar{r}$ is a Shannon observable that can be calculated as
\begin{equation}\label{eq:rbar}
    \bar{r} = \frac{\sum_{i=1}^n I(X_i;Y)}{I(\bm X;Y)}.
\end{equation}
This quantity takes values in $[0,n]$, with zero indicating that no individual source carries any information about the target alone, and $n$ implying that all information is fully redundant between all sources.
When information is highly redundant, individual sources each capture a large share of the total, and $\bar{r}$ is large. When information is predominantly synergistic, individual sources carry little information on their own, and $\bar{r}$ is close to zero.

\paragraph*{\textbf{Average degree of vulnerability}.}
A complementary perspective asks not how many sources \emph{can} provide a given atom, but how many sources it \emph{critically depends on} --- that is, how many sources are strictly necessary, in the sense that losing any one of them would render the atom inaccessible. We call this the atom's \emph{degree of vulnerability}, denoted $v$ (Fig.~\ref{fig:overview}c; see Methods \ref{sec:shannon_observables_formal}). An atom with $v=0$ does not critically depend on any individual source and is therefore \emph{robust}; in contrast, an atom with $v=3$ relies on three specific sources, making it much more vulnerable to disruption.

The \emph{average degree of vulnerability}, $\bar{v}$, is the corresponding weighted average across all atoms. It answers: ``for a typical piece of information about $Y$, how many sources does it critically depend on?'' We show (Proposition~\ref{prop:v_bar_shannon_observable}; proof in Methods \ref{sec:shannon_observables_formal}) that $\bar{v}$ is also a Shannon observable given by
\begin{equation}\label{eq:vbar}
    \bar{v} = \frac{\sum_{j=1}^n I(X_j;Y|\bm X_{-j})}{I(\bm X;Y)},
\end{equation}
where $\bm X_{-j}$ denotes all sources except $X_j$.
This quantity also takes values in $[0,n]$, with $0$ indicating that information is fully robust to the loss of any individual source, and $n$ indicating that all information is lost upon removal of any single source. Generally, small values indicate more robust encodings and large values indicate more vulnerable ones.

\paragraph*{\textbf{Basic examples}.}
To build intuition, consider three canonical systems with $n=3$ binary sources (Fig.~\ref{fig:overview}d). In the \textbf{COPY} system ($X_1=X_2=X_3=Y$), every source individually determines the target, so $\bar{r}=3$: information is maximally redundant. Since knowing any source suffices to infer the target, no single source is critical, giving $\bar{v}=0$: information is completely robust. In the \textbf{XOR} system ($Y = X_1\oplus X_2 \oplus X_3$), no individual source carries any information about $Y$, so $\bar{r}=0$. Yet each source is strictly necessary, so $\bar{v}=3$: information is maximally vulnerable. Finally, in the \textbf{UNIQUE} system ($Y=X_1$, with $X_2,X_3$ independent of $Y$), all information passes through exactly one source and depends on exactly that source, yielding $\bar{r}=\bar{v}=1$.

\begin{table*}[t]
\centering
\caption{\textbf{Comparison of approaches for quantifying redundancy and synergy.} `Directed' indicates whether the method distinguishes source and target variables. `PID-independent' indicates whether the result is invariant to the choice of PID measure. `Scaling' refers to the asymptotic growth, with the number of source variables $n$, of the number of information-theoretic functionals that must be evaluated, whether Shannon entropies, mutual informations, or PID atoms. The cost of evaluating each individual functional is a separate concern shared by all methods. $D(n)$ stands for the Dedekind numbers. }
\label{tab:comparison}
\renewcommand{\arraystretch}{1.3}
\begin{tabular}{l c c l}
\toprule
\textbf{Method} & \textbf{Directed\quad} & \textbf{PID-independent\quad} & \textbf{Scaling} \\
\midrule
\multicolumn{4}{l}{\textit{Directed methods (source $\to$ target)}} \\[4pt]
Full PID \cite{williams2010nonnegative}
  & \checkmark & $\times$
  & Super-exponential ($D(n){-}2$ terms) \\
$\Phi$ID~\cite{mediano2025toward}
  & \checkmark & $\times$
  & Super-exponential ($\bigl(D(n){-}2\bigr)^2$ terms) \\
Backbone decomposition~\cite{rosas2020operational}
  & \checkmark & $\times$
 & \makecell[l]{Super-exponential (redundancy-based);\\linear (synergy-based)$^{\dagger}$} \\
Representational complexity~\cite{ehrlich2023a}
  & \checkmark & $\times$
 & \makecell[l]{Super-exponential (redundancy-based);\\linear (synergy-based)$^{\dagger}$} \\
$\alpha$-synergy backbone~\cite{varley2024scalable}
  & $\checkmark^*$ & $(\checkmark)^*$
  & Exponential \\
$\bar{r}$ (this work) / RSI \cite{Chechik2002_group}
  & \checkmark & \checkmark
  & Linear \\
$\bar{v}$ (this work) / DRSI (this work)
  & \checkmark & \checkmark
  & Linear \\
\midrule
\multicolumn{4}{l}{\textit{Undirected methods (no source--target distinction)}} \\[4pt]
O-information~\cite{rosas2019quantifying}
  & $\times$ & \checkmark
  & Linear \\
TC \cite{watanabe1960information} / DTC \cite{te1978nonnegative}
  & $\times$ & \checkmark
  & Linear \\
Interaction information~\cite{mcgill1954multivariate}
  & $\times^{\ddagger}$ & \checkmark
  & Exponential\\
TSE complexity \cite{tononi1994measure}
& $\times$ & \checkmark
& Exponential \\
\bottomrule
\end{tabular}

\vspace{4pt}
{\footnotesize ${}^{\dagger}$ Both backbone decomposition and representational complexity have super-exponential scaling when evaluated from a redundancy-based PID ($D(n)-2$ atoms), and linear scaling in a synergy-based formulation ($n$ backbone terms)~\cite{rosas2020operational,ehrlich2023a}. \\[3pt] ${}^*$Directed and undirected versions exist. Not PID-based, but requires choosing an $\alpha$-synergy function. \\[3pt] ${}^{\ddagger}$Requires no source--target partition, but acquires a directed interpretation when one variable is designated as target.}

\end{table*}

\subsection{Redundancy and synergy admit two complementary interpretations}
\label{sec:different_views}
The measures $\bar{r}$ and $\bar{v}$ are not simply inverses of each other. They capture genuinely different aspects of how information is structured, reflecting two distinct ways to conceptualise redundancy and synergy.

The first perspective is \emph{source-level redundancy and synergy}. Here, redundancy means that information can be obtained from more than one individual source ($r > 1$), and source-level synergy means that information cannot be obtained from any individual source at all ($r = 0$). This is the perspective captured by $\bar{r}$.

The second perspective is \emph{robustness and multi-source vulnerability}. Robustness ($v=0$) means that the information survives the loss of any single source; it is redundantly protected, though not necessarily by being available from individual sources alone. Multi-source vulnerability ($v>1$) means that the information critically depends on more than one specific source, making it fragile. This is the perspective captured by $\bar{v}$.

These two perspectives are related but not equivalent. Source-level redundancy ($r>1$) implies robustness ($v=0$): if information is individually accessible from multiple sources, it is certainly robust to losing any one of them. Conversely, multi-source vulnerability ($v>1$) implies source-level synergy ($r=0$): if information critically depends on multiple specific sources, it certainly cannot be accessed from any single source. However, neither converse holds. For instance, in a system with three sources, information that is obtainable from any \emph{pair} of them but not from any individual source has $r=0$ and $v=0$ simultaneously. It is source-level synergistic because no individual source suffices, yet robust because losing any single source still leaves at least one complete pair intact.

Importantly, the two perspectives coincide only in the simplest case of $n=2$ sources, where redundancy (in either sense) is simply the redundancy atom $\Pi_\text{red}$ and synergy (in either sense) is the synergy atom $\Pi_\text{syn} $. For larger systems, they diverge, and $\bar{r}$ and $\bar{v}$ can exhibit qualitatively different behaviour. Recognising this distinction is essential for correctly interpreting any measure that aims to quantify the ``balance'' of redundancy and synergy.

\subsection{Interpreting the redundancy-synergy index and introducing a dual measure}

The framework developed above sheds new light on the \emph{redundancy-synergy index} (RSI)~\cite{Chechik2002_group}, one of the most widely used metrics for assessing the balance between redundant and synergistic information (see e.g.~\cite{timme2014synergy,brenner2000synergy,gawne1993independent,griffith2014quantifying,yang2004feature,mares2022selection,mares2023combined,luecke2024dynamical}). The RSI is defined as
\begin{equation}\label{eq:rsi_def}
    \text{RSI}(\bm X;Y) := \sum_{i=1}^n I(X_i;Y) - I(\bm X;Y).
\end{equation}
Its intuition is straightforward: if the sum of individual contributions exceeds the whole, the ``excess'' reflects redundancy; if the whole exceeds the sum, the ``deficit'' reflects synergy. For $n=2$ sources, PID analysis confirms that the RSI equals exactly redundancy minus synergy~\cite{williams2010nonnegative}. However, whether this clean interpretation generalises to larger systems has remained an open question.

Our framework provides the answer. From Eqs.~\eqref{eq:rbar} and~\eqref{eq:rsi_def} it follows that
\begin{equation}\label{eq:rsi_rbar}
    \text{RSI}(\bm X;Y) = (\bar{r}-1)\, I(\bm X;Y).
\end{equation}
Assuming non-negative PID atoms, a positive RSI (i.e., $\bar{r}>1$) indicates the presence of source-level redundant information, while a negative RSI (i.e., $\bar{r}<1$) indicates the presence of source-level synergistic information. Source-level synergy is guaranteed to account for more than half of the total when $\bar{r}<1/2$ (Appendix~\ref{app:bounds}). This confirms the broad intuition behind the RSI. That said, since the RSI is governed by $\bar{r}$, it directly inherits the asymmetry visible in Fig.~\ref{fig:overview}b: higher-order redundancies contribute more strongly, while all synergistic contributions ($r=0$) are treated as a single undifferentiated category. For $n=2$ this asymmetry is invisible, since there is only one order of redundancy and one synergistic category. For $n \geq 3$, however, the distinction becomes consequential and should be kept in mind when interpreting the RSI's magnitude.

The existence of two complementary perspectives on redundancy and synergy naturally invites the definition of a second index. If the RSI captures the balance viewed through the lens of source-level redundancy and synergy (the $\bar{r}$ perspective), one can ask: how does such a balance look through the lens of robustness and vulnerability (the $\bar{v}$ perspective)? This motivates a new metric that we call \emph{dual redundancy-synergy index} (DRSI), which is also a Shannon observable (see Methods \ref{sec:rsi_drsi_methods}):
\begin{equation}\label{eq:drsi}
    \text{DRSI}(\bm X;Y) := (1-\bar{v})\, I(\bm X;Y).
\end{equation}
Assuming non-negative PID atoms, a positive DRSI (i.e., $\bar{v}<1$) indicates the presence of robust information, while a negative DRSI (i.e., $\bar{v}>1$) signals the presence of vulnerable information that critically depends on more than one specific source. Robust information is guaranteed to account for more than half of the total when $\bar{v}<1/2$ (Appendix~\ref{app:bounds}). Just as the RSI emphasises redundancy by counting higher-order redundancies more heavily, the DRSI emphasises vulnerability by counting higher-order vulnerabilities more heavily. The formal duality of these measures is established in Appendix \ref{app:duality}.

Both the RSI and DRSI --- as well as $\bar{r}$ and $\bar{v}$ themselves --- share two properties that distinguish them from existing approaches to quantifying redundancy and synergy: they are PID-independent (their values do not depend on which PID definition is adopted) and their computational complexity scales linearly with the number of sources, making them applicable to systems with tens or hundreds of variables where a full PID is intractable. For completeness, Table~\ref{tab:comparison} provides a detailed comparison between the measures presented here and existing approaches in the literature.

\subsection{Redundancy and vulnerability in error-correcting codes}
\label{sec:hamming_results}

\begin{figure*}[!ht]
    \centering
    \includegraphics[width=0.8\textwidth]{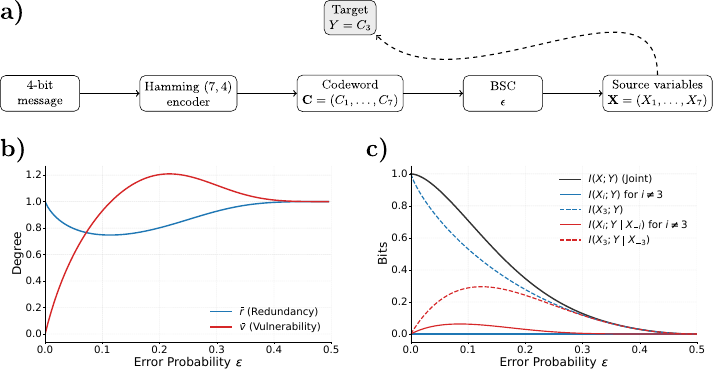}
    \caption{\textbf{Information decomposition of error correcting coding.}
\textbf{a)} Schematic of a Hamming (7,4) code transmitted over a noisy channel. A message of four independent random bits is encoded using a Hamming (7,4) code into a seven-bit codeword $\bm C = (C_1,\dots,C_7)$, transmitted through a binary symmetric channel (BSC) with bit-flip probability $\epsilon$, and received as $\bm X = (X_1,\dots,X_7)$. We analyse how information about the first message bit $Y = C_3$ is distributed across the received bits.
\textbf{b)} Evolution of the average degree of redundancy $\bar r$ (blue) and average degree of vulnerability $\bar v$ (red) as functions of the channel error probability $\epsilon$. The measures exhibit diverse behaviours across different noise regimes.
\textbf{c)} Mutual information terms underlying the Shannon observables: joint mutual information $I(\bm X; Y)$ (black), marginal mutual informations $I(X_i; Y)$ for $i \neq 3$ (solid blue) and $I(X_3; Y)$ (dashed blue), and conditional mutual informations $I(X_i; Y | \bm X_{-i})$ for $i \neq 3$ (solid red) and $I(X_3; Y | \bm X_{-3})$ (dashed red).}
    \label{fig:hamming}
\end{figure*}

In order to deepen our understanding of these quantities, we now apply $\bar{r}$ and $\bar{v}$ to the study of error-correcting codes. The analytical approach we develop is general and applies to any systematic binary block code transmitted over a binary symmetric channel (see Appendix~\ref{app:hamming}). Here we focus on the Hamming~(7,4) code as a concrete instance, which is simple enough that its structure is fully transparent yet rich enough to reveal a surprisingly intricate information structure.

The setup is as follows (Fig.~\ref{fig:hamming}a). A message of four random bits is encoded into a seven-bit codeword $\bm C = (C_1,\dots,C_7)$ using a classic Hamming~(7,4) error-correcting code~\cite{lin2001error}. Of the seven codeword positions, four carry the message bits (conventionally placed in positions 3,5,6,7) and three carry parity-check bits (placed in positions 1,2,4), creating interdependencies that enable single-error correction. The codeword is transmitted through a binary symmetric channel, where each bit independently flips with probability $\epsilon$, producing a received vector $\bm X = (X_1,\dots,X_7)$. Taking the perspective of the receiver, we treat the received bits as sources and one of the message bits, $Y = C_3$, as the target, and calculate $\bar{r}$ and $\bar{v}$ analytically as a function of $\epsilon$.

The results are shown in Figure~\ref{fig:hamming}b, with the underlying mutual information terms in Figure~\ref{fig:hamming}c. At zero noise ($\epsilon=0$) there are no bit flips during transmission, as the target can be read directly from $X_3$ while no other received bit carries information about $Y$ on its own. Hence, only $X_3$ contributes to $\bar{r}$, giving $\bar{r}=1$. At the same time, the parity structure means that any individual source can be reconstructed from combinations of the remaining sources. This implies that no single source is critical, giving $\bar{v}=0$.

As $\epsilon$ increases, $\bar{r}$ follows a U-shape. To understand this result, note that no source other than $X_3$ carries information about $Y$ on its own at any noise level, so $\bar{r}$ is simply the ratio of the information provided by $X_3$ alone to the total information provided by all sources jointly. As noise is first introduced, $X_3$ becomes unreliable, but the parity structure allows errors to be partially corrected using combinations of sources. This implies that, thanks to error-correction, information about $Y$ is not often obtained from individual sources, making $\bar{r}$ drop below $1$. As noise increases further, however, the error correction itself starts to break down: the remaining sources are connected to $Y$ only via the parity structure, which is affected by multiple independent noise terms. Consequently, their contribution degrades faster than the direct channel through $X_3$, making $\bar{r}$ grow back toward $1$. At high noise ($\epsilon \to 0.5$), $X_3$ is the only effective informative source and $\bar{r} \approx 1$.

In contrast, $\bar{v}$ follows an inverted U-shape reflecting a complementary process. At zero noise, any source could be perfectly reconstructed from the others via parity relations, so no source was critical and $\bar{v} = 0$. Noise makes this reconstruction imperfect: each source now carries some information about $Y$ that cannot be recovered from the remaining sources alone, which makes $\bar{v}$ rise. As noise increases further, however, the parity structure breaks down: with multiple corrupted bits, the parity checks become unreliable and the other sources contribute less and less beyond what $X_3$ alone provides. At high values of noise ($\epsilon\approx 0.5$) $\bar{v}\approx 1$, which is the value of a purely unique encoding.

Notably, there is a window in which $\bar{r}$ and $\bar{v}$ increase simultaneously: $\bar{r}$ has already turned around while $\bar{v}$ has not yet peaked. As discussed above, the two measures reflect distinct properties of the underlying information decomposition and need not move in opposite directions. This regime provides a concrete example of that fact, and one that would be hard to predict from the structure of the code alone.

This analysis demonstrates that $\bar{r}$ and $\bar{v}$ are not simply complements of each other: they can move independently, in opposite directions, or even in the same direction, depending on the underlying information structure. It also illustrates the practical value of Shannon observables: they are capable of revealing highly non-trivial features of a system's encoding strategy from quantities that are straightforward to compute.

\subsection{Redundancy and vulnerability in deep learning systems}
\label{sec:results_deep_learning}

Having established that Shannon observables capture
meaningful and non-trivial features of information structure even in
small systems, we now turn to larger-scale applications. In this section, we apply $\bar{r}$ and $\bar{v}$ to deep learning systems,
where each layer's neurons serve as the source variables and the task-relevant
variable (classification label or input image) serves as the target. To make
information-theoretic analysis well-defined in this setting, all activation
values are stochastically quantised and information-theoretic quantities are
computed directly from the resulting probability mass functions on the training
set (see Methods \ref{sec:methods_deep_learning} for details).

\begin{figure*}
    \centering
    \raisebox{15pt}{\includegraphics[scale=.607, clip, trim=35 145 205 8]{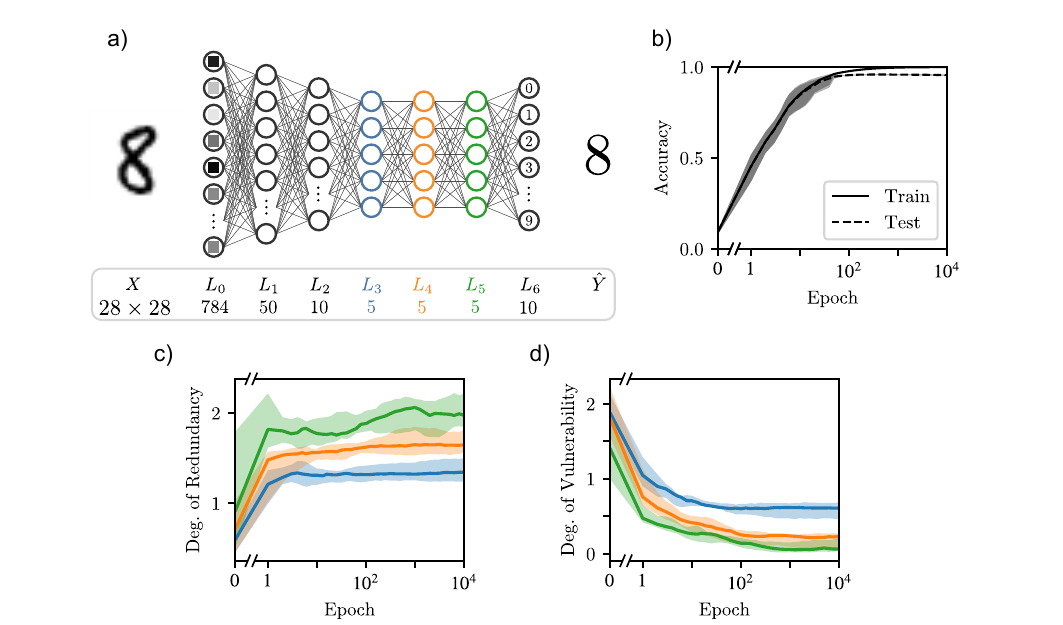}}%
    \llap{\raisebox{86pt}{\textcolor{white}{\rule{24pt}{18pt}}}\hspace{137pt}}%
    \llap{\raisebox{104pt}{\fontfamily{phv}\fontsize{8}{8}\selectfont a)}\hspace{152pt}}\hfill
    \includegraphics[scale=.777]{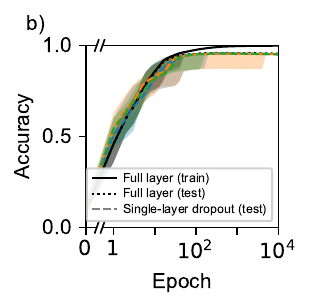}\hfill
    \includegraphics[scale=.777]{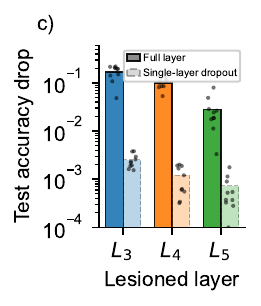}\hfill
    \includegraphics[scale=.777]{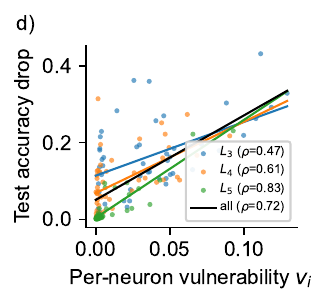}\\[6pt]
    \includegraphics[scale=.777]{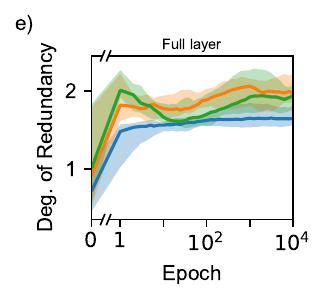}\hfill
    \includegraphics[scale=.777]{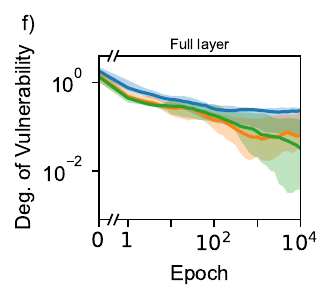}\hfill
    \includegraphics[scale=.777]{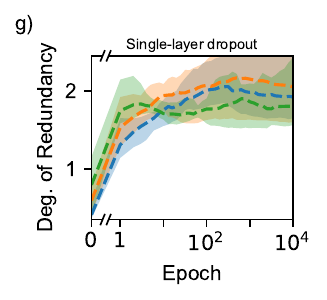}\hfill
    \includegraphics[scale=.777]{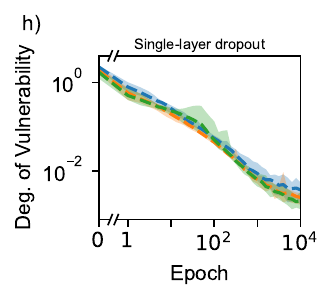}
    \caption{\textbf{Average degree of redundancy and vulnerability in a deep classification network.} These analyses showcase how Shannon observables can track layer-wise representational structure across training in a PID-independent and computationally tractable way. The degree of redundancy about the label increases from the first equal-sized hidden layer to the deeper ones, while the degree of vulnerability decreases across the hidden layers and over training.  In panels (b) and (e)--(h), lines represent the median, shaded regions the maximum and minimum of $10$ runs with random weight initialisations; dashed lines show networks trained with one-neuron dropout applied to the correspondingly coloured five-neuron layer.  \textbf{(a)}~The architecture of the MNIST classification network comprises five fully-connected hidden layers with activation values quantised to eight levels, three of which have the same size of $5$ neurons.     \textbf{(b)}~Evolution of classifier accuracy over training epochs; the test accuracy of the per-layer dropout networks (dashed, coloured) nearly coincides with that of the baseline networks.     \textbf{(c)}~Mean drop in test-set accuracy caused by lesioning single neurons at the end of training, for the standard networks (solid borders) and for the networks trained with one-neuron dropout on that layer (lighter fill, dashed borders); bars show the mean over neurons and runs, dots the individual run means.     \textbf{(d)}~Per-neuron vulnerability $v_i$ at the end of training (epoch $10^4$) against the drop in test-set accuracy caused by lesioning that neuron in the standard networks. Each dot corresponds to one neuron of the three equal-sized hidden layers in one run; coloured lines show per-layer fits and the black line the pooled least-squares fit as visual guides (Spearman $\rho=0.72$, $p<10^{-24}$).     \textbf{(e),~(f)}~Degree of redundancy and degree of vulnerability computed on the training set for the three equal-sized hidden layers of the standard networks.     \textbf{(g),~(h)}~The same quantities for the networks trained with one-neuron dropout on the respective layer (dashed).}
    \label{fig:mnist}
\end{figure*}

\paragraph*{\textbf{Analysis of a deep classification network.} }
For our first case study, we trained a small fully-connected quantised feed-forward network to solve the classic \emph{MNIST} handwritten digit classification task~\cite{lecun_gradient-based_1998}. Three of the hidden layers were chosen to have the same size so that $\bar r$ and $\bar v$ can be compared directly across these layers and throughout training (see Figure \ref{fig:mnist}.a, cf.~\cite{ehrlich2023a}). The networks have been trained for $10^4$ epochs, reaching a test set accuracy of $95.5(3)\,\%$ (see Figure \ref{fig:mnist}.b). Although a full PID is still computationally feasible for the five-source hidden layers analysed here~\cite{jansma2025fast}, repeating it across multiple checkpoints and runs to resolve training dynamics would be costly; the linear scaling of $\bar r$ and $\bar v$ makes this much more computationally lightweight. Moreover, because both quantities are Shannon observables, the resulting curves hold for any choice of PID measure.

Our analysis shows that the degree of redundancy of the neurons' activations with respect to the true classification label increases sharply during the first epoch of training in all three equal-sized layers. Afterwards, the first of these layers ($L_3$) stabilises at a lower level, while the two deeper layers transiently decrease and then increase again during later stages of training. At the end of training, $L_3$ carries a smaller degree of redundancy than each of the two deeper layers in every single run. Between the two deeper layers themselves, however, the ordering is not consistent across runs: the increase of redundancy with depth is carried by the step into the deeper equal-sized layers rather than being strictly monotonic.

\begin{figure*}
    \centering
    \includegraphics{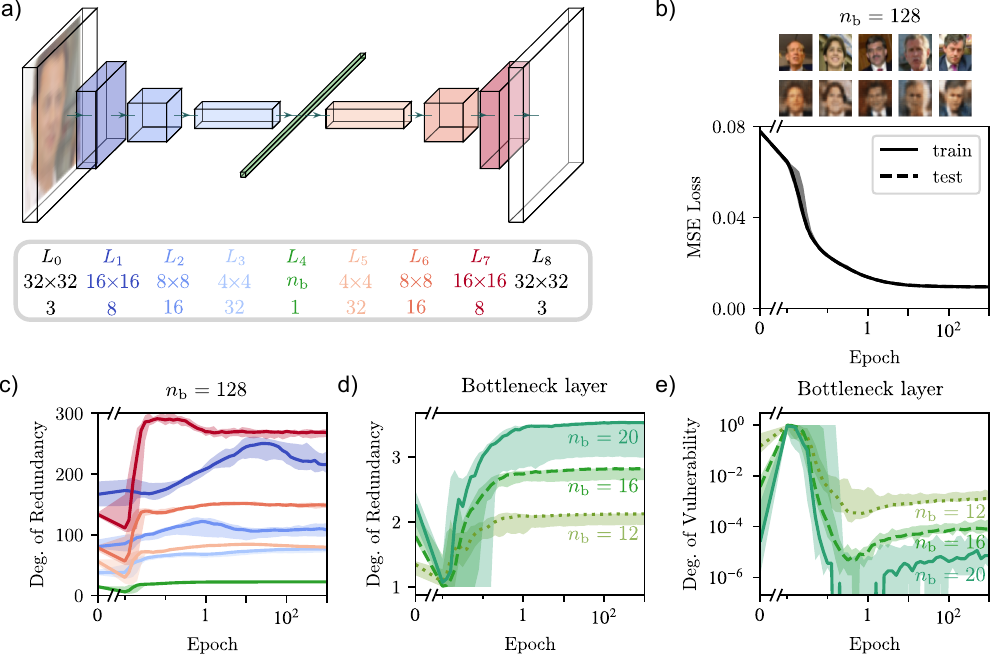}
    \caption{\textbf{Average degree of redundancy and vulnerability in a convolutional autoencoder.} The analyses illustrate the applicability of Shannon observables with up to 2048 source variables, which would be intractable with standard information decomposition approaches.  The degree of redundancy about the input is larger in decoder layers than in the size-matched encoder layers. Across varying bottleneck sizes, redundancy increases with bottleneck width while vulnerability decreases.   In all plots, lines represent the median, shaded regions the maximum and minimum of $10$ runs with random weight initialisations.    \textbf{(a)}~The architecture comprises a three-layer convolutional encoder, a fully-connected bottleneck layer with a varying number of neurons $n_\mathrm{b}$, and a three-layer convolutional decoder. Numbers below the layers reflect the size of the activation matrices as well as the number of channels.     \textbf{(b)}~Performance against training epochs for a bottleneck size of $n_\mathrm{b}=128$. Five original images (upper row) and their reconstruction (lower row) are shown for the converged networks.     \textbf{(c)}~Degree of redundancy computed on the training set; all activations from all convolutional filters have been treated as individual source variables. The corresponding degree of vulnerability (not shown) is equal to zero up to numerical error for all layers.     \textbf{(d),~(e)}~Degree of redundancy and degree of vulnerability of the bottleneck layer for varying small bottleneck sizes $n_\mathrm{b}$.}
    \label{fig:faces}
\end{figure*}

Likewise, the degree of vulnerability decreases over training in every layer and every run, and is higher for the first size-matched layer.
Please bear in mind that, as discussed in Section~\ref{sec:hamming_results}, there is no direct reason for the degree of vulnerability to show a trend opposite to that of the degree of redundancy, and hence this non-trivial signature should be regarded as a signature of how feed-forward deep learning systems process information.

The degree of vulnerability can moreover be reduced selectively by using dropout during training. We trained additional networks in which one of the three equal-sized hidden layers receives one-neuron dropout: in every forward pass, one randomly chosen neuron of that layer is silenced (see Methods). This intervention reduces the degree of vulnerability of the layer it acts on by roughly two orders of magnitude (compare Figures~\ref{fig:mnist}f and~\ref{fig:mnist}h), while leaving the classification accuracy of the network (Figure~\ref{fig:mnist}b) essentially unchanged. The degree of redundancy, on the other hand, is only increased moderately for the first analysed layer at the end of training (compare Figures~\ref{fig:mnist}e and~\ref{fig:mnist}g).

The degree of vulnerability also carries operational meaning for the intact networks: it predicts how much damage the loss of a single neuron causes. We lesioned each neuron of the three equal-sized hidden layers of the fully trained baseline networks individually, clamping its activation to zero, and recorded the resulting drop in test accuracy (see Methods). The per-neuron contribution $v_i$ to the degree of vulnerability correlates positively with the loss of test accuracy within every layer (Spearman $\rho = +0.47/{+0.61}/{+0.83}$ for $L_3/L_4/L_5$) as well as pooled across layers ($\rho=0.72$, $p<10^{-24}$; Figure~\ref{fig:mnist}d). This observation is backed by theory: Through Fano's inequality, a decoder-independent lower bound on the accuracy loss caused by ablating neuron $i$ can be derived to be monotonic in $v_i$, while a lower bound on the expected accuracy loss by ablating a random neuron can be found that is monotonic in $\bar{v}$ (Appendix~\ref{app:ablation_bounds}). The intuitive reason why $\bar{v}$ provides only a lower bound is that while ablating a neuron with $v_i = 0$ does not produce irrecoverable information loss for an optimal decoder, the frozen weights of a trained network might not be set up to utilise the remaining redundant encoding in other neurons to compensate for the lesion, leading to a drop in accuracy. Standard importance heuristics --- weight norms, activation statistics, and first-order Taylor saliency --- largely fail to predict the lesion damage consistently on the same data, with largely non-significant rank correlations (Appendix~\ref{app:importance_comparison}). Consistent with the near-zero vulnerability of the dropout-trained layers, lesioning a neuron of these layers reduces test accuracy by less than $0.3\,\%$ on average, compared to $2.8$--$16.9\,\%$ in the standard networks (Figure~\ref{fig:mnist}c).

\paragraph*{\textbf{Analysis of a convolutional autoencoder.} }
To further showcase the versatility of our framework, we used the degree of redundancy and vulnerability to investigate how a quantised convolutional autoencoder processes information.  For this purpose, we employed a standard convolutional autoencoder architecture with three convolutional encoding layers, a fully connected bottleneck layer of varying size $n_\mathrm{b}$, and three convolutional decoder layers (see Figure \ref{fig:faces}.a). For each layer, we treated all individual activations as source variables and the input image as the target. This yields up to 2048 sources in the largest convolutional layers (L1 and L7), a scale at which a full PID is computationally intractable. We trained the system to reconstruct face images from the \emph{Labeled Faces in the Wild} dataset~\cite{LFWTech}. For $n_\mathrm{b}=128$, the loss converges after $10^3$ training epochs and achieves a faithful reconstruction of the original figures (see Figure \ref{fig:faces}.b).

For a bottleneck of $n_\mathrm{b}=128$ units, the degree of redundancy in all layers except the first encoder layer increases rapidly within the first training epoch and remains mostly constant afterwards (Figure~\ref{fig:faces}c). After the first epoch, decoder layers exhibit higher redundancy than size-matched encoder layers. The degree of vulnerability is approximately zero in all layers throughout training, up to numerical error.

For smaller bottlenecks of $n_\mathrm{b}=12$ to $n_\mathrm{b}=20$ units, the degree of redundancy converges more slowly than for $n_\mathrm{b}=128$, with smaller bottlenecks reaching lower values (Figure~\ref{fig:faces}d). Although vulnerability is low at the end of training, its median values span roughly two orders of magnitude across bottleneck sizes, with larger bottlenecks exhibiting lower vulnerability (Figure~\ref{fig:faces}e).

Early in training, the smaller bottlenecks also exhibit a brief transient in which vulnerability rises to approximately one before falling to much lower values, accompanied by a drop in redundancy. During this period, the total mutual information falls below one bit in some runs. One possible explanation is that the remaining information is carried by individual neurons whose activations deviate from a constant value for only a few samples, producing an encoding with degrees of redundancy and vulnerability close to one.

\subsection{Redundancy and vulnerability in biological neural systems}
\label{sec:meg}

\begin{figure*}[t]
    \centering
    \includegraphics[width=0.9\textwidth]{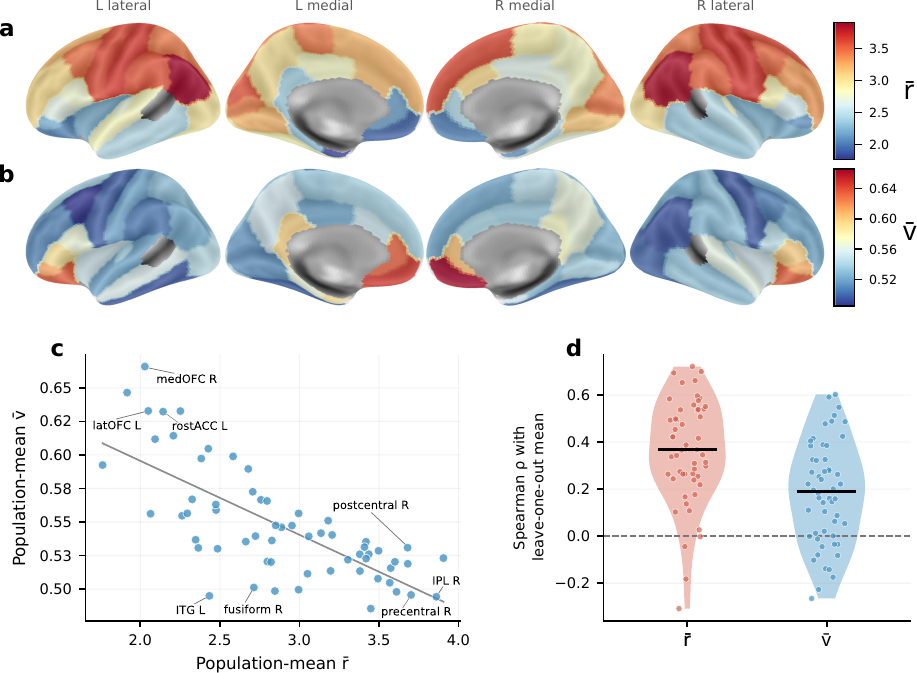}
    \caption{\textbf{Population-level cortical maps of the Shannon observables from resting-state MEG in healthy adults.} Shannon observables can be extended to study directed information flow, remaining tractable for whole-cortex analyses with 61 source regions per target. The analysis reveals pronounced regional structure in the redundancy and vulnerability of incoming information transfer (as approximated by linear-Gaussian state-space Granger causality) across the cortex.
    \textbf{(a)}~Sample-mean degree of redundancy $\bar{r}(i)$ for each of the 62 DKT atlas regions (four views per row: left lateral, left medial, right medial, right lateral).
    \textbf{(b)}~Corresponding sample-mean degree of vulnerability $\bar{v}(i)$.
    \textbf{(c)}~Scatter of the group-mean $\bar{r}$ versus $\bar{v}$ across the 62 cortical regions (Pearson $\rho = -0.73$, $p < 10^{-10}$); each dot is one region and the line shows the least-squares fit.
    \textbf{(d)}~Inter-subject consistency. For each subject, the Spearman correlation between that subject's 62-region profile and the leave-one-out group mean is shown separately for $\bar{r}$ (red) and $\bar{v}$ (blue). Violins show the cross-subject distribution ($n=54$); horizontal bars indicate the median; individual subjects are overlaid as dots. The dashed line marks zero.}
    \label{fig:meg_atlas}
\end{figure*}

The preceding analyses examined how information is encoded and distributed across feed-forward artificial neural networks. We now apply Shannon observables to
recurrent biological neural systems, where the question shifts from how a network
\emph{encodes} information to how information \emph{flows} between brain
regions over time.

To quantify directed information flow between brain regions, we employ Granger causality using linear state-space models \cite{barnett2015granger}, which provides a linear-Gaussian approximation to transfer entropy~\cite{schreiber2000measuring, barnett2009granger}. Transfer entropy measures how much the past of a source process helps predict the future of a target process beyond the target's own history~\cite{wibral2014transfer, bossomaier2016transfer}. Because it is itself a conditional mutual information, our framework can be applied to it directly: $\bar{r}$ and $\bar{v}$ then characterise how many individual source regions can redundantly provide directed information to a target region and how many source regions the transferred information critically depends on, respectively. See Methods \ref{sec:methods_information_flow} for formal definitions and estimation details.

We illustrate this extension with an analysis of resting-state MEG recordings from 54 healthy adults, acquired as part of the publicly available NIMH Healthy Research Volunteer Dataset~\cite{nugent2022nimh}. We parcellated the cortex into 62 regions and extracted a source-reconstructed time series for each region (see Methods, Section~\ref{sec:methods_meg} for details). For each region, we treated the remaining 61 regions as sources and that region as the target, so that $\bar{r}$ and $\bar{v}$ characterise the redundancy and vulnerability of the directed information flow that region receives from the rest of the network. With 61 source regions per target, this cortical-scale characterisation is computationally out of reach for a full PID, and is made tractable here by the linear scaling of $\bar r$ and $\bar v$.

Figure~\ref{fig:meg_atlas} shows the sample-mean cortical maps of $\bar{r}$
and $\bar{v}$ across the 54 subjects. Under the linear-Gaussian
approximation provided by state-space Granger causality,
$\bar{r}$ is substantially greater than one for all cortical regions (mean
across regions: $2.88 \pm 0.55$; range $1.76$--$3.91$), while $\bar{v}$ is
below one everywhere (mean: $0.55 \pm 0.04$; range $0.49$--$0.67$). This implies that information transfer in the brain is generally highly redundant and low in vulnerability. That said, the spatial distributions of $\bar{r}$ and $\bar{v}$ are far from uniform
(Fig.~\ref{fig:meg_atlas}a,b). Regions with the highest $\bar{r}$ ---
indicating the most redundant incoming information --- tend to cluster in
sensorimotor and parietal cortices, while orbitofrontal and cingulate regions
exhibit the highest $\bar{v}$.

Across the 62 cortical regions, $\bar{r}$ and $\bar{v}$ are strongly
negatively correlated (Pearson $\rho = -0.73$, $p < 10^{-10}$;
Fig.~\ref{fig:meg_atlas}c): regions receiving the most redundant information
also exhibit the lowest vulnerability. Although this anti-correlation may seem
intuitively obvious, it is not a mathematical necessity (see Section~\ref{sec:hamming_results}).
Thus, the consistent negative relationship observed here reflects an empirical regularity of resting-state neural information transfer, which suggests that the factors that drive individual differences in the redundancy of information transfer across brain regions also push vulnerability in the opposite direction.

The cortical maps in Fig.~\ref{fig:meg_atlas}a,b are population-level
averages that come with substantial inter-individual variability
(Fig.~\ref{fig:meg_atlas}d). To quantify how closely each subject's regional
profile matches the group pattern, we computed for each subject the Spearman
correlation between their 62-region vector and the leave-one-out group mean.
For $\bar{r}$, correlations ranged from $-0.31$ to $0.72$ (median $0.37$),
with 93\% of subjects showing positive correlations. For $\bar{v}$, the range
was $-0.27$ to $0.60$ (median $0.19$), with 76\% positive.
The regional
ordering of $\bar{r}$ is thus more consistent across individuals
than that of $\bar{v}$, and there is substantial spread in both distributions,
with some subjects diverging markedly from the group. This indicates that these
maps should be read as population-level summaries rather than universal
individual blueprints.

\section{Discussion}

This paper introduces Shannon observables as a way to extract interpretable and scalable summaries of multivariate information structure without resolving the full partial information decomposition. Rather than attempting to estimate all information atoms in large systems, our approach identifies aggregate properties that (i) remain meaningful at scale, (ii) are independent of the particular PID definition adopted, and (iii) can be computed directly from Shannon quantities. In this way, Shannon observables circumvent the operationalisation and scalability issues that hinder information decomposition, opening the door to a wide range of theoretical developments and empirical analyses in systems where full decomposition is computationally intractable.

\paragraph*{\textbf{Theoretical implications.} }
The formalism reveals that redundancy and synergy can be conceptualised in two complementary ways. First, at the source level, redundancy refers to information that is redundantly accessible from more than one individual source, while synergy describes information that cannot be obtained from any single source. Second, redundancy and synergy can be framed in terms of robustness and multi-source vulnerability: robustness represents information that remains accessible even if any individual source is lost, implying that it is redundantly contained within the system though not necessarily in multiple individual sources. Multi-source vulnerability, in contrast, captures information that critically depends on at least two specific sources, thus making it synergistic but in a stronger sense than source-level synergy. These different conceptualisations highlight the richness of redundancy and synergy in multivariate systems.

Building on these ideas, we introduced quantities that jointly position a system in a redundancy–vulnerability plane, reflecting these key features of multivariate information processing. The average degree of redundancy captures the average number of sources redundantly providing information, while the average degree of vulnerability represents the average number of sources on which information critically depends. Crucially, both measures are Shannon observables, relying solely on the Shannon entropies of the system. This eliminates the need to compute individual information atoms or to adopt a specific definition for them. Furthermore, the computational complexity is efficient, as the number of entropies required scales linearly with the number of source variables.

The presented approach reveals new insights into the RSI, a well-known information-theoretic metric whose interpretation has remained elusive, and introduces the novel DRSI.
We demonstrated that the RSI and DRSI quantify two complementary balances within multivariate information structure, but they do so asymmetrically and in different senses. The RSI is governed by the average degree of redundancy and hence emphasises higher-order source-level redundancies relative to source-level synergy. By contrast, the DRSI is governed by the average degree of vulnerability and therefore emphasises higher-order multi-source vulnerabilities relative to robustness. This distinction is invisible in the two-source case, but becomes essential for interpreting systems with three or more sources.

\paragraph*{\textbf{Scope and limitations.}}

These advantages come with an important tradeoff: Shannon observables are not substitutes for the full PID, but coarse-grained summaries of it. In particular, they do not identify which specific information atoms dominate a given system, nor do they recover the detailed balance among redundant, unique, and synergistic components.

A simple mixed system illustrates this loss of resolution. Let $X_1$, $X_2$, and $Z$ be independent uniformly distributed binary variables, define $X_3=X_4=Z$, and consider the two-bit target
\[
Y=(X_1\oplus X_2,Z).
\]
The first target component is encoded synergistically by $X_1$ and $X_2$, whereas the second is redundantly available from both $X_3$ and $X_4$. The complete system carries two bits of information about $Y$, and a direct calculation yields $(\bar r,\bar v)=(1,1)$, the same global values as a purely unique encoding. Yet, when the two source pairs are considered separately, $\{X_1,X_2\}$ yields $(\bar r,\bar v)=(0,2)$, while $\{X_3,X_4\}$ yields $(\bar r,\bar v)=(2,0)$. The global observables therefore conceal the coexistence of qualitatively opposed local information structures.

When such heterogeneity is scientifically relevant, $\bar r$ and $\bar v$ can be evaluated on theoretically motivated subsets of sources to expose local structure. Exhaustively analysing all possible subsets, however, would sacrifice the linear scaling of the global analysis (although it would still be considerably more tractable than resolving the full PID). Accordingly, the strength of the present framework lies not in resolving the full microscopic structure of multivariate information, but in identifying theoretically meaningful and empirically tractable global observables.

\paragraph*{\textbf{Coding Theory.} }

We demonstrated the practical capabilities of the Shannon observables through a range of applications spanning coding theory, deep learning, and neuroscience. Regarding coding theory, the Hamming~(7,4) analysis shows that even in a small system whose structure is fully known $\bar{r}$ and $\bar{v}$ exhibit non-trivial joint behaviour as the noise level is changed; behaviour that is directly interpretable in terms of the code's parity structure. The analytical derivation developed in the appendix applies to any systematic binary block code transmitted over a binary symmetric channel, providing a route for future work to characterise broader code families --- repetition codes, single-parity-check codes, longer Hamming codes, and algebraic families such as BCH codes --- through their $(\bar{r}, \bar{v})$ trajectories as a function of channel noise.

\paragraph*{\textbf{Deep Learning Systems.} }
In the MNIST classifier, redundancy generally increases and vulnerability decreases over training. At the end of training, both deeper size-matched layers exhibit higher redundancy and lower vulnerability than the first. Together, these trends indicate that information about digit identity becomes more redundantly represented across neurons and less dependent on particular neurons for its preservation.
This perspective complements information-bottleneck accounts of deep learning~\cite{goldfeld_estimating_2019,tishby2015deep,saxe_information_2019}: whereas the information-bottleneck framework tracks how much information about input and output is retained across layers, Shannon observables characterise how that information is distributed across neurons.

The closest point of contact is the work of Ehrlich et al. \cite{ehrlich2023a}, who studied the same kind of feed-forward network and asked how many neurons must be read jointly to recover a typical piece of label information. They found that this number, which they call representational complexity, falls over training and across layers toward one, so that the label becomes accessible from single neurons. Our results describe the same movement toward locality, but they separate two situations that representational complexity leaves indistinguishable. Information accessible from a single neuron may sit in only one specific neuron, so that removing it destroys access, or it may be present in many neurons at once, so that no single removal matters; both have a representational complexity of one. The increase in $\bar r$ shows that our networks develop the second situation: the label becomes available not from one neuron but from a growing number of them individually. The decrease in $\bar v$ shows the accompanying gain in robustness, as access comes to depend critically on fewer and fewer particular units. What forms over learning is thus not simply a localised code but a redundant and robust one.

A related distinction arises when interpreting neuronal specialisation. In restricted Boltzmann machines, Tax et al.~\cite{tax2017partial} computed the $I_{min}$ PID \cite{williams2010nonnegative} for randomly sampled neuron pairs and found increasing unique information over training, which they interpreted as increasing specialisation. This captures the extent to which two sampled neurons carry distinct information about the label. Population-level organisation can nevertheless take different forms while producing the same pairwise pattern. For instance, distinct neuronal subgroups may encode different aspects of the target, with each aspect redundantly represented within its corresponding subgroup. Pairs drawn from different subgroups would then exhibit substantial unique information, while the representation as a whole would combine high $\bar r$ with $\bar v$ near zero: information is redundantly present in multiple neurons within each group and not critically dependent on any single neuron. By contrast, a representation in which information is \textit{globally} unique to individual neurons would instead place both $\bar r$ and $\bar v$ near one. Our results are consistent with the former organisation, combining specialisation with redundant support. Pairwise PID and the population-level observables thus characterise complementary dimensions of representational organisation.

Our lesion and dropout experiments establish an operational role for vulnerability: per-neuron vulnerability predicts the loss of classification accuracy after neuron removal, while dropout reduces vulnerability and improves resistance to subsequent lesions. Proca et al.~\cite{proca2024synergistic} linked lesion damage to neurons' average pairwise synergy. However, even high average pairwise synergy can coexist with information being fully preserved elsewhere in the layer. Vulnerability captures whether information is actually lost from the remaining population, allowing us to derive theoretical bounds on classification loss and error after ablation.

For the convolutional autoencoder, larger bottlenecks exhibited higher redundancy and lower vulnerability. This pattern is consistent with the intuition that greater capacity leaves more room for redundant representations, while tighter capacity constraints may make individual neurons more critical as opportunities for redundant support diminish. The same intuition may help explain why decoder layers exhibit higher redundancy than size-matched encoder layers. Here, the potential difference lies in how much input information is retained: some information may be lost in passing through the bottleneck. The decoder layers would then represent less input information with the same number of neurons, leaving greater scope for redundant encoding.

Future work could examine how compression and reconstruction requirements shape redundancy and vulnerability, drawing on rate-distortion theory and work on disentanglement and $\beta$-VAE-style objectives~\cite{tapia2020information,burgess2018understanding}. More generally, Shannon observables characterise information organisation but do not by themselves identify the mechanisms that produce it. Combining them with circuit analysis~\cite{elhage2021mathematical,ameisen2025circuit} or causal abstractions~\cite{geiger2021causal,geiger2025causal} could help relate the observed information structures to the computations implemented by the networks.

\paragraph*{\textbf{Biological neural networks.} }By extending the framework from mutual information to transfer entropy, approximated here via linear state-space Granger causality under a Gaussian assumption, we showed that Shannon observables can characterise the redundancy and vulnerability of directed information flow in the brain. The resulting cortical atlas revealed pronounced regional structure: sensorimotor areas received highly redundant input, whereas $\bar{v}$, despite being low overall, exhibited a clear cortical gradient, with relatively higher values in association cortex than in sensorimotor regions. This suggests that incoming information in higher-order integrative areas depends more critically on particular source regions, while sensorimotor information flow is organised more robustly. A testable prediction is that perturbing regions with higher $\bar{v}$ should produce more widespread or less easily compensated downstream effects, a hypothesis that could be examined in perturbational, longitudinal, or task-based studies.

These findings are consistent with previous work reporting higher redundancy in brain regions associated with input-output functions and higher synergy in areas associated with the default-mode network and executive functions~\cite{luppi2022synergistic,varley2023information,varley2023partial,luppi2024information}. Our MEG-based analysis complements earlier characterisations of directed information flow in resting-state MEG~\cite{hillebrand2016direction} by decomposing transfer entropy itself into redundancy and vulnerability components, while expanding existent methodologies~\cite{stramaglia2014synergy,stramaglia2024disentangling} by characterising interactions across a much larger number of brain regions than was previously tractable. This is made possible by the linear complexity of Shannon observables.

However, several caveats apply to this analysis. The use of state-space Granger causality provides a linear-Gaussian approximation to transfer entropy~\cite{faes2017multiscale}, which may nonetheless miss non-linear components of directed information flow; task-driven gradients may also differ systematically from the resting-state pattern reported here; and the inter-subject variability evident in Fig.~\ref{fig:meg_atlas} cautions against treating the reported regional gradients as universal across subjects. Future work combining Shannon observables with non-linear transfer entropy estimators, task-based recordings, and longitudinal or lesion datasets could further clarify the functional significance of these and other regional gradients~\cite{shine2016dynamics,deco2021revisiting}.

\paragraph*{\textbf{Final thoughts.} }

Overall, the framework introduced in this paper addresses the key limitations of information decomposition and brings important theoretical and practical advantages, opening the way to new information decomposition analyses when a full decomposition is computationally unfeasible.
Given the broad applicability of information theoretic quantities~\cite{rosas2026faces}, future work may explore applications of analyses based on Shannon observables on a broad range of settings.
For instance, since transfer entropy and Granger causality are established tools in areas such as gene regulatory network inference~\cite{kim2021tenet} and Earth system sciences~\cite{runge2019inferring}, the decomposition of these quantities into redundancy and vulnerability is readily applicable in these domains.
Additionally, future work may investigate the landscape of possible Shannon observables, which may bring new insights into the possible ways of conceptualising synergy and redundancy. Also, future work may develop techniques to improve the estimation of Shannon observables in large-scale systems, using novel algorithmic approaches~\cite{belloli2025thoi}, de-biasing methods \cite{venkatesh2023gaussian}, or null models~\cite{liardi2025null}.

\section{Methods}

\subsection{Partial Information Decomposition}
\label{sec:PID}

Consider $n$ source variables $\bm X = (X_1,\dots,X_n)$ providing information about a target variable $Y$, as quantified by Shannon's mutual information $I(\bm X;Y) = H(Y) - H(Y|\bm X)$~\cite{cover1999elements}. \emph{Partial information decomposition} (PID)~\cite{williams2010nonnegative} addresses the question of how this information is distributed across the components of $\bm X$ by decomposing it into non-overlapping \emph{information atoms}, each characterising a qualitatively distinct mode of encoding.

For $n=2$ sources, the decomposition yields four atoms,  \begin{equation}      I(X_1,X_2;Y) = \Pi_\text{red} + \Pi_\text{unq1} + \Pi_\text{unq2} + \Pi_\text{syn},  \end{equation}  where $\Pi_\text{red}$ is information redundantly available from both sources, $\Pi_\text{unq1}$ and $\Pi_\text{unq2}$ are unique to each source, and $\Pi_\text{syn}$ is synergistic information accessible only from their joint observation.

For general $n$, the number of qualitatively distinct atoms grows super-exponentially and these simple labels no longer suffice. Instead, each atom is identified by specifying which collections of sources grant access to it. Formally, each atom is labelled by an \emph{antichain} $\alpha = \{\mathbf{a}_1,\ldots,\mathbf{a}_m\}$, a collection of subsets of the source index set $[n] = \{1,\ldots,n\}$ in which no element is a superset of another.

The atom $\Pi(\alpha)$ then represents the information about $Y$ that can be obtained \textit{if and only if} one has access to all variables in at least one of the collections $\mathbf{a}_i \in \alpha$. For example, $\Pi(\{1\}\{4,5\}\{2,9\})$ is the information about $Y$ that can be obtained if one has access to $X_1$ alone, \emph{or} to $\{X_4, X_5\}$ jointly, \emph{or} to $\{X_2, X_9\}$ jointly, but not through any other combination of sources. In this notation, the four atoms of the $n=2$ case (Eq.~\ref{eq:pid_n2}) correspond to $\Pi_\text{red} = \Pi(\{1\}\{2\})$, $\Pi_\text{unq1} = \Pi(\{1\})$, $\Pi_\text{unq2} = \Pi(\{2\})$, and $\Pi_\text{syn} = \Pi(\{1,2\})$. The total information decomposes as
\begin{equation}\label{eq:PID}
    I(\bm X;Y) = \sum_{\alpha \in \mathcal{A}_n} \Pi(\alpha),
\end{equation}
where $\mathcal{A}_n$ denotes the set of all antichains of $([n], \subseteq)$ (excluding the empty antichain and the
antichain $\{\emptyset\}$; see Refs.~\cite{gutknecht2021bits, Gutknecht2025babel} for details).

Given the meaning of the atoms described above, it is natural
to require that the mutual information provided by any collection
of sources $\mathbf{a} \subseteq [n]$ equals the sum of all atoms
accessible from $\mathbf{a}$, that is, all atoms whose antichain
contains some subset of $\mathbf{a}$:
\begin{equation} \label{eq:consistency_equation}
I(\mathbf{a};Y) = \sum_{\exists\, \mathbf{b} \subseteq \mathbf{a}:\, \mathbf{b} \in \alpha} \Pi(\alpha).
\end{equation}
This is the \textit{consistency equation} of PID, which is also visually apparent in a mereological diagram (Fig.~\ref{fig:illustration_r_v}) by observing which atoms (contiguous areas) are contained in which mutual information terms (black ovals). We adopt as a definition of a PID any set of quantities $\Pi(\alpha)$ satisfying Eq.~\eqref{eq:consistency_equation}; the results of this paper depend only on this requirement. By subtracting Eq.~\eqref{eq:consistency_equation} from the joint mutual information,  and applying the chain rule, the consistency equation can equivalently be written in conditional form:
\begin{equation} \label{eq:consistency_equation_cmi}
I(\mathbf{a}^C;Y|\mathbf{a}) = \sum_{\neg\exists\, \mathbf{b} \subseteq \mathbf{a}:\, \mathbf{b} \in \alpha} \Pi(\alpha),
\end{equation}
which sums exactly the atoms that are \emph{not} accessible from $\mathbf{a}$. These are of course the atoms one gains access to through the remaining sources $\mathbf{a}^C$, given that $\mathbf{a}$ is already known; which is why they are required to equal the conditional mutual information on the left-hand side. Both forms of the consistency equation are used in the proofs that follow.

\subsection{Shannon observables: degree of redundancy and degree of vulnerability}
\label{sec:shannon_observables_formal}

Building on the PID framework described above, we now formally define the two quantities whose Shannon-observable character is the central theoretical result of this paper.

\paragraph{Degree of redundancy.}
For any information atom $\Pi(\alpha)$ associated with the antichain $\alpha = \{\mathbf{a}_1,\ldots,\mathbf{a}_m\}$, its \emph{degree of redundancy} is
\begin{equation}\label{eq:def_r}
    r(\alpha) := \big|\big\{i \in [n] \;\big|\; \{i\} \in \alpha\big\}\big|,
\end{equation}
i.e.\ the number of individual sources through which this atom can be accessed. For example, $\Pi(\{1\}\{2\}\{3,4\})$ has $r=2$ because two singletons ($\{1\}$ and $\{2\}$) appear in the antichain.  Figure~\ref{fig:illustration_r_v}a illustrates the grouping of all atoms by their degree of redundancy for $n=3$.

\begin{figure*}[t]
    \centering
    \def\svgwidth{\linewidth}\begingroup%
  \makeatletter%
  \providecommand\color[2][]{%
    \errmessage{(Inkscape) Color is used for the text in Inkscape, but the package 'color.sty' is not loaded}%
    \renewcommand\color[2][]{}%
  }%
  \providecommand\transparent[1]{%
    \errmessage{(Inkscape) Transparency is used (non-zero) for the text in Inkscape, but the package 'transparent.sty' is not loaded}%
    \renewcommand\transparent[1]{}%
  }%
  \providecommand\rotatebox[2]{#2}%
  \newcommand*\fsize{\dimexpr\f@size pt\relax}%
  \newcommand*\lineheight[1]{\fontsize{\fsize}{#1\fsize}\selectfont}%
  \ifx\svgwidth\undefined%
    \setlength{\unitlength}{552.7710009bp}%
    \ifx\svgscale\undefined%
      \relax%
    \else%
      \setlength{\unitlength}{\unitlength * \real{\svgscale}}%
    \fi%
  \else%
    \setlength{\unitlength}{\svgwidth}%
  \fi%
  \global\let\svgwidth\undefined%
  \global\let\svgscale\undefined%
  \makeatother%
  \begin{picture}(1,0.43114096)%
    \lineheight{1}%
    \setlength\tabcolsep{0pt}%
    \put(0,0){\includegraphics[width=\unitlength,page=1]{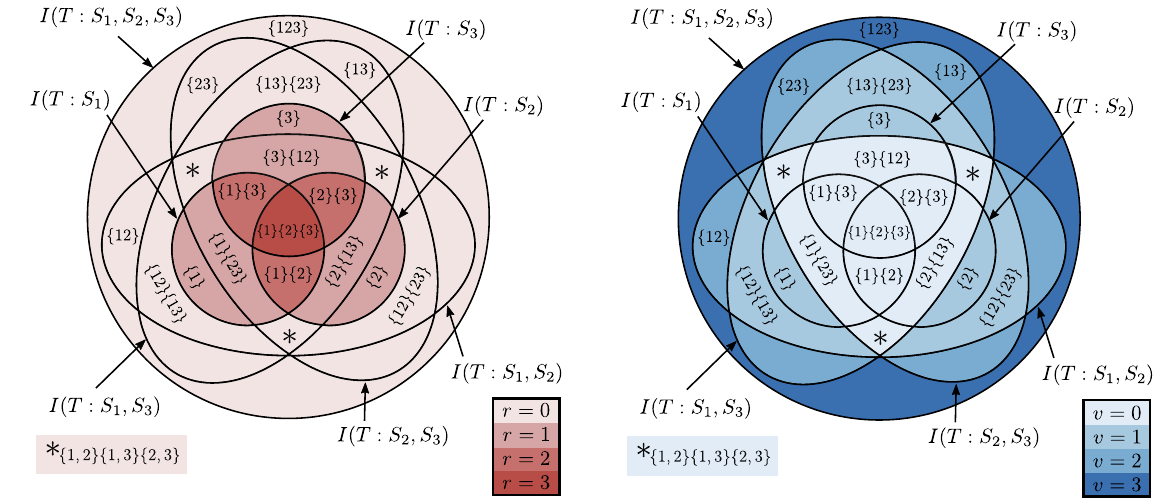}}%
    \put(-0.00071904,0.41665337){\makebox(0,0)[lt]{\lineheight{1.25}\smash{\begin{tabular}[t]{l}a)\end{tabular}}}}%
    \put(0.49016378,0.4164323){\makebox(0,0)[lt]{\lineheight{1.25}\smash{\begin{tabular}[t]{l}b)\end{tabular}}}}%
  \end{picture}%
\endgroup%

    \caption{Mereological representation of the mutual information and PID atoms of $n=3$ sources, with colouring according to a) degree of redundancy, and b) degree of vulnerability.}
    \label{fig:illustration_r_v}
\end{figure*}

The \emph{average degree of redundancy} is the weighted average of $r$ over all atoms:
\begin{equation}\label{eq:def_rbar}
    \bar{r} := \sum_{k=0}^n k \sum_{r(\alpha)=k} \frac{\Pi(\alpha)}{I(\bm X;Y)},
\end{equation}
where the weight of each degree $k$ is the fraction of the total mutual information carried by atoms of that degree. The quantity $\bar{r}$ takes values in $[0,n]$ and is ill-defined if $I(\bm X;Y)=0$.

\begin{proposition}\label{prop:r_bar_shannon_observable}
The average degree of redundancy is a Shannon observable. For any PID satisfying the consistency equations~\eqref{eq:consistency_equation},
\begin{equation}\label{eq:rbar_proof}
    \bar{r} = \frac{\sum_{i=1}^n I(X_i;Y)}{I(\bm X;Y)}.
\end{equation}
\end{proposition}

\begin{proof}
Define $R(\bm X;Y) := \sum_{i=1}^n I(X_i;Y)$. By the consistency equation~\eqref{eq:consistency_equation}, the marginal mutual information $I(X_i;Y)$ equals the sum of all atoms $\Pi(\alpha)$ for which some element of $\alpha$ is a subset of $\{i\}$; for singletons this means precisely the atoms with $\{i\}\in\alpha$. Summing over all $i$, each atom $\Pi(\alpha)$ is counted once for every singleton present in $\alpha$, hence exactly $r(\alpha)$ times:
\begin{equation}
    R(\bm X;Y) = \sum_{k=0}^n k \sum_{r(\alpha)=k} \Pi(\alpha).
\end{equation}
Dividing by $I(\bm X;Y) = \sum_\alpha \Pi(\alpha)$ yields the result.
\end{proof}

Importantly, while the total amount of information at any particular degree $k$ generally depends on the specific PID measure used, the weighted average $\bar{r}$ does not: it is determined entirely by Shannon entropies.

\paragraph{Degree of vulnerability.}
The complementary quantity counts how many individual sources an atom \emph{critically depends on}, in the sense that losing any one of them would render the atom inaccessible:
\begin{equation}\label{eq:def_v}
    v(\alpha) := \big|\{i \in [n] \;\big|\; \forall\, \mathbf{a} \in \alpha :\, i \in \mathbf{a}\}\big|.
\end{equation}
Thus $v(\alpha)$ counts the indices that appear in \emph{every} element of the antichain. An atom with $v(\alpha)=0$ does not critically depend on any individual source and is called \emph{robust}; an atom with $v(\alpha)\geq 2$ critically depends on multiple sources and is hence called \emph{multi-source vulnerable}. Figure~\ref{fig:illustration_r_v}b illustrates the grouping of all atoms by their degree of vulnerability for $n=3$.

The \emph{average degree of vulnerability} is defined analogously:
\begin{equation}\label{eq:def_vbar}
    \bar{v} := \sum_{j=0}^n j \sum_{v(\alpha)=j} \frac{\Pi(\alpha)}{I(\bm X;Y)}.
\end{equation}

\begin{proposition}\label{prop:v_bar_shannon_observable}
The average degree of vulnerability is a Shannon observable. For any PID satisfying the consistency equations,
\begin{equation}\label{eq:vbar_proof}
    \bar{v} = \frac{\sum_{j=1}^n I(X_j;Y|\bm X_{-j})}{I(\bm X;Y)},
\end{equation}
where $\bm X_{-j}$ denotes all sources except $X_j$.
\end{proposition}

\begin{proof}
Define $V(\bm X;Y) := \sum_{j=1}^n I(X_j;Y|\bm X_{-j})$. By the conditional form of the consistency equation~\eqref{eq:consistency_equation_cmi}, the conditional mutual information $I(X_j;Y|\bm X_{-j})$ equals the sum of all atoms that are \emph{not} accessible without $X_j$---that is, every atom for which $j$ appears in all elements of the antichain. Summing over all $j$, each atom $\Pi(\alpha)$ is counted once for every index that is critical for it, hence exactly $v(\alpha)$ times:
\begin{equation}
    V(\bm X;Y) = \sum_{j=0}^n j \sum_{v(\alpha)=j} \Pi(\alpha).
\end{equation}
In particular, atoms with $v(\alpha)=0$ contribute nothing: if an atom does not critically depend on any source, it does not appear in any conditional mutual information term. Dividing by $I(\bm X;Y)$ yields the result.
\end{proof}

\paragraph{Degree-resolved information.}

The proofs above naturally partition the total mutual information $I(\bm X;Y)$ into bins according to the degree of redundancy or vulnerability of each atom. We denote the total information carried by atoms of redundancy degree~$k$ and vulnerability degree~$k$ by
\begin{align}
    I_\text{r}^{(k)}(\bm X;Y) &:= \sum_{r(\alpha)=k} \Pi(\alpha), \label{eq:def_Irk}\\[4pt]
    I_\text{v}^{(k)}(\bm X;Y) &:= \sum_{v(\alpha)=k} \Pi(\alpha), \label{eq:def_Ivk}
\end{align}
respectively. Unlike $\bar{r}$ and $\bar{v}$, these individual terms are generally not Shannon observables: their values depend on the choice of PID. However, they provide the natural vocabulary for stating the PID decomposition of the RSI and DRSI, to which we turn next.

\subsection{Decomposition of the RSI and DRSI}
\label{sec:rsi_drsi_methods}

The RSI and DRSI were introduced in the Results as Shannon observables capturing two complementary perspectives on the balance of redundancy and synergy. Here we show that both quantities admit explicit decompositions in terms of PID atoms, making precise the sense in which each is a ``redundancy minus synergy'' type of quantity and revealing the specific weighting each one applies.

\paragraph{RSI.}
We first show that the RSI (defined in Eq.~\ref{eq:rsi_def}) decomposes into a weighted sum of source-level redundancies, where higher orders of redundancy carry greater weight, minus the total source-level synergy.

\begin{proposition}\label{prop:rsi}
\begin{equation}\label{eq:rsi_k_0_to_n}
    \normalfont{\text{RSI}}(\bm X;Y) = \sum_{k=2}^n (k-1)\, I_{\normalfont{\text{r}}}^{(k)}(\bm X;Y) - I_{\normalfont{\text{r}}}^{(0)}(\bm X;Y).
\end{equation}
\end{proposition}

\begin{proof}
By Proposition~\ref{prop:r_bar_shannon_observable} and the definition of the RSI,
\begin{align}
    \text{RSI}(\bm X;Y)
    &= R(\bm X;Y) - I(\bm X;Y) \nonumber\\
    &= \sum_{k=0}^n (k-1) \sum_{r(\alpha)=k} \Pi(\alpha) \nonumber\\
    &= \sum_{k=2}^n (k-1) \sum_{r(\alpha)=k} \Pi(\alpha) - \sum_{r(\alpha)=0} \Pi(\alpha),
\end{align}
where the second line uses $R(\bm X;Y) = \sum_k k\sum_{r(\alpha)=k}\Pi(\alpha)$ from the proof of Proposition~\ref{prop:r_bar_shannon_observable} and $I(\bm X;Y) = \sum_k \sum_{r(\alpha)=k}\Pi(\alpha)$, and the third line separates the $k=0$ and $k=1$ terms (the latter vanishes since its prefactor is zero).
\end{proof}

Thus the RSI is positive if and only if the weighted redundancies outweigh the source-level synergy, and negative if source-level synergy dominates. The weighting is asymmetric: atoms accessible from $k$ individual sources contribute with weight $k{-}1$, so higher-order redundancies are counted more heavily, while all source-level synergistic atoms ($r=0$) enter with uniform weight. For $n=2$ this asymmetry is invisible since there is only one order of redundancy ($k=2$), recovering the classical result $\text{RSI} = \Pi(\{1\}\{2\}) - \Pi(\{1,2\})$.

\paragraph{DRSI.}
The DRSI was defined in the Results (Eq.~\ref{eq:drsi}) via the average degree of vulnerability as $\text{DRSI}(\bm X;Y) = (1-\bar{v})\,I(\bm X;Y)$. Since $\bar{v}$ is a Shannon observable by Proposition~\ref{prop:v_bar_shannon_observable}, the DRSI is immediately a Shannon observable as well. We now show that it admits a PID decomposition analogous to that of the RSI, with the roles of redundancy and synergy exchanged: the DRSI equals the total robust information minus a weighted sum of vulnerabilities, where higher orders of vulnerability carry greater weight.

\begin{proposition}\label{prop:drsi_decomposition}
\begin{equation}\label{eq:drsi_k_0_to_n}
    \normalfont{\text{DRSI}}(\bm X;Y) = I_{\normalfont{\text{v}}}^{(0)}(\bm X;Y) - \sum_{k=2}^n (k-1)\, I_{\normalfont{\text{v}}}^{(k)}(\bm X;Y).
\end{equation}
\end{proposition}

\begin{proof}
Starting from the definition $\text{DRSI} = (1-\bar{v})\,I(\bm X;Y)$ and expanding $\bar{v}$ using Eq.~\eqref{eq:def_vbar}:
\begin{align}
    \text{DRSI}(\bm X;Y)
    &= I(\bm X;Y) - \bar{v}\, I(\bm X;Y) \nonumber\\
    &= \sum_{k=0}^n I_\text{v}^{(k)} - \sum_{k=0}^n k\, I_\text{v}^{(k)} \nonumber\\
    &= \sum_{k=0}^n (1-k)\, I_\text{v}^{(k)} \nonumber\\
    &= I_\text{v}^{(0)} - \sum_{k=2}^n (k-1)\, I_\text{v}^{(k)},
\end{align}
where the last line uses that the $k=1$ term vanishes.
\end{proof}

The DRSI thus mirrors the RSI with the roles interchanged: where the RSI weights higher-order redundancies more heavily and lumps all synergy into a single term, the DRSI weights higher-order vulnerabilities more heavily and lumps all robustness into a single term. In this sense, the RSI and DRSI are dual to each other, a relationship that is made formally precise through the lattice duality discussed in Appendix~\ref{app:duality}.

\subsection{Information-theoretic analysis of deep neural networks}
\label{sec:methods_deep_learning}
The layers of feed-forward neural networks can be interpreted as successive `channels' that re-encode relevant information to make it extractable for the output layer
while removing task-irrelevant elements from the inputs. Following this line of reasoning, one can apply information decomposition to the network's activation patterns in order to obtain insights on how information about the target --- e.g., the classification label or the input of an autoencoder --- is encoded by the neural units, how this encoding changes among successive hidden layers, and how it develops throughout training.

Despite its promise, early attempts of using information theory in the study of neural networks have revealed potential pitfalls in this approach. Since classical neural networks work with continuous activation values and are deterministic, they define an
injective map from the inputs to their hidden layer activations. Neural networks thus have (at least in theory) an infinite channel capacity, precluding standard information-theoretic analyses~\cite{saxe_information_2019, goldfeld_estimating_2019, geiger_information_2021}. Attempts to measure information-theoretic quantities in such settings by binning values or employing continuous mutual information estimators thus do not yield meaningful information-theoretic indicators, but can at most be reinterpreted as measures for geometric clustering depending on the estimation scheme \cite{geiger_information_2021}. These shortcomings can be remedied by instead analysing settings that break the network's injectivity, for example by adding Gaussian noise to the activations \cite{goldfeld_estimating_2019}, or by quantising the activations while training and evaluating~\cite{ehrlich2023a}.

The above considerations are concerned with the well-definedness of information theoretic quantities in neural network analysis. After resolving these theoretical issues, we still face the practical challenge of estimation from finite data samples. Here, it is crucial to distinguish clearly between training and test contexts. The test dataset is typically viewed as a finite sample drawn from a potentially infinite underlying data-generating process, making it necessary to statistically estimate information-theoretic quantities. In contrast, our approach is to carry out information-theoretic analyses exclusively on the training set, treating it as the complete population with each element assigned equal probability. This allows us to directly \textit{calculate} information-theoretic quantities, bypassing the need for estimation. The scientific question we address with this approach shifts from how the network encodes an unknown, out-of-sample distribution to how it learns to encode this particular training dataset.

The deep learning analyses presented in Results \ref{sec:results_deep_learning} consider two architectures: a feed-forward image classifier and a convolutional autoencoder. In both cases, activation values are stochastically discretised to eight levels, and the full probability mass function over discretised activations is constructed from the training set, from which all information-theoretic quantities are computed directly. The neural networks in this work have been implemented using the ``nninfo'' Python package for information-theoretic analysis of deep neural networks (\url{https://github.com/Priesemann-Group/nninfo}, \cite{ehrlich2023a}) which builds on the PyTorch framework.

To achieve finite and non-trivial information dynamics, the activations of all neurons have been quantised stochastically to one of the two nearest of eight quantisation levels between the bounds of the activation function $\sigma_\mathrm{min}$ and $\sigma_\mathrm{max}$ with equal distances $\epsilon = (\sigma_\mathrm{max}-\sigma_\mathrm{min})/(n_\mathrm{levels} - 1)$. The rounded activation values $\ell$ are obtained from the raw activation values $\hat{\ell}$ as
\begin{equation*}
    \ell = \epsilon\lambda + \sigma_\mathrm{min}\quad
\end{equation*}
where
\begin{equation*}
\quad \lambda =
    \begin{cases}
        \left\lceil \frac{\hat{\ell} - \sigma_\mathrm{min}}{\epsilon} \right\rceil & r \leq \left(\frac{\hat{\ell} - \sigma_\mathrm{min}}{\epsilon} \mod 1\right)\\[0.5cm]
      \left\lfloor \frac{\hat{\ell} - \sigma_\mathrm{min}}{\epsilon} \right\rfloor & r > \left(\frac{\hat{\ell} - \sigma_\mathrm{min}}{\epsilon} \mod 1 \right),
   \end{cases}
\end{equation*}
Here, $r \in [0, 1]$ is drawn i.i.d.~from a uniform distribution for each neuron and each evaluation.

The choice of discretisation level must strike a balance between two competing requirements. On the one hand, the binning must be fine enough not to conflate activation differences that the network needs to distinguish in later layers, as this would negatively affect task performance. On the other hand, the discretisation must produce a sufficiently strong information bottleneck so that the geometric clustering of similar input features becomes detectable by information-theoretic means. For the MNIST classifier, Appendix~\ref{app:quantization_robustness} tests sensitivity to the number of quantisation levels by repeating training and analysis with six and ten levels in addition to the eight-level setting of the main text. The main qualitative trends in redundancy and vulnerability are preserved across these three settings.

The \emph{MNIST classifier network} is a fully-connected feed-forward deep neural network using a one-hot output representation and tanh activation functions. The networks were trained $10$ times for $10^4$ epochs with different Xavier random weight initialisations, using cross entropy loss and a constant SGD learning rate of $0.01$. They are trained on the standard MNIST training set of $60{,}000$ images, with pixel intensities divided by $256$. To obtain a larger test set, the $10{,}000$ MNIST test images are complemented by the $50{,}000$ additional test images of the QMNIST dataset~\cite{yadav2019cold}, which were reconstructed from the same NIST source data as the MNIST test set and are disjoint from the training set, giving $60{,}000$ test images in total. They achieve a final train set accuracy of $99.89(3)\%$ and a test set accuracy of $95.5(3)\%$.

For the \emph{single-neuron lesioning} analysis (Figures~\ref{fig:mnist}c and~\ref{fig:mnist}d), each neuron of the three five-neuron hidden layers of the fully trained MNIST classifiers (epoch $10^4$) was ablated individually by clamping its activation to zero, without rescaling the surviving activations. The lesion damage is quantified as the drop in test-set accuracy of the lesioned network relative to the intact network; this was repeated for every neuron and all ten runs, and analogously for the altered layer of the networks trained with one-neuron dropout in one layer (Figure~\ref{fig:mnist}c). Because the lesioned neuron is constant, the information about the target that remains available to the downstream layers is at most $I(\bm X_{-i};Y)$. This motivates comparing the lesion damage of neuron $i$ to its per-neuron vulnerability $v_i := I(X_i;Y|\bm X_{-i})/I(\bm X;Y)$, i.e., its contribution to the average degree of vulnerability $\bar{v}=\sum_i v_i$ in Eq.~\eqref{eq:vbar_proof}. The $v_i$ are obtained from the same measurements on the training set that underlie Figure~\ref{fig:mnist}f, at the final training epoch. Appendix~\ref{app:ablation_bounds} makes this connection formal by deriving decoder-independent lower bounds on the ablation damage in terms of $v_i$.

For the \emph{one-neuron dropout} conditions (dashed lines in Figures~\ref{fig:mnist}b, g and~h), additional networks were trained in which one of the three five-neuron hidden layers receives a dropout intervention throughout training: in every forward pass, after quantisation of that layer's activations, one neuron of the layer --- chosen uniformly at random and independently for each sample --- is set to zero, without rescaling of the surviving activations. Architecture, quantisation and training protocol are otherwise identical to the baseline networks; ten runs with different Xavier random weight initialisations were trained per condition. The dropout is a pure training-time intervention: all performance evaluations and information-theoretic measurements are carried out without it.

For the \emph{Face autoencoder network}, the ``Labeled Faces in the Wild''~\cite{LFWTech} dataset is augmented by adding images that have been flipped horizontally as well as rotated by $-10$, $-5$, $0$, $5$ or $10$ degrees. Combining these augmentations, the dataset is expanded to ten times more samples which have similar statistics as the original images, which enables the estimation of small information-theoretic quantities. The images are preprocessed by resizing the original images to $40 \times 40$ pixels before cropping $4$ pixels from each border to produce $32 \times 32$ images with three colour channels. The augmented dataset of $132{,}330$ images is shuffled with a fixed random seed and split into $100{,}000$ training and $32{,}330$ test images. Because the augmentation is applied before the split, the test loss is not claimed to measure generalisation to unseen faces and is reported only as a consistency check (Table~\ref{tab:face_losses}); all information-theoretic quantities are computed on the training set.

The network has three convolutional encoder layers with a kernel size of $3 \times 3$, padding of $1$, stride of $1$ and $2 \times 2$ max pooling layers after each. The output of the third convolutional layer is linearised and fully connected to a variable-size bottleneck layer, after which three transpose convolutional layers with kernel size $3$, padding of $1$ and stride of $2$ finally produce a reconstruction of the input image. As in the MNIST example, all layers use tanh activation functions. The autoencoder networks use mean square error and the Adam optimiser with an initial learning rate of $0.001$. The autoencoder networks were run $10$ times each for $n_\mathrm{b}=128$, $n_\mathrm{b}=20$, $n_\mathrm{b}=16$ and $n_\mathrm{b}=12$ neurons in the bottleneck layer each with different Xavier weight initialisations. The final train and test set losses after $10^3$ epochs are summarised in Table \ref{tab:face_losses}.

Training the face autoencoder networks required approximately $290$ NVIDIA A100 GPU-hours, and their information-theoretic analysis required several thousand CPU core-hours on AMD Zen 3 CPUs. The trained baseline MNIST classifiers were reused from a previous publication~\cite{ehrlich2023a}; the additional training runs for the quantisation-robustness and one-neuron dropout conditions, together with all MNIST analyses, required less than $10^3$ CPU core-hours.

\begin{table}
\begin{tabular}{r|r|r}
$n_\mathrm{b}$ & Final train loss & Final test loss\\
\hline
$128$ & $9.43(9)\times 10^{-3}$ & $9.5(1)\times 10^{-3}$\\
$20$ & $1.70(4)\times 10^{-2}$ & $1.70(4)\times 10^{-2}$\\
$16$ & $1.86(3)\times 10^{-2}$ & $1.86(3)\times 10^{-2}$\\
$12$ & $2.12(4)\times 10^{-2}$ & $2.12(4)\times 10^{-2}$
\end{tabular}
\caption{Final train and test set mean square error losses for the face autoencoder network for different bottleneck sizes $n_\mathrm{b}$.}
\label{tab:face_losses}
\end{table}

\subsection{Shannon observables for directed information flow}
\label{sec:methods_information_flow}

Consider $n$ source processes and a target process observed over time. Let $\bm X_t^{-} = (X_{1,t}^{-},\dots,X_{n,t}^{-})$ denote the pasts of the source processes and $Y_t^{-}$ the past of the target, each including the current state. Transfer entropy~\cite{schreiber2000measuring,bossomaier2016transfer,wibral2014transfer} is defined as the conditional mutual information
\begin{equation}
\mathcal{T}_{\bm X \to Y} := I(\bm X_t^{-} ; Y_{t+1} \mid Y_t^{-}),
\end{equation}
which quantifies how much the past of a source process helps predict the next state of a target process beyond the target's own history. Because transfer entropy is a conditional mutual information, the Shannon observables framework can be used to decompose it, revealing the redundancy and vulnerability structure of directed information flow. One treats $X_{1,t}^{-},\dots,X_{n,t}^{-}$ as the sources and $Y_{t+1}$ as the target, replacing all mutual information terms by their counterparts conditioned on $Y_t^{-}$ (see Appendix~\ref{app:te} for additional formal details). Exactly as Propositions~\ref{prop:r_bar_shannon_observable} and~\ref{prop:v_bar_shannon_observable} express $\bar{r}$ and $\bar{v}$ in terms of ordinary mutual information, the corresponding derivation yields the conditional Shannon observables in terms of transfer entropies:
\begin{align}
\bar{r}_{\mathcal{T}} &= \frac{\sum_{i=1}^n \mathcal{T}_{X_i \to Y}}{\mathcal{T}_{\bm X \to Y}}, \label{eq:rbar_te}\\
\bar{v}_{\mathcal{T}} &= \frac{\sum_{i=1}^n \mathcal{T}_{X_i \to Y \mid \bm X_{-i}}}{\mathcal{T}_{\bm X \to Y}}. \label{eq:vbar_te}
\end{align}
The interpretation of these quantities is also analogous: $\bar{r}_{\mathcal{T}}$ measures the average number of individual source processes from which information transfer to the target can be obtained, while $\bar{v}_{\mathcal{T}}$ measures the average number of source processes on which information transfer critically depends.

Estimating transfer entropy nonparametrically becomes impractical in
high-dimensional systems due to the curse of dimensionality: the amount of
data required grows exponentially with the number of source processes. To make
the framework applicable to high-dimensional systems, we
employ Granger causality (GC), denoted by $\mathcal{F}$ here,  as a linear
approximation~\cite{geweke1982measurement,geweke1984measures}. Under the
assumption that the processes are jointly Gaussian, GC is proportional to
transfer entropy~\cite{barnett2009granger}, so that the Shannon observables
defined in Eqs.~\eqref{eq:rbar_te}--\eqref{eq:vbar_te} can be computed from
GC estimates.

Specifically, we use \emph{state-space models}~\cite{barnett2015granger}, which
are equivalent to  vector autoregressive moving-average models  (VARMA) and therefore capture
moving-average components that standard vector autoregressive (VAR) models
miss. This distinction is important for neural data, where routine
preprocessing steps---like filtering or downsampling---introduce moving-average components that can severely
bias VAR-based estimates~\cite{barnett2015granger}.

In the MEG analyses presented in Results \ref{sec:meg}, we consider a network of $N$ brain regions, each
observed as a time series. Every region can serve as a target, with the
remaining $N-1$ regions acting as sources. This yields per-region Shannon
observables $\bar{r}(i)$ and $\bar{v}(i)$, characterising the redundancy
and vulnerability of the information \emph{received} by region~$i$ from the
rest of the network. Computing these requires three types of GC estimate for
each target region~$i$: (a)~the \emph{joint} GC,
$\mathcal{F}_{\bm X_{-i} \to i}$, from all other regions to~$i$, which
quantifies the total directed information the region receives and serves as
the denominator; (b)~\emph{pairwise} GC, $\mathcal{F}_{j \to i}$, from each
source region~$j$ to~$i$ ignoring all other regions, whose sum forms the
numerator of $\bar{r}(i)$; and (c)~\emph{conditional} GC,
$\mathcal{F}_{j \to i \mid \bm X_{-\{i,j\}}}$, from each source~$j$ to~$i$
given all remaining regions, whose sum forms the numerator of $\bar{v}(i)$.
The resulting expressions mirror
Eqs.~\eqref{eq:rbar_te}--\eqref{eq:vbar_te} exactly:
\begin{align}
\bar{r}(i) &= \frac{\sum_{j \neq i}
  \mathcal{F}_{j \to i}}{\mathcal{F}_{\bm X_{-i} \to i}},
  \label{eq:rbar_gc}\\
\bar{v}(i) &= \frac{\sum_{j \neq i}
  \mathcal{F}_{j \to i \mid \bm X_{-\{i,j\}}}}
  {\mathcal{F}_{\bm X_{-i} \to i}}. \label{eq:vbar_gc}
\end{align}

\subsection{MEG data and analysis}
\label{sec:methods_meg}

We analysed resting-state MEG recordings from the NIMH Healthy Research Volunteer Dataset~\cite{nugent2022nimh}, acquired with a 275-channel CTF magnetometer system during eyes-open rest (6 minutes per subject). Of the subjects with complete fiducial coil registration data required for MEG-to-MRI co-registration, one was excluded due to insufficient recording length and one due to a processing failure during cortical surface reconstruction, yielding a final sample of 54 subjects. Raw data were bandpass filtered (1--150\,Hz), notch filtered at 60 and 120\,Hz and resampled to 300\,Hz. The continuous recordings were segmented into 2-second epochs, and epochs exceeding a peak-to-peak amplitude of 4\,pT were rejected.

Individual cortical anatomy was reconstructed using FastSurfer~\cite{henschel2020fastsurfer}, and the cortex was parcellated into 62 regions according to the DKT atlas~\cite{klein2012101}. All MEG data processing and source reconstruction were performed using MNE-Python~\cite{gramfort2013meg}. MEG-to-MRI coregistration was performed by aligning the positions of three fiducial head coils, recorded during the MEG session, to their corresponding locations in MRI coordinates. Source-level activity was estimated using linearly constrained minimum variance (LCMV) beamforming~\cite{van1997localization}, and a single time course per region was obtained via PCA within each parcel.

State-space Granger causality was computed using the MVGC2 toolbox~\cite{barnett2015granger,barnett2014mvgc}. All regional time series were demeaned and variance-normalised per subject, and a single state-space model was fit to all included epochs jointly. Joint, pairwise, and conditional pairwise GC were then computed from the fitted model parameters, and the per-region Shannon observables $\bar{r}(i)$ and $\bar{v}(i)$ were obtained via Eqs.~\eqref{eq:rbar_gc}--\eqref{eq:vbar_gc}.

\section{Data availability}
The raw MEG and MRI data for the cortical atlas analysis are publicly available as part of the NIMH Healthy Research Volunteer Dataset on OpenNeuro (\url{https://openneuro.org/datasets/ds004215}). The source-reconstructed cortical time courses derived from these data are available on Zenodo (\url{https://doi.org/10.5281/zenodo.19250047}). The datasets used for the deep learning analyses are publicly available: MNIST~\cite{lecun_gradient-based_1998}, QMNIST~\cite{yadav2019cold} and Labeled Faces in the Wild~\cite{LFWTech} (deep-funneled version).

\section{Code availability}
The code for the Hamming (7,4) analysis is available at \url{https://github.com/aarongutknecht/shannon-observables-hamming}. The code for the deep learning analyses is available at \url{https://github.com/Priesemann-Group/nninfo.git}. The code for the MEG analysis is available at \url{https://github.com/aarongutknecht/shannon-observables-meg}.

\section{Acknowledgements}
We thank Andreas Schneider and Lionel Barnett for valuable discussions related to this work. The work of F.R. has been supported by the UK ARIA Safeguarded AI programme and by the PIBBSS Affiliateship programme. A.G., D.E., A.M., and M.W. are employed at the Göttingen Campus Institute for Dynamics of Biological Networks (CIDBN) funded by the Volkswagen Stiftung.

\appendix

\section{Analytical computation of Shannon observables for systematic binary block codes
under a binary symmetric channel}
\label{app:hamming}

This appendix details the analytical computation of the Shannon observables
$\bar{r}$ and $\bar{v}$ for the coding setup described in
Section~\ref{sec:hamming_results}. The derivation applies to any systematic
binary block code that maps $k$-bit messages to $n$-bit codewords and is
transmitted over a binary symmetric channel, with a uniformly distributed
message and with the target chosen as one of the message bits appearing in a
systematic coordinate. The Hamming~(7,4) code shown in
Figure~\ref{fig:hamming} is one particular instance of this setup, with block
length $n=7$ and message length $k=4$.

Concretely, a uniformly random message $\bm m \in \{0,1\}^k$ is encoded into a
codeword $\bm C = (C_1,\dots,C_n) \in \mathcal{C} \subseteq \{0,1\}^n$, where
$|\mathcal{C}| = 2^k$. The codeword is transmitted through a binary symmetric
channel with flip probability $\epsilon$, producing the received word
$\bm X = (X_1,\dots,X_n)$ with $X_i = C_i \oplus N_i$, where
$N_1,\dots,N_n \stackrel{\text{i.i.d.}}{\sim} \mathrm{Bernoulli}(\epsilon)$.
The target is taken to be one of the message bits. Since the code is
systematic, there is a designated systematic coordinate $j$ such that this
message bit appears verbatim as $Y = C_j$.

\paragraph{Required information-theoretic quantities}

By Propositions~\ref{prop:r_bar_shannon_observable}
and~\ref{prop:v_bar_shannon_observable}, the Shannon observables are
\begin{equation}
\bar{r}
= \frac{\sum_{i=1}^{n} I(X_i;Y)}{I(\bm X;Y)},
\qquad
\bar{v}
= \frac{\sum_{i=1}^{n} I(X_i;Y \mid \bm X_{-i})}{I(\bm X;Y)},
\end{equation}
where $\bm X_{-i} := (X_1,\dots,X_{i-1},X_{i+1},\dots,X_n)$. Since the message
is uniform and $Y$ is one of its bits, $Y$ is balanced and binary, so
$H(Y) = 1$. Hence all required mutual information terms can be written in terms
of conditional entropies:
\begin{align}
I(\bm X;Y) &= 1 - H(Y \mid \bm X), \label{eq:ham_mi_joint}\\
I(X_i;Y) &= 1 - H(Y \mid X_i), \label{eq:ham_mi_marginal}\\
I(X_i;Y \mid \bm X_{-i}) &= H(Y \mid \bm X_{-i}) - H(Y \mid \bm X). \label{eq:ham_cmi}
\end{align}
It therefore suffices to compute the conditional entropy $H(Y \mid \bm X)$
together with $H(Y \mid X_i)$ and $H(Y \mid \bm X_{-i})$ for each position
$i \in \{1,\dots,n\}$.

\paragraph{Conditional entropy from likelihoods}

We illustrate the approach for the joint case $H(Y \mid \bm X)$; the marginal
and leave-one-out cases follow by the same logic. To evaluate
$H(Y \mid \bm X)$, we need the conditional probabilities
$P(Y = y \mid \bm X = x)$, which follow from Bayes' rule and the channel
model. Define the likelihoods
\begin{multline}
L_y(x) := P(\bm X = x \mid Y = y)\\
= \frac{1}{2^{k-1}} \sum_{\substack{c \in \mathcal{C} \\ c_j = y}}
   \epsilon^{\,d(x,c)}(1-\epsilon)^{n-d(x,c)},
\label{eq:ham_likelihood}
\end{multline}
where $d(x,c) := \#\{i : x_i \neq c_i\}$ is the Hamming distance. The factor
$1/2^{k-1}$ arises because the uniform message distribution induces a uniform
distribution on the $2^k$ codewords (the encoding being injective by virtue of
the message bits appearing verbatim in the systematic coordinates), and
because $j$ is a systematic coordinate corresponding to one message bit, so
exactly $2^{k-1}$ codewords satisfy $c_j=y$ for each value of $y$.

Since $P(Y=0) = P(Y=1) = \tfrac{1}{2}$, Bayes' rule gives, for every received
word $x$ with $P(\bm X=x)>0$,
\begin{equation}
P(Y=1 \mid \bm X = x) = \frac{L_1(x)}{L_0(x) + L_1(x)}.
\end{equation}
Because $Y$ is binary,
$H(Y \mid \bm X = x) = h_2\!\bigl(P(Y=1 \mid \bm X = x)\bigr)$, where $h_2$
is the binary entropy function. Averaging over $x$ and using
$P(\bm X = x) = \tfrac{1}{2}(L_0(x) + L_1(x))$ yields
\begin{multline}
H(Y \mid \bm X)
= \sum_{x \in \{0,1\}^n} \frac{L_0(x) + L_1(x)}{2}\\
\times\; h_2\!\left(\frac{L_1(x)}{L_0(x) + L_1(x)}\right).
\label{eq:ham_HT_joint}
\end{multline}
In this and the following entropy sums, terms corresponding to zero-probability
events are taken to contribute zero. This convention covers, in particular, the
noiseless case $\epsilon=0$.

\paragraph{Single source and leave-one-out conditioned entropies}

The entropies $H(Y \mid X_i)$ and $H(Y \mid \bm X_{-i})$ are obtained by the
same reasoning, replacing $L_y(x)$ with its marginalisations. Throughout, we
write $(x_i, x_{-i})$ for the $n$-vector with $x_i$ at position $i$ and the
components of $x_{-i}$ at the remaining $n-1$ positions, so that
$L_y(x_i, x_{-i}) = L_y(x)$. The marginalisations are then
\begin{align}
L_y(x_i) &:= P(X_i = x_i \mid Y = y) \nonumber\\
&\;= \sum_{x_{-i} \in \{0,1\}^{n-1}} L_y(x_i, x_{-i}),
   \label{eq:ham_lik_single}\\
L_y(x_{-i}) &:= P(\bm X_{-i} = x_{-i} \mid Y = y) \nonumber\\
&\;= \sum_{x_i \in \{0,1\}} L_y(x_i, x_{-i}). \label{eq:ham_lik_drop}
\end{align}
Substituting these into the same formula gives
\begin{multline}
H(Y \mid X_i)
= \sum_{x_i \in \{0,1\}} \frac{L_0(x_i) + L_1(x_i)}{2}\\
\times\; h_2\!\left(\frac{L_1(x_i)}{L_0(x_i) + L_1(x_i)}\right),
\label{eq:ham_HT_single}
\end{multline}
\begin{multline}
H(Y \mid \bm X_{-i})
= \sum_{x_{-i} \in \{0,1\}^{n-1}} \frac{L_0(x_{-i}) + L_1(x_{-i})}{2}\\
\times\; h_2\!\left(\frac{L_1(x_{-i})}{L_0(x_{-i}) + L_1(x_{-i})}\right).
\label{eq:ham_HT_drop}
\end{multline}
Inserting Eqs.~\eqref{eq:ham_HT_joint}, \eqref{eq:ham_HT_single},
and~\eqref{eq:ham_HT_drop} into Eqs.~\eqref{eq:ham_mi_joint}--\eqref{eq:ham_cmi}
gives all required mutual information terms, and hence the Shannon observables
$\bar{r}$ and $\bar{v}$ as functions of $\epsilon$. If
$I(\bm X;Y)=0$, the normalised observables are undefined; otherwise, the sums
above give exact values for any fixed finite code and any specified channel
error probability. For the Hamming~(7,4) code with $j=3$, this procedure
produces the curves shown in Figure~\ref{fig:hamming}.

\section{Duality between degree of redundancy and degree of vulnerability}
\label{app:duality}

The redundancy lattice $(\mathcal{A}_n, \preceq)$ \cite{williams2010nonnegative} admits a natural order-reversing bijection $\alpha \mapsto -\alpha$ that maps it to its dual $(\mathcal{A}_n, \succeq)$~\cite{mages2024non}. Concretely, for an antichain $\alpha = \{\mathbf{a}_1,\ldots,\mathbf{a}_m\} \subseteq 2^{[n]}$, this mapping is defined as
\begin{equation}\label{eq:duality_map}
    -\alpha \;:=\; \min_{\subseteq} \!\Big\{ \mathrm{set}(t) \;\Big|\; t \in \bigtimes_{\mathbf{a} \in \alpha} \mathbf{a} \Big\},
\end{equation}
where $\bigtimes$ denotes the Cartesian product and $\mathrm{set}(t) := \{i \mid i \text{ appears in the tuple } t\}$ converts each tuple to a set. The $\min_\subseteq$ operation retains only the $\subseteq$-minimal elements, so that $-\alpha$ is again an antichain. This mapping satisfies for all $\alpha,\beta \in \mathcal{A}_n$
\begin{equation}\label{eq:order_reversing}
\alpha \preceq \beta \;\Longleftrightarrow\; -\beta \preceq -\alpha,
\end{equation}
and is an involution, i.e.\ $-(-\alpha) = \alpha$ (see~\cite{mages2024non} for proofs of these properties). The central result of this appendix is that this lattice duality interchanges the degree of redundancy and the degree of vulnerability. It thereby also provides a formal justification for the name ``dual'' in dual redundancy-synergy index.

\begin{proposition}[Duality interchanges $r$ and $v$]\label{prop:rv_duality}
For all $\alpha \in \mathcal{A}_n$,
\begin{equation}\label{eq:r_v_duality}
    r(\alpha) = v(-\alpha) \qquad \text{and} \qquad v(\alpha) = r(-\alpha).
\end{equation}
\end{proposition}

\begin{proof}
Write $\alpha = \{\mathbf{a}_1,\ldots,\mathbf{a}_m\}$. Since $r$ counts the singletons $\{j\}$ present in an antichain and $v$ counts the indices $j$ belonging to every set in an antichain, proving $r(\alpha)=v(-\alpha)$ amounts to establishing, for every $j\in[n]$, the equivalence
\begin{equation}\label{eq:singleton_iff_intersection}
    \{j\} \in \alpha \;\Longleftrightarrow\; j \in \bigcap_{\mathbf{b} \in -\alpha} \mathbf{b}.
\end{equation}

\noindent ($\Rightarrow$)\; Suppose $\{j\}\in\alpha$. Then $j\in\mathbf{a}_k$ for $\mathbf{a}_k = \{j\}$, so every tuple $t\in\bigtimes_{\mathbf{a}\in\alpha}\mathbf{a}$ must select $j$ from this component. Hence $j\in\mathrm{set}(t)$ for every such tuple, and therefore $j$ belongs to every element of $\mathcal{S} := \{\mathrm{set}(t) \mid t\in\bigtimes_{\mathbf{a}\in\alpha}\mathbf{a}\}$, and in particular to all $\subseteq$-minimal elements, i.e.\ to every $\mathbf{b}\in -\alpha$.

\noindent ($\Leftarrow$)\; Suppose $j\in\bigcap_{\mathbf{b}\in -\alpha}\mathbf{b}$. Since every element of $\mathcal{S}$ is a superset of some element of $-\alpha$, it follows that $j\in\mathrm{set}(t)$ for all $t\in\bigtimes_{\mathbf{a}\in\alpha}\mathbf{a}$. In other words, whenever we choose one index from each $\mathbf{a}\in\alpha$ and collect them into a set, $j$ is always present. Now assume for contradiction that $\{j\}\notin\alpha$. Then every $\mathbf{a}_k\in\alpha$ either is a singleton $\{i\}$ with $i\neq j$, or contains at least two elements. In either case, $\mathbf{a}_k$ contains some index different from $j$, so we can select one such index from each $\mathbf{a}_k$, producing a tuple whose set does not contain $j$ --- a contradiction.

The second identity $v(\alpha)=r(-\alpha)$ follows immediately from the first and the involution property $-(-\alpha)=\alpha$: applying $r(\alpha)=v(-\alpha)$ to $-\alpha$ gives $r(-\alpha) = v(-(-\alpha)) = v(\alpha)$.
\end{proof}

We now state the consequences of this duality for the average degrees and the RSI and DRSI.

\begin{corollary}[Duality of $\bar{r}$ and $\bar{v}$]\label{cor:rbar_vbar_duality}
Replacing $\Pi(\alpha)$ by $\Pi(-\alpha)$ in the definition of $\bar{r}$ yields $\bar{v}$, and vice versa:
\begin{equation}
    \frac{\sum_{\alpha} r(\alpha)\,\Pi(-\alpha)}{I(\bm X;Y)} = \bar{v}, \qquad
    \frac{\sum_{\alpha} v(\alpha)\,\Pi(-\alpha)}{I(\bm X;Y)} = \bar{r}.
\end{equation}
\end{corollary}

\begin{proof}
Since $\alpha \mapsto -\alpha$ is a bijection on $\mathcal{A}_n$, substituting $\beta = -\alpha$ and applying Proposition~\ref{prop:rv_duality} gives $\sum_{\alpha} r(\alpha)\,\Pi(-\alpha) = \sum_{\beta} r(-\beta)\,\Pi(\beta) = \sum_{\beta} v(\beta)\,\Pi(\beta)$. Dividing by $I(\bm X;Y)$ yields the first identity; the second follows by the same argument using $v(-\beta) = r(\beta)$.
\end{proof}

\begin{corollary}[Duality of RSI and DRSI]\label{cor:rsi_drsi_duality}
Replacing $\Pi(\alpha)$ by $\Pi(-\alpha)$ in the expression for the RSI yields $-\normalfont{\text{DRSI}}$, and vice versa.
\end{corollary}

\begin{proof}
By Corollary~\ref{cor:rbar_vbar_duality}, replacing $\Pi(\alpha)$ by $\Pi(-\alpha)$ sends $\bar{r} \mapsto \bar{v}$. Since $\text{RSI} = (\bar{r}-1)\,I(\bm X;Y)$, this replacement yields $(\bar{v}-1)\,I(\bm X;Y) = -(1-\bar{v})\,I(\bm X;Y) = -\text{DRSI}$. The converse follows by applying the replacement twice and using $-(-\alpha) = \alpha$.
\end{proof}

\begin{remark}
The sign reversal arises because both the RSI and DRSI are defined as ``redundancy minus synergy'' in their respective senses. This is merely a convention. Directly dualising the RSI would produce ``synergy minus redundancy'' in the vulnerability sense, which is $-\text{DRSI}$.
\end{remark}

\section{Bounds on total redundancy and vulnerability derived from the average degrees}
\label{app:bounds}
For the interpretation of the degrees of redundancy and vulnerability in PID terms it is important to understand the bounds placed by certain values of $\bar{r}$ and $\bar{v}$ on different kinds of information atoms. We have the following proposition (for the definitions of $I_r^{(k)}$ and $I_v^{(k)}$ see Eqs.~\eqref{eq:def_Irk} and~\eqref{eq:def_Ivk}):
\begin{proposition}
Assuming non-negativity of PID atoms, $\frac{I_r^{(0)} }{I(\bm X;Y)} \geq 1- \bar{r}$    and $\frac{I_v^{(0)} }{I(\bm X;Y)} \geq 1- \bar{v}$.
\end{proposition}
\begin{proof}
We argue the case for $\bar{r}$. The argument for $\bar{v}$ is analogous.  First, the normalised k-redundancies must sum up to 1:
\begin{equation}
   \sum_{k=0}^n \frac{I_r^{(k)} }{I(\bm X;Y)} = 1
\end{equation}
Thus,
\begin{align}
\frac{I_r^{(0)} }{I(\bm X;Y)} &= 1 -  \left( \sum_{k=1}^n \frac{I_r^{(k)} }{I(\bm X;Y)}\right)\\
&\geq 1 -  \left( \sum_{k=1}^n k \frac{I_r^{(k)} }{I(\bm X;Y)}\right) = 1 - \bar{r}
\end{align}
\end{proof}
Note that for $\bar{r}>1$ or $\bar{v}>1$ the statement is true but trivial because assuming non-negativity of atoms already implies that $I_r^{(0)}$ and $I_v^{(0)}$ must be non-negative.
\begin{corollary} As special cases the proposition implies
\begin{enumerate}
\item $\bar{r}<1$ means that there must be some source-level synergy, i.e. $I_r^{(0)}>0$.
\item $\bar{r}<0.5$ means that source-level synergy is predominant, i.e. $\frac{I_r^{(0)}}{I(\bm X;Y)}>0.5$.
\item $\bar{v}<1$ means that there must be some robustness, i.e. $I_v^{(0)}>0$.
 \item $\bar{v}<0.5$ means that robustness is predominant, i.e. $\frac{I_v^{(0)}}{I(\bm X;Y)}>0.5$.
\end{enumerate}
\end{corollary}

The interpretation of values of $\bar{r}$ and $\bar{v}$ larger than one is slightly more complicated, as it involves the components of these measures that distinguish between different orders of redundancy and vulnerability—where higher-order terms enter with larger coefficients. Thus, values larger than one can always be explained in two distinct ways: either by a large weight on lower-order terms or by a smaller weight on higher-order terms. Hence, in general the bounds on redundancy and vulnerability to be inferred from such values depend on $n$ since this determines how many orders of redundancy or vulnerability are distinguished. We have the following proposition for these bounds:

\begin{proposition}
Assuming non-negativity of information atoms,
\begin{equation}
\frac{\sum_{k=2}^n I_r^{(k)}}{I(\bm X;Y)} \geq \frac{\bar{r}-1}{n-1}
\end{equation}
and
\begin{equation}
\frac{\sum_{k=2}^n I_v^{(k)}}{I(\bm X;Y)} \geq \frac{\bar{v}-1}{n-1}
\end{equation}
\end{proposition}
\begin{proof}
Again we argue the case $\bar{r}$ with the argument for $\bar{v}$ proceeding in the same way. For any given value of $\bar{r}>1$ the scenario requiring the least proportion of redundancy is the one where all redundancy is of the highest-order and there is no source-level synergy. In this case, we have
\begin{align}
\bar{r} &= \frac{I_r^{(1)} }{I(\bm X;Y)} + n\frac{I_r^{(n)}}{I(\bm X;Y)} \\
&= \left(1-\frac{I_r^{(n)}}{I(\bm X;Y)}\right) + n\frac{I_r^{(n)}}{I(\bm X;Y)} \\
&= (n-1) \frac{I_r^{(n)}}{I(\bm X;Y)}+1
\end{align}
which implies
\begin{equation}
\frac{\bar{r} -1}{n-1}= \frac{I_r^{(n)}}{I(\bm X;Y)}
\end{equation}
This is therefore the minimal value required to achieve the given value of $\bar{r}$.
\end{proof}

\begin{corollary}
The proposition entails in particular that
\begin{enumerate}
    \item $\bar{r}>1$ means that there must be some source-level redundancy, i.e. $I_r^{(k)}>0$ for some $k>1$.
    \item $\bar{r}>\frac{1}{2}(n+1)$ means that source-level redundancy is predominant, i.e. $\frac{\sum_{k=2}^nI_r^{(k)}}{I(\bm X;Y)}>0.5$.
    \item $\bar{v}>1$ means that there must be some multi-source vulnerability, i.e. $I_v^{(k)}>0$ for some $k>1$.
     \item $\bar{v}> \frac{1}{2}(n+1)$ means that multi-source vulnerability is predominant, i.e. $\frac{\sum_{k=2}^nI_v^{(k)}}{I(\bm X;Y)}>0.5$.
\end{enumerate}
\end{corollary}

In 2. and 4. we can see the dependence on $n$: the larger $n$ the higher the values necessary to infer a predominance of redundancy or vulnerability. A more general statement regarding the required values to infer a certain percentage of redundancy or vulnerability is
\begin{corollary}
\begin{equation}
\bar{r} > \lambda(n-1)+1 \;\Longrightarrow\; \frac{\sum_{k=2}^nI_r^{(k)}}{I(\bm X;Y)}>\lambda
\end{equation}
\begin{equation}
\bar{v} > \lambda(n-1)+1 \;\Longrightarrow\; \frac{\sum_{k=2}^nI_v^{(k)}}{I(\bm X;Y)}>\lambda
\end{equation}
That is, source-level redundancy exceeds proportion $\lambda$ whenever $\bar{r}$ exceeds $\lambda(n-1)+1$, and multi-source vulnerability exceeds proportion $\lambda$ whenever $\bar{v}$ exceeds the same threshold.
\end{corollary}

\section{Operational meaning of the average degree of vulnerability: bounds on ablation damage}
\label{app:ablation_bounds}

The average degree of vulnerability $\bar{v}$ is defined in the main text as
a structural property of an information decomposition: it counts, for a
typical piece of information about $Y$, how many sources that information
critically depends on.  This appendix establishes a complementary,
\emph{operational} reading of the same quantity.  We show that the
per-source terms of $\bar{v}$ lower-bound the performance degradation
suffered by \emph{any} decoder when the corresponding source is ablated, and
that for the logarithmic loss this bound is attained exactly by an optimally
re-fitted decoder.  The average degree of vulnerability therefore quantifies
not only how information is organised, but how much performance is
unavoidably lost when a single source fails.

This provides a formal underpinning for \emph{ablation} (or
\emph{knock-out}) analyses, a standard tool for assessing the importance of
individual units in biological and artificial networks, and for
permutation-based variable importance in
statistics~\cite{breiman2001random}.  Throughout, entropies and mutual
informations are measured in bits, $\mathcal{Y}$ is finite with
$M = |\mathcal{Y}| \geq 2$, and all sources are discrete.

\paragraph{Setup.}
Let $\bm X = (X_1,\dots,X_n)$ be the sources and $Y$ the target.  Recalling
the Shannon-observable form of the average degree of vulnerability,
Eq.~\eqref{eq:vbar_proof}, where $\bm X_{-j}$ denotes all sources except
$X_j$, we write
\begin{equation}\label{eq:vj}
    v_j \;\coloneqq\; \frac{I(X_j;Y \mid \bm X_{-j})}{I(\bm X;Y)},
    \qquad\text{so that}\qquad
    \bar{v} \;=\; \sum_{j=1}^{n} v_j ,
\end{equation}
and $v_j$ is the share of the total information $I(\bm X;Y)$ that critically
depends on source $X_j$.

\begin{definition}[Ablation of a source]\label{def:ablation}
An \emph{ablation} of source $j$ replaces $X_j$ by a surrogate $\xi_j$ that
is statistically independent of $(\bm X, Y)$, yielding the ablated source
vector
\begin{equation}
    \tilde{\bm X}^{(j)} \;=\; \big(\bm X_{-j},\; \xi_j\big).
\end{equation}
\end{definition}

Definition~\ref{def:ablation} covers the interventions used in practice:
\emph{knock-out} ($\xi_j \equiv 0$) and \emph{mean ablation}
($\xi_j \equiv \mathbb{E}[X_j]$) take $\xi_j$ deterministic, while
\emph{resampling} or \emph{permutation} ablation draws $\xi_j \sim p_{X_j}$
independently of everything else, which is the population limit of permuting
the values of $X_j$ across samples~\cite{breiman2001random}.  What matters
for the results below is only that the surrogate carries no information
about $Y$ beyond what $\bm X_{-j}$ already provides.

\begin{lemma}[Ablation identity]\label{lem:ablation}
For any ablation in the sense of Definition~\ref{def:ablation},
\begin{equation}\label{eq:ablation_identity}
\begin{aligned}
    H\big(Y \mid \tilde{\bm X}^{(j)}\big)
    &= H\big(Y \mid \bm X_{-j}\big) \\
    &= H(Y \mid \bm X) \;+\; v_j\, I(\bm X;Y).
\end{aligned}
\end{equation}
Moreover, for any further (possibly stochastic) processing $g$ of the
ablated sources,
$H\big(Y\mid g(\tilde{\bm X}^{(j)})\big) \geq H\big(Y\mid \tilde{\bm X}^{(j)}\big)$.
\end{lemma}

\begin{proof}
Since $\xi_j$ is independent of $(\bm X,Y)$ we have
$p\big(\tilde{\bm x}^{(j)}\mid \bm x_{-j}, y\big)
 = p\big(\tilde{\bm x}^{(j)}\mid \bm x_{-j}\big)$, i.e.\
$Y \to \bm X_{-j} \to \tilde{\bm X}^{(j)}$ is a Markov chain, so the data
processing inequality~\cite{cover1999elements} gives
$I\big(\tilde{\bm X}^{(j)};Y\big) \leq I(\bm X_{-j};Y)$.  Conversely,
$\bm X_{-j}$ is a deterministic function of $\tilde{\bm X}^{(j)}$, so the
same inequality applied to $Y \to \tilde{\bm X}^{(j)} \to \bm X_{-j}$ yields
$I\big(\tilde{\bm X}^{(j)};Y\big) \geq I(\bm X_{-j};Y)$.  The two mutual
informations therefore coincide, and so do the corresponding conditional
entropies, which gives the first equality.  For the second, the chain rule
$I(\bm X;Y) = I(\bm X_{-j};Y) + I(X_j;Y\mid\bm X_{-j})$ together with
$H(Y\mid \bm X_{-j}) = H(Y) - I(\bm X_{-j};Y)$,
$H(Y\mid \bm X) = H(Y) - I(\bm X;Y)$ and definition~\eqref{eq:vj} give
$H(Y\mid\bm X_{-j}) = H(Y\mid\bm X) + v_j I(\bm X;Y)$.  The final claim is
the data processing inequality applied to
$Y \to \tilde{\bm X}^{(j)} \to g(\tilde{\bm X}^{(j)})$.
\end{proof}

Lemma~\ref{lem:ablation} is the crux of what follows: \emph{ablating source
$j$ raises the residual uncertainty about $Y$ by exactly $v_j I(\bm X;Y)$
bits}, irrespective of which ablation surrogate is used and of any
downstream processing.  In particular, rescaling the surviving sources to
compensate for the lost activation --- as is customary in
inverted-dropout-style corrections --- cannot recover any of the lost
information.

\subsection{Lower bound for classification error}

\begin{theorem}[Ablation bound via Fano's inequality]\label{thm:fano}
Let $\hat{Y} = d\big(\tilde{\bm X}^{(j)}\big)$ be the output of an arbitrary,
possibly randomised, decoder acting on the ablated sources, and let
$P_e^{(j)} = \Pr\big[\hat{Y} \neq Y\big]$.  Define
\begin{equation}\label{eq:phi}
    \varphi(p) \;=\; h(p) + p \log_2 (M-1),
    \qquad p \in \big[0,\, 1-\tfrac{1}{M}\big],
\end{equation}
where $h(p) = -p\log_2 p - (1-p)\log_2(1-p)$ is the binary entropy function.
Then $\varphi$ is a strictly increasing, concave bijection from
$[0, 1-\frac{1}{M}]$ onto $[0, \log_2 M]$, and
\begin{equation}\label{eq:fano_bound}
\begin{aligned}
    P_e^{(j)}
    &\geq
    \varphi^{-1}\!\Big( H(Y\mid\bm X) + v_j\, I(\bm X;Y) \Big) \\
    &\geq
    \varphi^{-1}\!\big( v_j\, I(\bm X;Y) \big).
\end{aligned}
\end{equation}
\end{theorem}

\begin{proof}
Differentiating~\eqref{eq:phi} gives
$\varphi'(p) = \log_2\!\big[(1-p)(M-1)/p\big]$, which is strictly positive
for $p < 1-\frac{1}{M}$ and vanishes at $p = 1-\frac{1}{M}$; moreover
$\varphi(0) = 0$ and $\varphi\big(1-\frac{1}{M}\big) = \log_2 M$.  As the sum
of the concave function $h$ and a linear term, $\varphi$ is concave.  This
establishes the claimed bijection, so $\varphi^{-1}$ is well defined,
increasing and convex on $[0,\log_2 M]$.

Fano's inequality~\cite{fano1961transmission,cover1999elements} states that
$\varphi\big(P_e^{(j)}\big) \geq H\big(Y\mid\hat{Y}\big)$.  Because
$Y \to \bm X_{-j} \to \tilde{\bm X}^{(j)} \to \hat{Y}$ is a Markov chain,
Lemma~\ref{lem:ablation} gives
\begin{equation}
\begin{aligned}
    H\big(Y\mid\hat{Y}\big)
    &\geq H\big(Y\mid\tilde{\bm X}^{(j)}\big) \\
    &= H(Y\mid\bm X) + v_j\, I(\bm X;Y)
    \;\eqqcolon\; c .
\end{aligned}
\end{equation}
If $P_e^{(j)} \leq 1-\frac{1}{M}$, applying the increasing map
$\varphi^{-1}$ to $\varphi\big(P_e^{(j)}\big) \geq c$ yields
$P_e^{(j)} \geq \varphi^{-1}(c)$.  If instead
$P_e^{(j)} > 1-\frac{1}{M}$, the same conclusion holds trivially because
$\varphi^{-1}(c) \leq 1-\frac{1}{M}$ by construction.  The second inequality
in~\eqref{eq:fano_bound} follows from $H(Y\mid\bm X)\geq 0$ and monotonicity
of $\varphi^{-1}$.
\end{proof}

\begin{corollary}[Accuracy drop]\label{cor:accdrop}
If the intact system attains error probability $P_e^{\mathrm{intact}}$, the
accuracy drop caused by ablating source $j$ satisfies
\begin{equation}
\begin{aligned}
    \Delta\mathrm{acc}_j
    &\coloneqq P_e^{(j)} - P_e^{\mathrm{intact}} \\
    &\geq
    \varphi^{-1}\!\Big( H(Y\mid\bm X) + v_j\, I(\bm X;Y) \Big)
    - P_e^{\mathrm{intact}} .
\end{aligned}
\end{equation}
\end{corollary}

\begin{corollary}[Random single-source failure]\label{cor:vbar_fano}
Let $J$ be uniformly distributed on $\{1,\dots,n\}$, modelling the failure of
a single, randomly chosen source.  Then
\begin{equation}\label{eq:vbar_fano}
    \mathbb{E}_J\big[P_e^{(J)}\big]
    \;\geq\;
    \varphi^{-1}\!\Big( H(Y\mid \bm X) + \tfrac{\bar{v}}{n}\, I(\bm X;Y)\Big).
\end{equation}
\end{corollary}

\begin{proof}
By Theorem~\ref{thm:fano}, $\varphi^{-1}$ is convex and increasing, so
Jensen's inequality gives
$\mathbb{E}_J\big[\varphi^{-1}(H(Y\mid\bm X) + v_J I(\bm X;Y))\big]
 \geq \varphi^{-1}\big(H(Y\mid\bm X) + \mathbb{E}_J[v_J]\, I(\bm X;Y)\big)$,
and $\mathbb{E}_J[v_J] = \bar{v}/n$ by~\eqref{eq:vj}.
\end{proof}

Equation~\eqref{eq:vbar_fano} is the sense in which $\bar{v}$ itself --- and
not merely its individual terms --- controls fragility: a system whose
average degree of vulnerability is large cannot be robust to the loss of a
randomly chosen source, no matter how its downstream decoder is designed or
retrained.

\subsection{A tight bound for logarithmic loss}

Fano's inequality is known to be loose when the information deficit is
small~\cite{hellman1970probability,feder1994relations,sason2018arimoto}:
as $\epsilon \to 0$ one has
$\varphi^{-1}(\epsilon) = \Theta\big(\epsilon / \log_2(1/\epsilon)\big)$, so
the bound~\eqref{eq:fano_bound} becomes numerically weak precisely in the
regime of highly redundant codes, where $v_j$ is close to zero.  For the
logarithmic loss, by contrast, an exactly tight bound is available.

\begin{theorem}[Ablation bound for logarithmic loss]\label{thm:logloss}
Let $q\big(\cdot\mid\tilde{\bm X}^{(j)}\big)$ be any conditional distribution
on $\mathcal{Y}$ produced by a decoder from the ablated sources, and let
$\mathrm{CE}_j
 = \mathbb{E}\big[-\log_2 q\big(Y\mid\tilde{\bm X}^{(j)}\big)\big]$
be the resulting cross-entropy.  Then
\begin{equation}\label{eq:ce_bound}
    \mathrm{CE}_j
    \;\geq\; H(Y) - I(\bm X_{-j};Y)
    \;=\; H(Y\mid\bm X) + v_j\, I(\bm X;Y),
\end{equation}
with equality if and only if $q$ equals the Bayes posterior
$p\big(\cdot\mid\tilde{\bm X}^{(j)}\big)$ almost surely.
\end{theorem}

\begin{proof}
Writing the cross-entropy relative to the true posterior gives the standard
decomposition
\begin{equation}\label{eq:ce_decomp}
    \mathrm{CE}_j
    = H\big(Y\mid\tilde{\bm X}^{(j)}\big)
    + \mathbb{E}\Big[
        D_{\mathrm{KL}}\big(
            p(\cdot\mid\tilde{\bm X}^{(j)}) \,\big\|\,
            q(\cdot\mid\tilde{\bm X}^{(j)})
        \big)\Big],
\end{equation}
where the second term is non-negative and vanishes if and only if $q = p$
almost surely.  Lemma~\ref{lem:ablation} identifies the first term with
$H(Y\mid\bm X) + v_j I(\bm X;Y)$.
\end{proof}

\begin{corollary}[Log-loss degradation]\label{cor:ce_drop}
Let $\delta \coloneqq \mathrm{CE}_{\mathrm{intact}} - H(Y\mid\bm X) \geq 0$
denote the calibration gap of the intact decoder.  Then
\begin{align}
    \Delta \mathrm{CE}_j
    &\;\coloneqq\; \mathrm{CE}_j - \mathrm{CE}_{\mathrm{intact}}
    \;\geq\; v_j\, I(\bm X;Y) \;-\; \delta ,\\
    \sum_{j=1}^n \Delta\mathrm{CE}_j
    &\;\geq\; \bar{v}\, I(\bm X;Y) \;-\; n\,\delta .
\end{align}
In particular, for a well-calibrated intact decoder ($\delta \to 0$) the
total logarithmic-loss degradation incurred by ablating each source in turn
is at least $\bar{v}\, I(\bm X;Y)$ bits, and the mean degradation under a
random single-source failure is at least
$\tfrac{\bar{v}}{n} I(\bm X;Y)$ bits.
\end{corollary}

\begin{remark}[Information floor versus decoder mismatch]\label{rem:mismatch}
Applying~\eqref{eq:ce_decomp} to both the ablated and the intact decoder
yields the exact decomposition
\begin{equation}\label{eq:mismatch}
    \Delta\mathrm{CE}_j
    = \underbrace{v_j\, I(\bm X;Y)}_{\text{information loss}}
    \;+\;
    \underbrace{\mathbb{E}\big[D_{\mathrm{KL}}^{(j)}\big]
                - \mathbb{E}\big[D_{\mathrm{KL}}^{\mathrm{intact}}\big]}
              _{\text{decoder mismatch}} ,
\end{equation}
where $D_{\mathrm{KL}}^{(j)}$ and $D_{\mathrm{KL}}^{\mathrm{intact}}$ are the
posterior mismatches in~\eqref{eq:ce_decomp} for the ablated and intact
inputs.  The bounds of this appendix are statements about the first term
only, and therefore hold for the best decoder that could be fitted to the
ablated representation.  In an ablation experiment on a \emph{trained}
network the downstream weights are held fixed and are not re-optimised for
the ablated code, so the second term is typically positive and may dominate.
Equation~\eqref{eq:mismatch} makes this distinction precise: the Shannon
observable predicts the irreducible, decoder-independent component of the
damage, while the excess over the bound quantifies how far the frozen
decoder is from being Bayes-optimal on the ablated representation.
\end{remark}

Together, Theorem~\ref{thm:fano} and Theorem~\ref{thm:logloss} identify
$v_j\, I(\bm X;Y)$ as the decoder-independent floor on the damage caused by
ablating source $j$: it lower-bounds the achievable error probability through
Fano's inequality, and it is exactly the logarithmic loss incurred by an
optimally re-fitted decoder.  The average degree of vulnerability $\bar{v}$
thus admits a direct operational reading as the total, decoder-independent
performance cost of single-source failures, measured in units of the
system's total information.

\section{Lesion damage versus standard neuron-importance measures}
\label{app:importance_comparison}

\begin{figure*}
    \centering
    \includegraphics[width=\textwidth]{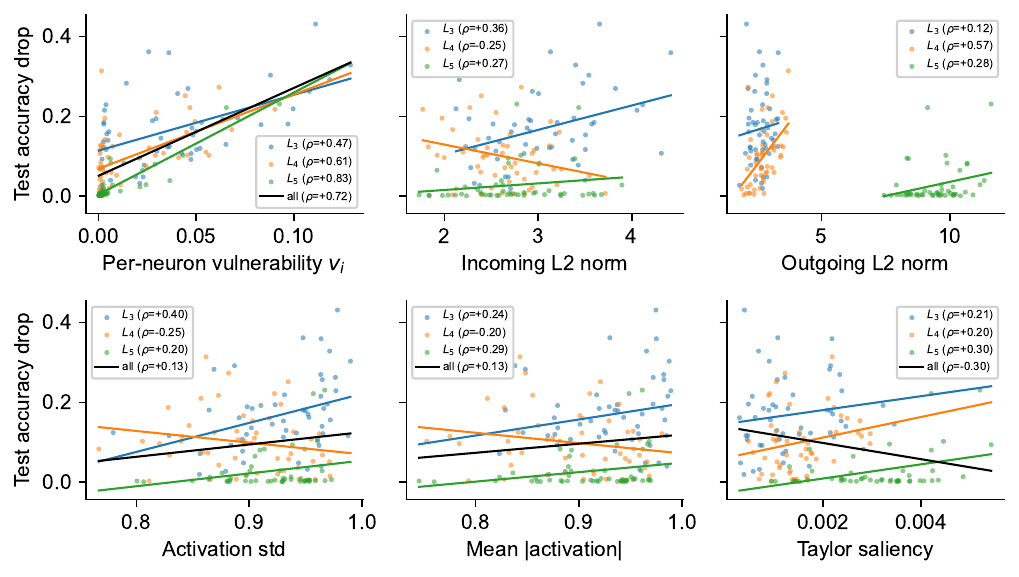}
    \caption{\textbf{Lesion damage versus standard neuron-importance measures.} Single-neuron lesion damage (drop in test-set accuracy, cf.~Figure~\ref{fig:mnist}d) of the fully trained MNIST classifiers, plotted against the per-neuron vulnerability $v_i$ and five standard importance measures, for every neuron of the three five-neuron hidden layers and all $10$ runs. All importance measures are computed on the training set, as is $v_i$. Coloured lines show per-layer least-squares fits as visual guides, with per-layer Spearman rank correlations in the legends; black lines show the pooled fit for measures that are comparable across layers. The raw incoming and outgoing L2 weight norms are not --- the dimensions of the weight vectors depend on the widths of the adjacent layers --- so for these only within-layer fits are shown.}
    \label{fig:importance_comparison}
\end{figure*}

\begin{table*}
\caption{Spearman rank correlations $\rho$ between the importance measures of Figure~\ref{fig:importance_comparison} and the single-neuron lesion damage, within each of the three equal-sized hidden layers and pooled across them ($n=50$ neuron-lesions per layer, $150$ pooled; two-sided $p$-values in parentheses). Only the per-neuron vulnerability $v_i$ correlates significantly within every layer. The pooled values of the raw weight norms and of the Taylor saliency (marked $^\dagger$) are dominated by between-layer differences in scale and are reported for completeness only.}
\label{tab:importance_spearman}

\begin{tabular}{lcccc}
\toprule
measure & $L_3$ & $L_4$ & $L_5$ & pooled\\
\midrule
vulnerability $v_i$ & $+0.47$ $(6\times 10^{-4})$ & $+0.61$ $(<10^{-5})$ & $+0.83$ $(<10^{-12})$ & $+0.72$ $(<10^{-24})$\\
incoming L2 norm    & $+0.36$ $(0.011)$ & $-0.25$ $(0.080)$ & $+0.27$ $(0.063)$ & $+0.19$ $(0.017)^\dagger$\\
outgoing L2 norm    & $+0.12$ $(0.424)$ & $+0.57$ $(<10^{-4})$ & $+0.28$ $(0.049)$ & $-0.40$ $(<10^{-6})^\dagger$\\
activation std      & $+0.40$ $(0.004)$ & $-0.25$ $(0.084)$ & $+0.20$ $(0.165)$ & $+0.13$ $(0.105)$\\
mean $|$activation$|$ & $+0.24$ $(0.092)$ & $-0.20$ $(0.175)$ & $+0.29$ $(0.038)$ & $+0.13$ $(0.103)$\\
Taylor saliency     & $+0.21$ $(0.147)$ & $+0.20$ $(0.174)$ & $+0.30$ $(0.032)$ & $-0.30$ $(2\times 10^{-4})^\dagger$\\
\bottomrule
\end{tabular}
\end{table*}

To put the predictive power of the per-neuron vulnerability $v_i$ (Figure~\ref{fig:mnist}d) into perspective, Figure~\ref{fig:importance_comparison} compares it, on the same lesioning data, against standard heuristic importance measures: the L2 norms of a neuron's incoming and outgoing weight vectors, two activation statistics (the standard deviation and the mean absolute value of the quantised activations), and the first-order Taylor saliency $\mathbb{E}\,\big|a\, \partial \mathcal{L}/\partial a\big|$ commonly used in network pruning.

Only $v_i$ predicts the lesion damage consistently: its rank correlation with the accuracy drop is positive and significant within every layer (Spearman $\rho=+0.47/{+0.61}/{+0.83}$ for $L_3/L_4/L_5$) as well as pooled across layers ($\rho=+0.72$); Table~\ref{tab:importance_spearman} lists all correlation coefficients and $p$-values. None of the alternative measures correlates significantly within all three layers. We report rank correlations since the accuracy bound of Appendix~\ref{app:ablation_bounds} is monotone but nonlinear, and the $v_i$ are strongly right-skewed. The weight norms yield sign-inconsistent within-layer correlations (e.g.~$+0.36/{-0.25}/{+0.27}$ for the incoming norm), and their raw magnitudes are in any case not comparable between layers. The activation statistics carry no consistent signal in this network, since the quantised tanh units are largely saturated and active for almost every sample. The Taylor saliency correlates only weakly within layers ($+0.21/{+0.20}/{+0.30}$), and its pooled correlation is even negative ($-0.30$), illustrating how magnitude-based measures can be dominated by between-layer confounds. In contrast, $v_i$ is a dimensionless share of the system's total information and therefore remains meaningful both within and across layers, in line with the decoder-independent bounds of Appendix~\ref{app:ablation_bounds}.

\section{Robustness of the deep-learning results to the choice of quantisation}
\label{app:quantization_robustness}

\begin{figure*}
    \centering
    \includegraphics[width=.95\textwidth]{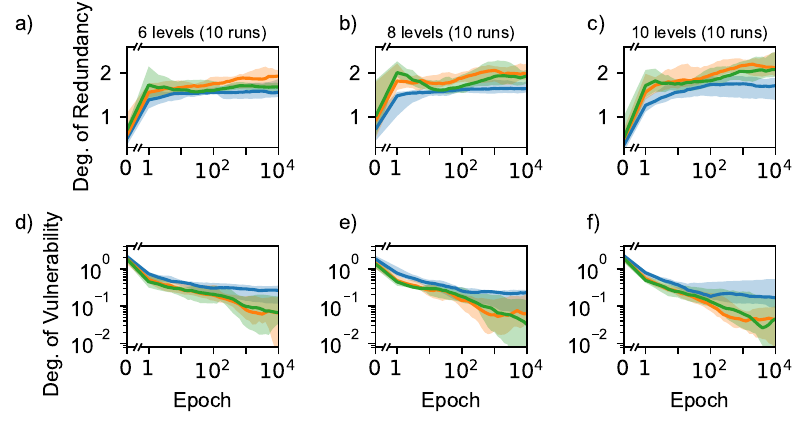}
    \caption{\textbf{Robustness of the MNIST results to the number of quantisation levels.} The analysis of Figure~\ref{fig:mnist}e,\,f repeated for networks trained and measured with six (\textbf{(a)},~\textbf{(d)}), eight (\textbf{(b)},~\textbf{(e)}; the setting of the main text), and ten (\textbf{(c)},~\textbf{(f)}) quantisation levels, with $10$ runs per setting. Lines show the median, shaded regions the maximum and minimum over runs; colours denote the three equal-sized hidden layers as in Figure~\ref{fig:mnist}. The qualitative features of the main-text analysis are preserved at every level count.}
    \label{fig:quantization_robustness}
\end{figure*}

The deep-learning analyses of the main text quantise all activations to eight levels. To verify that the qualitative conclusions do not depend on this choice, we repeated the full MNIST training and analysis pipeline with six and with ten quantisation levels ($10$ runs each; Figure~\ref{fig:quantization_robustness}). The main qualitative findings are preserved across these three quantisation settings.

\section{Shannon observables for conditional mutual information}
\label{app:te}

The Methods section states that Shannon observables extend to transfer entropy
by conditioning all information-theoretic quantities on the target history
$Y_t^{-}$, and presents the resulting expressions
(Eqs.~\eqref{eq:rbar_te}--\eqref{eq:vbar_te}) without proof.  Here we supply
the formal justification by showing that the PID atom-based definitions of the
degree of redundancy and vulnerability (Eqs.~\eqref{eq:def_rbar}
and~\eqref{eq:def_vbar}) extend naturally to the conditional setting, and that
the proofs of
Propositions~\ref{prop:r_bar_shannon_observable}
and~\ref{prop:v_bar_shannon_observable} carry over unchanged.

Although the derivation is formulated in terms of transfer entropy,
the argument applies to any conditional mutual information
$I(\bm X; Y \mid Z)$ and is not specific to time-series data.

\paragraph{Conditional PID atoms.}
Let $Z$ be an arbitrary conditioning variable (in the transfer-entropy
application, $Z = Y_t^{-}$).  For each realisation $z$, fixing
$Z{=}z$ yields a joint distribution $p(\bm X, Y \mid Z{=}z)$ over sources
and target.  This is simply a joint distribution like any other, and one can
therefore apply PID to it, obtaining atoms $\Pi(\alpha \mid Z{=}z)$ that
satisfy the consistency equation~\eqref{eq:consistency_equation} with all
probabilities evaluated under $p(\cdot \mid Z{=}z)$---in particular, all
mutual information terms in~\eqref{eq:consistency_equation} become conditional
mutual informations given $Z{=}z$.  This is not a special
construction: it is the same principle that underlies any conditional
information-theoretic quantity, where one first evaluates the quantity for each
realisation of the conditioning variable and then averages over its
distribution.

Accordingly, the \emph{conditional atoms} are defined by averaging over
the distribution of~$Z$:
\begin{equation}\label{eq:cond_atoms}
    \Pi(\alpha \mid Z)
    \;:=\;
    \mathbb{E}_{p(z)}\!\big[\Pi(\alpha \mid Z{=}z)\big].
\end{equation}
For each realisation $z$, the atoms $\Pi(\alpha \mid Z{=}z)$ decompose the
corresponding mutual information: $I(\bm X; Y \mid Z{=}z) = \sum_{\alpha} \Pi(\alpha \mid Z{=}z)$.
Averaging both sides over $p(z)$ gives
\begin{equation}\label{eq:cond_atom_decomp}
    I(\bm X; Y \mid Z)
    \;=\;
    \sum_{\alpha \in \mathcal{A}_n} \Pi(\alpha \mid Z).
\end{equation}
By the same reasoning, the consistency equation carries over to its conditional
form: for any $\mathbf{a} \subseteq [n]$,
\begin{equation}\label{eq:cond_consistency}
    I(\mathbf{a}; Y \mid Z)
    \;=\;
    \sum_{\exists\,\mathbf{b}\subseteq\mathbf{a}:\,\mathbf{b}\in\alpha}
    \Pi(\alpha \mid Z),
\end{equation}
since the identity holds for each realisation $z$ and averaging preserves it.

\paragraph{Conditional degrees of redundancy and vulnerability.}
With the conditional atoms in hand, the conditional degrees are defined exactly
as in the unconditional case
(Eqs.~\eqref{eq:def_rbar} and~\eqref{eq:def_vbar}), replacing every
$\Pi(\alpha)$ by $\Pi(\alpha \mid Z)$ and every $I(\bm X;Y)$ by
$I(\bm X;Y \mid Z)$:
\begin{align}
    \bar{r}(\bm X; Y \mid Z)
    &:= \sum_{k=0}^n k \sum_{r(\alpha)=k}
    \frac{\Pi(\alpha \mid Z)}{I(\bm X;Y \mid Z)},
    \label{eq:cond_rbar_def}\\[4pt]
    \bar{v}(\bm X; Y \mid Z)
    &:= \sum_{j=0}^n j \sum_{v(\alpha)=j}
    \frac{\Pi(\alpha \mid Z)}{I(\bm X;Y \mid Z)}.
    \label{eq:cond_vbar_def}
\end{align}

\paragraph{Conditional Shannon-observable property.}
Because the conditional atoms satisfy the conditional consistency
equation~\eqref{eq:cond_consistency}, which has the same structure as the
unconditional one with all mutual informations replaced by conditional mutual
informations given $Z$, the proofs of
Propositions~\ref{prop:r_bar_shannon_observable}
and~\ref{prop:v_bar_shannon_observable} apply verbatim with $I(\cdot;\cdot)$
replaced by $I(\cdot;\cdot \mid Z)$ throughout.  The result is
\begin{align}
    \bar{r}(\bm X; Y \mid Z)
    &= \frac{\sum_{i=1}^{n} I(X_i; Y \mid Z)}{I(\bm X; Y \mid Z)},
    \label{eq:cond_rbar_shannon}\\[4pt]
    \bar{v}(\bm X; Y \mid Z)
    &= \frac{\sum_{j=1}^{n} I(X_j; Y \mid Z, \bm X_{-j})}{I(\bm X; Y \mid Z)}.
    \label{eq:cond_vbar_shannon}
\end{align}
Setting $Z = Y_t^{-}$, $\bm X = \bm X_t^{-}$, and $Y = Y_{t+1}$ recovers
Eqs.~\eqref{eq:rbar_te}--\eqref{eq:vbar_te}.

\bibliography{references}

\begin{thebibliography}{107}%
\makeatletter
\providecommand \@ifxundefined [1]{%
 \@ifx{#1\undefined}
}%
\providecommand \@ifnum [1]{%
 \ifnum #1\expandafter \@firstoftwo
 \else \expandafter \@secondoftwo
 \fi
}%
\providecommand \@ifx [1]{%
 \ifx #1\expandafter \@firstoftwo
 \else \expandafter \@secondoftwo
 \fi
}%
\providecommand \natexlab [1]{#1}%
\providecommand \enquote  [1]{``#1''}%
\providecommand \bibnamefont  [1]{#1}%
\providecommand \bibfnamefont [1]{#1}%
\providecommand \citenamefont [1]{#1}%
\providecommand \href@noop [0]{\@secondoftwo}%
\providecommand \href [0]{\begingroup \@sanitize@url \@href}%
\providecommand \@href[1]{\@@startlink{#1}\@@href}%
\providecommand \@@href[1]{\endgroup#1\@@endlink}%
\providecommand \@sanitize@url [0]{\catcode `\\12\catcode `\$12\catcode
  `\&12\catcode `\#12\catcode `\^12\catcode `\_12\catcode `\%12\relax}%
\providecommand \@@startlink[1]{}%
\providecommand \@@endlink[0]{}%
\providecommand \url  [0]{\begingroup\@sanitize@url \@url }%
\providecommand \@url [1]{\endgroup\@href {#1}{\urlprefix }}%
\providecommand \urlprefix  [0]{URL }%
\providecommand \Eprint [0]{\href }%
\providecommand \doibase [0]{http://dx.doi.org/}%
\providecommand \selectlanguage [0]{\@gobble}%
\providecommand \bibinfo  [0]{\@secondoftwo}%
\providecommand \bibfield  [0]{\@secondoftwo}%
\providecommand \translation [1]{[#1]}%
\providecommand \BibitemOpen [0]{}%
\providecommand \bibitemStop [0]{}%
\providecommand \bibitemNoStop [0]{.\EOS\space}%
\providecommand \EOS [0]{\spacefactor3000\relax}%
\providecommand \BibitemShut  [1]{\csname bibitem#1\endcsname}%
\let\auto@bib@innerbib\@empty
\bibitem [{\citenamefont {Waldrop}(1993)}]{waldrop1993complexity}%
  \BibitemOpen
  \bibfield  {author} {\bibinfo {author} {\bibfnamefont {M.~M.}\ \bibnamefont
  {Waldrop}},\ }\href@noop {} {\emph {\bibinfo {title} {Complexity: The
  emerging science at the edge of order and chaos}}}\ (\bibinfo  {publisher}
  {Simon and Schuster},\ \bibinfo {year} {1993})\BibitemShut {NoStop}%
\bibitem [{\citenamefont {Gleick}(2011)}]{gleick2011information}%
  \BibitemOpen
  \bibfield  {author} {\bibinfo {author} {\bibfnamefont {J.}~\bibnamefont
  {Gleick}},\ }\href@noop {} {\emph {\bibinfo {title} {The information: A
  history, a theory, a flood}}}\ (\bibinfo  {publisher} {Vintage},\ \bibinfo
  {year} {2011})\BibitemShut {NoStop}%
\bibitem [{\citenamefont {Bengio}\ and\ \citenamefont
  {Malkin}(2024)}]{bengio2024machine}%
  \BibitemOpen
  \bibfield  {author} {\bibinfo {author} {\bibfnamefont {Y.}~\bibnamefont
  {Bengio}}\ and\ \bibinfo {author} {\bibfnamefont {N.}~\bibnamefont
  {Malkin}},\ }\href@noop {} {\bibfield  {journal} {\bibinfo  {journal}
  {Bulletin of the American Mathematical Society}\ }\textbf {\bibinfo {volume}
  {61}},\ \bibinfo {pages} {457} (\bibinfo {year} {2024})}\BibitemShut
  {NoStop}%
\bibitem [{\citenamefont {Rajpal}\ \emph {et~al.}(2025)\citenamefont {Rajpal},
  \citenamefont {Stengel}, \citenamefont {Mediano}, \citenamefont {Rosas},
  \citenamefont {Viegas}, \citenamefont {Marquet},\ and\ \citenamefont
  {Jensen}}]{rajpal2025information}%
  \BibitemOpen
  \bibfield  {author} {\bibinfo {author} {\bibfnamefont {H.}~\bibnamefont
  {Rajpal}}, \bibinfo {author} {\bibfnamefont {C.~v.}\ \bibnamefont {Stengel}},
  \bibinfo {author} {\bibfnamefont {P.~A.}\ \bibnamefont {Mediano}}, \bibinfo
  {author} {\bibfnamefont {F.~E.}\ \bibnamefont {Rosas}}, \bibinfo {author}
  {\bibfnamefont {E.}~\bibnamefont {Viegas}}, \bibinfo {author} {\bibfnamefont
  {P.~A.}\ \bibnamefont {Marquet}}, \ and\ \bibinfo {author} {\bibfnamefont
  {H.~J.}\ \bibnamefont {Jensen}},\ }\href@noop {} {\bibfield  {journal}
  {\bibinfo  {journal} {Communications Biology}\ }\textbf {\bibinfo {volume}
  {8}},\ \bibinfo {pages} {1231} (\bibinfo {year} {2025})}\BibitemShut
  {NoStop}%
\bibitem [{\citenamefont {Cang}\ and\ \citenamefont
  {Nie}(2020)}]{cang2020inferring}%
  \BibitemOpen
  \bibfield  {author} {\bibinfo {author} {\bibfnamefont {Z.}~\bibnamefont
  {Cang}}\ and\ \bibinfo {author} {\bibfnamefont {Q.}~\bibnamefont {Nie}},\
  }\href@noop {} {\bibfield  {journal} {\bibinfo  {journal} {Nature
  Communications}\ }\textbf {\bibinfo {volume} {11}},\ \bibinfo {pages} {2084}
  (\bibinfo {year} {2020})}\BibitemShut {NoStop}%
\bibitem [{\citenamefont {Park}\ \emph {et~al.}(2021)\citenamefont {Park},
  \citenamefont {Supek},\ and\ \citenamefont {Lehner}}]{park2021higher}%
  \BibitemOpen
  \bibfield  {author} {\bibinfo {author} {\bibfnamefont {S.}~\bibnamefont
  {Park}}, \bibinfo {author} {\bibfnamefont {F.}~\bibnamefont {Supek}}, \ and\
  \bibinfo {author} {\bibfnamefont {B.}~\bibnamefont {Lehner}},\ }\href@noop {}
  {\bibfield  {journal} {\bibinfo  {journal} {Nature communications}\ }\textbf
  {\bibinfo {volume} {12}},\ \bibinfo {pages} {7051} (\bibinfo {year}
  {2021})}\BibitemShut {NoStop}%
\bibitem [{\citenamefont {Marinazzo}\ \emph {et~al.}(2022)\citenamefont
  {Marinazzo}, \citenamefont {Van~Roozendaal}, \citenamefont {Rosas},
  \citenamefont {Stella}, \citenamefont {Comolatti}, \citenamefont {Colenbier},
  \citenamefont {Stramaglia},\ and\ \citenamefont
  {Rosseel}}]{marinazzo2022information}%
  \BibitemOpen
  \bibfield  {author} {\bibinfo {author} {\bibfnamefont {D.}~\bibnamefont
  {Marinazzo}}, \bibinfo {author} {\bibfnamefont {J.}~\bibnamefont
  {Van~Roozendaal}}, \bibinfo {author} {\bibfnamefont {F.~E.}\ \bibnamefont
  {Rosas}}, \bibinfo {author} {\bibfnamefont {M.}~\bibnamefont {Stella}},
  \bibinfo {author} {\bibfnamefont {R.}~\bibnamefont {Comolatti}}, \bibinfo
  {author} {\bibfnamefont {N.}~\bibnamefont {Colenbier}}, \bibinfo {author}
  {\bibfnamefont {S.}~\bibnamefont {Stramaglia}}, \ and\ \bibinfo {author}
  {\bibfnamefont {Y.}~\bibnamefont {Rosseel}},\ }\href@noop {} {\bibfield
  {journal} {\bibinfo  {journal} {arXiv preprint arXiv:2205.01035}\ } (\bibinfo
  {year} {2022})}\BibitemShut {NoStop}%
\bibitem [{\citenamefont {Varley}\ and\ \citenamefont
  {Kaminski}(2022)}]{varley2022untangling}%
  \BibitemOpen
  \bibfield  {author} {\bibinfo {author} {\bibfnamefont {T.~F.}\ \bibnamefont
  {Varley}}\ and\ \bibinfo {author} {\bibfnamefont {P.}~\bibnamefont
  {Kaminski}},\ }\href@noop {} {\bibfield  {journal} {\bibinfo  {journal}
  {Entropy}\ }\textbf {\bibinfo {volume} {24}},\ \bibinfo {pages} {1387}
  (\bibinfo {year} {2022})}\BibitemShut {NoStop}%
\bibitem [{\citenamefont {Rosas}\ \emph {et~al.}(2018)\citenamefont {Rosas},
  \citenamefont {Mediano}, \citenamefont {Ugarte},\ and\ \citenamefont
  {Jensen}}]{rosas2018information}%
  \BibitemOpen
  \bibfield  {author} {\bibinfo {author} {\bibfnamefont {F.}~\bibnamefont
  {Rosas}}, \bibinfo {author} {\bibfnamefont {P.~A.}\ \bibnamefont {Mediano}},
  \bibinfo {author} {\bibfnamefont {M.}~\bibnamefont {Ugarte}}, \ and\ \bibinfo
  {author} {\bibfnamefont {H.~J.}\ \bibnamefont {Jensen}},\ }\href@noop {}
  {\bibfield  {journal} {\bibinfo  {journal} {Entropy}\ }\textbf {\bibinfo
  {volume} {20}},\ \bibinfo {pages} {793} (\bibinfo {year} {2018})}\BibitemShut
  {NoStop}%
\bibitem [{\citenamefont {Orio}\ \emph {et~al.}(2023)\citenamefont {Orio},
  \citenamefont {Mediano},\ and\ \citenamefont {Rosas}}]{orio2023dynamical}%
  \BibitemOpen
  \bibfield  {author} {\bibinfo {author} {\bibfnamefont {P.}~\bibnamefont
  {Orio}}, \bibinfo {author} {\bibfnamefont {P.~A.}\ \bibnamefont {Mediano}}, \
  and\ \bibinfo {author} {\bibfnamefont {F.~E.}\ \bibnamefont {Rosas}},\
  }\href@noop {} {\bibfield  {journal} {\bibinfo  {journal} {Chaos: An
  Interdisciplinary Journal of Nonlinear Science}\ }\textbf {\bibinfo {volume}
  {33}} (\bibinfo {year} {2023})}\BibitemShut {NoStop}%
\bibitem [{\citenamefont {Scagliarini}\ \emph {et~al.}(2022)\citenamefont
  {Scagliarini}, \citenamefont {Marinazzo}, \citenamefont {Guo}, \citenamefont
  {Stramaglia},\ and\ \citenamefont {Rosas}}]{scagliarini2022quantifying}%
  \BibitemOpen
  \bibfield  {author} {\bibinfo {author} {\bibfnamefont {T.}~\bibnamefont
  {Scagliarini}}, \bibinfo {author} {\bibfnamefont {D.}~\bibnamefont
  {Marinazzo}}, \bibinfo {author} {\bibfnamefont {Y.}~\bibnamefont {Guo}},
  \bibinfo {author} {\bibfnamefont {S.}~\bibnamefont {Stramaglia}}, \ and\
  \bibinfo {author} {\bibfnamefont {F.~E.}\ \bibnamefont {Rosas}},\ }\href@noop
  {} {\bibfield  {journal} {\bibinfo  {journal} {Physical Review Research}\
  }\textbf {\bibinfo {volume} {4}},\ \bibinfo {pages} {013184} (\bibinfo {year}
  {2022})}\BibitemShut {NoStop}%
\bibitem [{\citenamefont {Gatica}\ \emph {et~al.}(2021)\citenamefont {Gatica},
  \citenamefont {Cofr{\'e}}, \citenamefont {Mediano}, \citenamefont {Rosas},
  \citenamefont {Orio}, \citenamefont {Diez}, \citenamefont {Swinnen},\ and\
  \citenamefont {Cortes}}]{gatica2021high}%
  \BibitemOpen
  \bibfield  {author} {\bibinfo {author} {\bibfnamefont {M.}~\bibnamefont
  {Gatica}}, \bibinfo {author} {\bibfnamefont {R.}~\bibnamefont {Cofr{\'e}}},
  \bibinfo {author} {\bibfnamefont {P.~A.}\ \bibnamefont {Mediano}}, \bibinfo
  {author} {\bibfnamefont {F.~E.}\ \bibnamefont {Rosas}}, \bibinfo {author}
  {\bibfnamefont {P.}~\bibnamefont {Orio}}, \bibinfo {author} {\bibfnamefont
  {I.}~\bibnamefont {Diez}}, \bibinfo {author} {\bibfnamefont {S.~P.}\
  \bibnamefont {Swinnen}}, \ and\ \bibinfo {author} {\bibfnamefont {J.~M.}\
  \bibnamefont {Cortes}},\ }\href@noop {} {\bibfield  {journal} {\bibinfo
  {journal} {Brain Connectivity}\ }\textbf {\bibinfo {volume} {11}},\ \bibinfo
  {pages} {734} (\bibinfo {year} {2021})}\BibitemShut {NoStop}%
\bibitem [{\citenamefont {Luppi}\ \emph {et~al.}(2022)\citenamefont {Luppi},
  \citenamefont {Mediano}, \citenamefont {Rosas}, \citenamefont {Holland},
  \citenamefont {Fryer}, \citenamefont {O’Brien}, \citenamefont {Rowe},
  \citenamefont {Menon}, \citenamefont {Bor},\ and\ \citenamefont
  {Stamatakis}}]{luppi2022synergistic}%
  \BibitemOpen
  \bibfield  {author} {\bibinfo {author} {\bibfnamefont {A.~I.}\ \bibnamefont
  {Luppi}}, \bibinfo {author} {\bibfnamefont {P.~A.}\ \bibnamefont {Mediano}},
  \bibinfo {author} {\bibfnamefont {F.~E.}\ \bibnamefont {Rosas}}, \bibinfo
  {author} {\bibfnamefont {N.}~\bibnamefont {Holland}}, \bibinfo {author}
  {\bibfnamefont {T.~D.}\ \bibnamefont {Fryer}}, \bibinfo {author}
  {\bibfnamefont {J.~T.}\ \bibnamefont {O’Brien}}, \bibinfo {author}
  {\bibfnamefont {J.~B.}\ \bibnamefont {Rowe}}, \bibinfo {author}
  {\bibfnamefont {D.~K.}\ \bibnamefont {Menon}}, \bibinfo {author}
  {\bibfnamefont {D.}~\bibnamefont {Bor}}, \ and\ \bibinfo {author}
  {\bibfnamefont {E.~A.}\ \bibnamefont {Stamatakis}},\ }\href@noop {}
  {\bibfield  {journal} {\bibinfo  {journal} {Nature Neuroscience}\ }\textbf
  {\bibinfo {volume} {25}},\ \bibinfo {pages} {771} (\bibinfo {year}
  {2022})}\BibitemShut {NoStop}%
\bibitem [{\citenamefont {Varley}\ \emph
  {et~al.}(2023{\natexlab{a}})\citenamefont {Varley}, \citenamefont {Sporns},
  \citenamefont {Schaffelhofer}, \citenamefont {Scherberger},\ and\
  \citenamefont {Dann}}]{varley2023information}%
  \BibitemOpen
  \bibfield  {author} {\bibinfo {author} {\bibfnamefont {T.~F.}\ \bibnamefont
  {Varley}}, \bibinfo {author} {\bibfnamefont {O.}~\bibnamefont {Sporns}},
  \bibinfo {author} {\bibfnamefont {S.}~\bibnamefont {Schaffelhofer}}, \bibinfo
  {author} {\bibfnamefont {H.}~\bibnamefont {Scherberger}}, \ and\ \bibinfo
  {author} {\bibfnamefont {B.}~\bibnamefont {Dann}},\ }\href@noop {} {\bibfield
   {journal} {\bibinfo  {journal} {Proceedings of the National Academy of
  Sciences}\ }\textbf {\bibinfo {volume} {120}},\ \bibinfo {pages}
  {e2207677120} (\bibinfo {year} {2023}{\natexlab{a}})}\BibitemShut {NoStop}%
\bibitem [{\citenamefont {Varley}\ \emph
  {et~al.}(2023{\natexlab{b}})\citenamefont {Varley}, \citenamefont {Pope},
  \citenamefont {Faskowitz},\ and\ \citenamefont
  {Sporns}}]{varley2023multivariate}%
  \BibitemOpen
  \bibfield  {author} {\bibinfo {author} {\bibfnamefont {T.~F.}\ \bibnamefont
  {Varley}}, \bibinfo {author} {\bibfnamefont {M.}~\bibnamefont {Pope}},
  \bibinfo {author} {\bibfnamefont {J.}~\bibnamefont {Faskowitz}}, \ and\
  \bibinfo {author} {\bibfnamefont {O.}~\bibnamefont {Sporns}},\ }\href@noop {}
  {\bibfield  {journal} {\bibinfo  {journal} {Communications Biology}\ }\textbf
  {\bibinfo {volume} {6}},\ \bibinfo {pages} {451} (\bibinfo {year}
  {2023}{\natexlab{b}})}\BibitemShut {NoStop}%
\bibitem [{\citenamefont {Herzog}\ \emph {et~al.}(2024)\citenamefont {Herzog},
  \citenamefont {Barbey}, \citenamefont {Islam}, \citenamefont {Rueda-Delgado},
  \citenamefont {Nolan}, \citenamefont {Prado}, \citenamefont {Krylova},
  \citenamefont {Izyurov}, \citenamefont {Javaheripour}, \citenamefont
  {Danyeli} \emph {et~al.}}]{herzog2024high}%
  \BibitemOpen
  \bibfield  {author} {\bibinfo {author} {\bibfnamefont {R.}~\bibnamefont
  {Herzog}}, \bibinfo {author} {\bibfnamefont {F.~M.}\ \bibnamefont {Barbey}},
  \bibinfo {author} {\bibfnamefont {M.~N.}\ \bibnamefont {Islam}}, \bibinfo
  {author} {\bibfnamefont {L.}~\bibnamefont {Rueda-Delgado}}, \bibinfo {author}
  {\bibfnamefont {H.}~\bibnamefont {Nolan}}, \bibinfo {author} {\bibfnamefont
  {P.}~\bibnamefont {Prado}}, \bibinfo {author} {\bibfnamefont
  {M.}~\bibnamefont {Krylova}}, \bibinfo {author} {\bibfnamefont
  {I.}~\bibnamefont {Izyurov}}, \bibinfo {author} {\bibfnamefont
  {N.}~\bibnamefont {Javaheripour}}, \bibinfo {author} {\bibfnamefont {L.~V.}\
  \bibnamefont {Danyeli}},  \emph {et~al.},\ }\href@noop {} {\bibfield
  {journal} {\bibinfo  {journal} {Translational Psychiatry}\ }\textbf {\bibinfo
  {volume} {14}},\ \bibinfo {pages} {310} (\bibinfo {year} {2024})}\BibitemShut
  {NoStop}%
\bibitem [{\citenamefont {Pope}\ \emph {et~al.}(2024)\citenamefont {Pope},
  \citenamefont {Varley},\ and\ \citenamefont {Sporns}}]{pope2024time}%
  \BibitemOpen
  \bibfield  {author} {\bibinfo {author} {\bibfnamefont {M.}~\bibnamefont
  {Pope}}, \bibinfo {author} {\bibfnamefont {T.~F.}\ \bibnamefont {Varley}}, \
  and\ \bibinfo {author} {\bibfnamefont {O.}~\bibnamefont {Sporns}},\
  }\href@noop {} {\bibfield  {journal} {\bibinfo  {journal} {bioRxiv}\ ,\
  \bibinfo {pages} {2024}} (\bibinfo {year} {2024})}\BibitemShut {NoStop}%
\bibitem [{\citenamefont {Varley}\ \emph {et~al.}(2024)\citenamefont {Varley},
  \citenamefont {Sporns}, \citenamefont {Stevenson}, \citenamefont {Welch},
  \citenamefont {Myers}, \citenamefont {Vanhatalo},\ and\ \citenamefont
  {Tokariev}}]{varley2024emergence}%
  \BibitemOpen
  \bibfield  {author} {\bibinfo {author} {\bibfnamefont {T.~F.}\ \bibnamefont
  {Varley}}, \bibinfo {author} {\bibfnamefont {O.}~\bibnamefont {Sporns}},
  \bibinfo {author} {\bibfnamefont {N.~J.}\ \bibnamefont {Stevenson}}, \bibinfo
  {author} {\bibfnamefont {M.~G.}\ \bibnamefont {Welch}}, \bibinfo {author}
  {\bibfnamefont {M.~M.}\ \bibnamefont {Myers}}, \bibinfo {author}
  {\bibfnamefont {S.}~\bibnamefont {Vanhatalo}}, \ and\ \bibinfo {author}
  {\bibfnamefont {A.}~\bibnamefont {Tokariev}},\ }\href@noop {} {\bibfield
  {journal} {\bibinfo  {journal} {bioRxiv}\ ,\ \bibinfo {pages} {2024}}
  (\bibinfo {year} {2024})}\BibitemShut {NoStop}%
\bibitem [{\citenamefont {Tax}\ \emph {et~al.}(2017)\citenamefont {Tax},
  \citenamefont {Mediano},\ and\ \citenamefont {Shanahan}}]{tax2017partial}%
  \BibitemOpen
  \bibfield  {author} {\bibinfo {author} {\bibfnamefont {T.~M.}\ \bibnamefont
  {Tax}}, \bibinfo {author} {\bibfnamefont {P.~A.}\ \bibnamefont {Mediano}}, \
  and\ \bibinfo {author} {\bibfnamefont {M.}~\bibnamefont {Shanahan}},\
  }\href@noop {} {\bibfield  {journal} {\bibinfo  {journal} {Entropy}\ }\textbf
  {\bibinfo {volume} {19}},\ \bibinfo {pages} {474} (\bibinfo {year}
  {2017})}\BibitemShut {NoStop}%
\bibitem [{\citenamefont {Proca}\ \emph {et~al.}(2024)\citenamefont {Proca},
  \citenamefont {Rosas}, \citenamefont {Luppi}, \citenamefont {Bor},
  \citenamefont {Crosby},\ and\ \citenamefont
  {Mediano}}]{proca2024synergistic}%
  \BibitemOpen
  \bibfield  {author} {\bibinfo {author} {\bibfnamefont {A.~M.}\ \bibnamefont
  {Proca}}, \bibinfo {author} {\bibfnamefont {F.~E.}\ \bibnamefont {Rosas}},
  \bibinfo {author} {\bibfnamefont {A.~I.}\ \bibnamefont {Luppi}}, \bibinfo
  {author} {\bibfnamefont {D.}~\bibnamefont {Bor}}, \bibinfo {author}
  {\bibfnamefont {M.}~\bibnamefont {Crosby}}, \ and\ \bibinfo {author}
  {\bibfnamefont {P.~A.}\ \bibnamefont {Mediano}},\ }\href@noop {} {\bibfield
  {journal} {\bibinfo  {journal} {PLoS computational biology}\ }\textbf
  {\bibinfo {volume} {20}},\ \bibinfo {pages} {e1012178} (\bibinfo {year}
  {2024})}\BibitemShut {NoStop}%
\bibitem [{\citenamefont {McSharry}\ \emph {et~al.}(2024)\citenamefont
  {McSharry}, \citenamefont {Kaplanis}, \citenamefont {Rosas},\ and\
  \citenamefont {Mediano}}]{kaplanis2023learning}%
  \BibitemOpen
  \bibfield  {author} {\bibinfo {author} {\bibfnamefont {D.}~\bibnamefont
  {McSharry}}, \bibinfo {author} {\bibfnamefont {C.}~\bibnamefont {Kaplanis}},
  \bibinfo {author} {\bibfnamefont {F.~E.}\ \bibnamefont {Rosas}}, \ and\
  \bibinfo {author} {\bibfnamefont {P.~A.}\ \bibnamefont {Mediano}},\ }in\
  \href {\doibase 10.52202/079017-3798} {\emph {\bibinfo {booktitle} {Advances
  in Neural Information Processing Systems}}},\ Vol.~\bibinfo {volume} {37},\
  \bibinfo {editor} {edited by\ \bibinfo {editor} {\bibfnamefont
  {A.}~\bibnamefont {Globerson}}, \bibinfo {editor} {\bibfnamefont
  {L.}~\bibnamefont {Mackey}}, \bibinfo {editor} {\bibfnamefont
  {D.}~\bibnamefont {Belgrave}}, \bibinfo {editor} {\bibfnamefont
  {A.}~\bibnamefont {Fan}}, \bibinfo {editor} {\bibfnamefont {U.}~\bibnamefont
  {Paquet}}, \bibinfo {editor} {\bibfnamefont {J.}~\bibnamefont {Tomczak}}, \
  and\ \bibinfo {editor} {\bibfnamefont {C.}~\bibnamefont {Zhang}}}\ (\bibinfo
  {publisher} {Curran Associates, Inc.},\ \bibinfo {year} {2024})\ pp.\
  \bibinfo {pages} {119547--119572}\BibitemShut {NoStop}%
\bibitem [{\citenamefont {Kong}\ \emph {et~al.}(2023)\citenamefont {Kong},
  \citenamefont {Liu}, \citenamefont {Li}, \citenamefont {Yogatama},\ and\
  \citenamefont {Steeg}}]{kong2023interpretable}%
  \BibitemOpen
  \bibfield  {author} {\bibinfo {author} {\bibfnamefont {X.}~\bibnamefont
  {Kong}}, \bibinfo {author} {\bibfnamefont {O.}~\bibnamefont {Liu}}, \bibinfo
  {author} {\bibfnamefont {H.}~\bibnamefont {Li}}, \bibinfo {author}
  {\bibfnamefont {D.}~\bibnamefont {Yogatama}}, \ and\ \bibinfo {author}
  {\bibfnamefont {G.~V.}\ \bibnamefont {Steeg}},\ }\href@noop {} {\bibfield
  {journal} {\bibinfo  {journal} {arXiv preprint arXiv:2310.07972}\ } (\bibinfo
  {year} {2023})}\BibitemShut {NoStop}%
\bibitem [{\citenamefont {Luppi}\ \emph {et~al.}(2024)\citenamefont {Luppi},
  \citenamefont {Rosas}, \citenamefont {Mediano}, \citenamefont {Menon},\ and\
  \citenamefont {Stamatakis}}]{luppi2024information}%
  \BibitemOpen
  \bibfield  {author} {\bibinfo {author} {\bibfnamefont {A.~I.}\ \bibnamefont
  {Luppi}}, \bibinfo {author} {\bibfnamefont {F.~E.}\ \bibnamefont {Rosas}},
  \bibinfo {author} {\bibfnamefont {P.~A.}\ \bibnamefont {Mediano}}, \bibinfo
  {author} {\bibfnamefont {D.~K.}\ \bibnamefont {Menon}}, \ and\ \bibinfo
  {author} {\bibfnamefont {E.~A.}\ \bibnamefont {Stamatakis}},\ }\href@noop {}
  {\bibfield  {journal} {\bibinfo  {journal} {Trends in Cognitive Sciences}\ }
  (\bibinfo {year} {2024})}\BibitemShut {NoStop}%
\bibitem [{\citenamefont {Williams}\ and\ \citenamefont
  {Beer}(2010)}]{williams2010nonnegative}%
  \BibitemOpen
  \bibfield  {author} {\bibinfo {author} {\bibfnamefont {P.~L.}\ \bibnamefont
  {Williams}}\ and\ \bibinfo {author} {\bibfnamefont {R.~D.}\ \bibnamefont
  {Beer}},\ }\href@noop {} {\bibfield  {journal} {\bibinfo  {journal} {arXiv
  preprint arXiv:1004.2515}\ } (\bibinfo {year} {2010})}\BibitemShut {NoStop}%
\bibitem [{\citenamefont {Lizier}\ \emph {et~al.}(2018)\citenamefont {Lizier},
  \citenamefont {Bertschinger}, \citenamefont {Jost},\ and\ \citenamefont
  {Wibral}}]{lizier2018information}%
  \BibitemOpen
  \bibfield  {author} {\bibinfo {author} {\bibfnamefont {J.~T.}\ \bibnamefont
  {Lizier}}, \bibinfo {author} {\bibfnamefont {N.}~\bibnamefont
  {Bertschinger}}, \bibinfo {author} {\bibfnamefont {J.}~\bibnamefont {Jost}},
  \ and\ \bibinfo {author} {\bibfnamefont {M.}~\bibnamefont {Wibral}},\
  }\href@noop {} {\enquote {\bibinfo {title} {Information decomposition of
  target effects from multi-source interactions: Perspectives on previous,
  current and future work},}\ } (\bibinfo {year} {2018})\BibitemShut {NoStop}%
\bibitem [{\citenamefont {Mediano}\ \emph {et~al.}(2025)\citenamefont
  {Mediano}, \citenamefont {Rosas}, \citenamefont {Luppi}, \citenamefont
  {Carhart-Harris}, \citenamefont {Bor}, \citenamefont {Seth},\ and\
  \citenamefont {Barrett}}]{mediano2025toward}%
  \BibitemOpen
  \bibfield  {author} {\bibinfo {author} {\bibfnamefont {P.~A.}\ \bibnamefont
  {Mediano}}, \bibinfo {author} {\bibfnamefont {F.~E.}\ \bibnamefont {Rosas}},
  \bibinfo {author} {\bibfnamefont {A.~I.}\ \bibnamefont {Luppi}}, \bibinfo
  {author} {\bibfnamefont {R.~L.}\ \bibnamefont {Carhart-Harris}}, \bibinfo
  {author} {\bibfnamefont {D.}~\bibnamefont {Bor}}, \bibinfo {author}
  {\bibfnamefont {A.~K.}\ \bibnamefont {Seth}}, \ and\ \bibinfo {author}
  {\bibfnamefont {A.~B.}\ \bibnamefont {Barrett}},\ }\href@noop {} {\bibfield
  {journal} {\bibinfo  {journal} {Proceedings of the National Academy of
  Sciences}\ }\textbf {\bibinfo {volume} {122}},\ \bibinfo {pages}
  {e2423297122} (\bibinfo {year} {2025})}\BibitemShut {NoStop}%
\bibitem [{\citenamefont {Gutknecht}\ \emph {et~al.}(2021)\citenamefont
  {Gutknecht}, \citenamefont {Wibral},\ and\ \citenamefont
  {Makkeh}}]{gutknecht2021bits}%
  \BibitemOpen
  \bibfield  {author} {\bibinfo {author} {\bibfnamefont {A.~J.}\ \bibnamefont
  {Gutknecht}}, \bibinfo {author} {\bibfnamefont {M.}~\bibnamefont {Wibral}}, \
  and\ \bibinfo {author} {\bibfnamefont {A.}~\bibnamefont {Makkeh}},\ }\href
  {\doibase 10.1098/rspa.2021.0110} {\bibfield  {journal} {\bibinfo  {journal}
  {Proceedings of the Royal Society A: Mathematical, Physical and Engineering
  Sciences}\ }\textbf {\bibinfo {volume} {477}},\ \bibinfo {pages} {20210110}
  (\bibinfo {year} {2021})}\BibitemShut {NoStop}%
\bibitem [{\citenamefont {Lizier}\ \emph {et~al.}(2013)\citenamefont {Lizier},
  \citenamefont {Flecker},\ and\ \citenamefont {Williams}}]{lizier2013towards}%
  \BibitemOpen
  \bibfield  {author} {\bibinfo {author} {\bibfnamefont {J.~T.}\ \bibnamefont
  {Lizier}}, \bibinfo {author} {\bibfnamefont {B.}~\bibnamefont {Flecker}}, \
  and\ \bibinfo {author} {\bibfnamefont {P.~L.}\ \bibnamefont {Williams}},\
  }in\ \href@noop {} {\emph {\bibinfo {booktitle} {2013 IEEE Symposium on
  Artificial Life (ALIFE)}}}\ (\bibinfo {organization} {IEEE},\ \bibinfo {year}
  {2013})\ pp.\ \bibinfo {pages} {43--51}\BibitemShut {NoStop}%
\bibitem [{\citenamefont {Ehrlich}\ \emph {et~al.}(2023)\citenamefont
  {Ehrlich}, \citenamefont {Schneider}, \citenamefont {Priesemann},
  \citenamefont {Wibral},\ and\ \citenamefont {Makkeh}}]{ehrlich2023a}%
  \BibitemOpen
  \bibfield  {author} {\bibinfo {author} {\bibfnamefont {D.~A.}\ \bibnamefont
  {Ehrlich}}, \bibinfo {author} {\bibfnamefont {A.~C.}\ \bibnamefont
  {Schneider}}, \bibinfo {author} {\bibfnamefont {V.}~\bibnamefont
  {Priesemann}}, \bibinfo {author} {\bibfnamefont {M.}~\bibnamefont {Wibral}},
  \ and\ \bibinfo {author} {\bibfnamefont {A.}~\bibnamefont {Makkeh}},\ }\href
  {https://openreview.net/forum?id=R8TU3pfzFr} {\bibfield  {journal} {\bibinfo
  {journal} {Transactions on Machine Learning Research}\ } (\bibinfo {year}
  {2023})}\BibitemShut {NoStop}%
\bibitem [{\citenamefont {Schneider}\ \emph {et~al.}(2025)\citenamefont
  {Schneider}, \citenamefont {Neuhaus}, \citenamefont {Ehrlich}, \citenamefont
  {Makkeh}, \citenamefont {Ecker}, \citenamefont {Priesemann},\ and\
  \citenamefont {Wibral}}]{schneider2025should}%
  \BibitemOpen
  \bibfield  {author} {\bibinfo {author} {\bibfnamefont {A.~C.}\ \bibnamefont
  {Schneider}}, \bibinfo {author} {\bibfnamefont {V.}~\bibnamefont {Neuhaus}},
  \bibinfo {author} {\bibfnamefont {D.~A.}\ \bibnamefont {Ehrlich}}, \bibinfo
  {author} {\bibfnamefont {A.}~\bibnamefont {Makkeh}}, \bibinfo {author}
  {\bibfnamefont {A.~S.}\ \bibnamefont {Ecker}}, \bibinfo {author}
  {\bibfnamefont {V.}~\bibnamefont {Priesemann}}, \ and\ \bibinfo {author}
  {\bibfnamefont {M.}~\bibnamefont {Wibral}},\ }in\ \href@noop {} {\emph
  {\bibinfo {booktitle} {Proc. 13th Int. Conf. on Learning Representations
  (ICLR)}}}\ (\bibinfo {year} {2025})\ \bibinfo {note}
  {\url{https://openreview.net/forum?id=CLE09ESvul}}\BibitemShut {NoStop}%
\bibitem [{\citenamefont {Makkeh}\ \emph {et~al.}(2025)\citenamefont {Makkeh},
  \citenamefont {Graetz}, \citenamefont {Schneider}, \citenamefont {Ehrlich},
  \citenamefont {Priesemann},\ and\ \citenamefont
  {Wibral}}]{makkeh2025general}%
  \BibitemOpen
  \bibfield  {author} {\bibinfo {author} {\bibfnamefont {A.}~\bibnamefont
  {Makkeh}}, \bibinfo {author} {\bibfnamefont {M.}~\bibnamefont {Graetz}},
  \bibinfo {author} {\bibfnamefont {A.~C.}\ \bibnamefont {Schneider}}, \bibinfo
  {author} {\bibfnamefont {D.~A.}\ \bibnamefont {Ehrlich}}, \bibinfo {author}
  {\bibfnamefont {V.}~\bibnamefont {Priesemann}}, \ and\ \bibinfo {author}
  {\bibfnamefont {M.}~\bibnamefont {Wibral}},\ }\href {\doibase
  10.1073/pnas.2408125122} {\bibfield  {journal} {\bibinfo  {journal}
  {Proceedings of the National Academy of Sciences}\ }\textbf {\bibinfo
  {volume} {122}},\ \bibinfo {pages} {e2408125122} (\bibinfo {year}
  {2025})}\BibitemShut {NoStop}%
\bibitem [{\citenamefont {Wibral}\ \emph {et~al.}(2017)\citenamefont {Wibral},
  \citenamefont {Priesemann}, \citenamefont {Kay}, \citenamefont {Lizier},\
  and\ \citenamefont {Phillips}}]{wibral2017partial}%
  \BibitemOpen
  \bibfield  {author} {\bibinfo {author} {\bibfnamefont {M.}~\bibnamefont
  {Wibral}}, \bibinfo {author} {\bibfnamefont {V.}~\bibnamefont {Priesemann}},
  \bibinfo {author} {\bibfnamefont {J.~W.}\ \bibnamefont {Kay}}, \bibinfo
  {author} {\bibfnamefont {J.~T.}\ \bibnamefont {Lizier}}, \ and\ \bibinfo
  {author} {\bibfnamefont {W.~A.}\ \bibnamefont {Phillips}},\ }\href@noop {}
  {\bibfield  {journal} {\bibinfo  {journal} {Brain and cognition}\ }\textbf
  {\bibinfo {volume} {112}},\ \bibinfo {pages} {25} (\bibinfo {year}
  {2017})}\BibitemShut {NoStop}%
\bibitem [{\citenamefont {Wollstadt}\ \emph {et~al.}(2023)\citenamefont
  {Wollstadt}, \citenamefont {Schmitt},\ and\ \citenamefont
  {Wibral}}]{wollstadt2023rigorous}%
  \BibitemOpen
  \bibfield  {author} {\bibinfo {author} {\bibfnamefont {P.}~\bibnamefont
  {Wollstadt}}, \bibinfo {author} {\bibfnamefont {S.}~\bibnamefont {Schmitt}},
  \ and\ \bibinfo {author} {\bibfnamefont {M.}~\bibnamefont {Wibral}},\
  }\href@noop {} {\bibfield  {journal} {\bibinfo  {journal} {Journal of Machine
  Learning Research}\ }\textbf {\bibinfo {volume} {24}},\ \bibinfo {pages} {1}
  (\bibinfo {year} {2023})}\BibitemShut {NoStop}%
\bibitem [{\citenamefont {Varley}\ and\ \citenamefont
  {Bongard}(2024)}]{varley2024evolving}%
  \BibitemOpen
  \bibfield  {author} {\bibinfo {author} {\bibfnamefont {T.~F.}\ \bibnamefont
  {Varley}}\ and\ \bibinfo {author} {\bibfnamefont {J.}~\bibnamefont
  {Bongard}},\ }\href@noop {} {\bibfield  {journal} {\bibinfo  {journal}
  {Chaos: An Interdisciplinary Journal of Nonlinear Science}\ }\textbf
  {\bibinfo {volume} {34}} (\bibinfo {year} {2024})}\BibitemShut {NoStop}%
\bibitem [{\citenamefont {Rosas}\ \emph {et~al.}(2020)\citenamefont {Rosas},
  \citenamefont {Mediano}, \citenamefont {Rassouli},\ and\ \citenamefont
  {Barrett}}]{rosas2020operational}%
  \BibitemOpen
  \bibfield  {author} {\bibinfo {author} {\bibfnamefont {F.~E.}\ \bibnamefont
  {Rosas}}, \bibinfo {author} {\bibfnamefont {P.~A.}\ \bibnamefont {Mediano}},
  \bibinfo {author} {\bibfnamefont {B.}~\bibnamefont {Rassouli}}, \ and\
  \bibinfo {author} {\bibfnamefont {A.~B.}\ \bibnamefont {Barrett}},\
  }\href@noop {} {\bibfield  {journal} {\bibinfo  {journal} {Journal of Physics
  A: Mathematical and Theoretical}\ }\textbf {\bibinfo {volume} {53}},\
  \bibinfo {pages} {485001} (\bibinfo {year} {2020})}\BibitemShut {NoStop}%
\bibitem [{\citenamefont {Finn}\ and\ \citenamefont
  {Lizier}(2018)}]{finn2018pointwise}%
  \BibitemOpen
  \bibfield  {author} {\bibinfo {author} {\bibfnamefont {C.}~\bibnamefont
  {Finn}}\ and\ \bibinfo {author} {\bibfnamefont {J.}~\bibnamefont {Lizier}},\
  }\href@noop {} {\bibfield  {journal} {\bibinfo  {journal} {Entropy}\ }\textbf
  {\bibinfo {volume} {20}},\ \bibinfo {pages} {297} (\bibinfo {year}
  {2018})}\BibitemShut {NoStop}%
\bibitem [{\citenamefont {Makkeh}\ \emph {et~al.}(2021)\citenamefont {Makkeh},
  \citenamefont {Gutknecht},\ and\ \citenamefont {Wibral}}]{makkeh2021isx}%
  \BibitemOpen
  \bibfield  {author} {\bibinfo {author} {\bibfnamefont {A.}~\bibnamefont
  {Makkeh}}, \bibinfo {author} {\bibfnamefont {A.~J.}\ \bibnamefont
  {Gutknecht}}, \ and\ \bibinfo {author} {\bibfnamefont {M.}~\bibnamefont
  {Wibral}},\ }\href@noop {} {\bibfield  {journal} {\bibinfo  {journal}
  {Physical Review E}\ }\textbf {\bibinfo {volume} {103}},\ \bibinfo {pages}
  {032149} (\bibinfo {year} {2021})}\BibitemShut {NoStop}%
\bibitem [{\citenamefont {Kolchinsky}(2022)}]{Kolchinsky2022}%
  \BibitemOpen
  \bibfield  {author} {\bibinfo {author} {\bibfnamefont {A.}~\bibnamefont
  {Kolchinsky}},\ }\href@noop {} {\bibfield  {journal} {\bibinfo  {journal}
  {Entropy}\ }\textbf {\bibinfo {volume} {24}},\ \bibinfo {pages} {403}
  (\bibinfo {year} {2022})}\BibitemShut {NoStop}%
\bibitem [{\citenamefont {van Enk}(2023)}]{van2023pooling}%
  \BibitemOpen
  \bibfield  {author} {\bibinfo {author} {\bibfnamefont {S.~J.}\ \bibnamefont
  {van Enk}},\ }\href@noop {} {\bibfield  {journal} {\bibinfo  {journal}
  {Physical Review E}\ }\textbf {\bibinfo {volume} {107}},\ \bibinfo {pages}
  {054133} (\bibinfo {year} {2023})}\BibitemShut {NoStop}%
\bibitem [{\citenamefont {Gutknecht}\ \emph {et~al.}(2025)\citenamefont
  {Gutknecht}, \citenamefont {Makkeh},\ and\ \citenamefont
  {Wibral}}]{Gutknecht2025babel}%
  \BibitemOpen
  \bibfield  {author} {\bibinfo {author} {\bibfnamefont {A.~J.}\ \bibnamefont
  {Gutknecht}}, \bibinfo {author} {\bibfnamefont {A.}~\bibnamefont {Makkeh}}, \
  and\ \bibinfo {author} {\bibfnamefont {M.}~\bibnamefont {Wibral}},\ }\href
  {\doibase 10.1098/rspa.2024.0174} {\bibfield  {journal} {\bibinfo  {journal}
  {Proceedings of the Royal Society A}\ }\textbf {\bibinfo {volume} {481}},\
  \bibinfo {pages} {20240174} (\bibinfo {year} {2025})}\BibitemShut {NoStop}%
\bibitem [{\citenamefont {Liardi}\ \emph {et~al.}(2026)\citenamefont {Liardi},
  \citenamefont {Down}, \citenamefont {Blackburne}, \citenamefont {Neri},\ and\
  \citenamefont {Mediano}}]{liardi2026mathematical}%
  \BibitemOpen
  \bibfield  {author} {\bibinfo {author} {\bibfnamefont {A.}~\bibnamefont
  {Liardi}}, \bibinfo {author} {\bibfnamefont {K.~J.}\ \bibnamefont {Down}},
  \bibinfo {author} {\bibfnamefont {G.}~\bibnamefont {Blackburne}}, \bibinfo
  {author} {\bibfnamefont {M.}~\bibnamefont {Neri}}, \ and\ \bibinfo {author}
  {\bibfnamefont {P.~A.}\ \bibnamefont {Mediano}},\ }\href@noop {} {\bibfield
  {journal} {\bibinfo  {journal} {arXiv preprint arXiv:2603.06678}\ } (\bibinfo
  {year} {2026})}\BibitemShut {NoStop}%
\bibitem [{\citenamefont {Jansma}\ \emph {et~al.}(2025)\citenamefont {Jansma},
  \citenamefont {Mediano},\ and\ \citenamefont {Rosas}}]{jansma2025fast}%
  \BibitemOpen
  \bibfield  {author} {\bibinfo {author} {\bibfnamefont {A.}~\bibnamefont
  {Jansma}}, \bibinfo {author} {\bibfnamefont {P.~A.}\ \bibnamefont {Mediano}},
  \ and\ \bibinfo {author} {\bibfnamefont {F.~E.}\ \bibnamefont {Rosas}},\
  }\href@noop {} {\bibfield  {journal} {\bibinfo  {journal} {Physical Review
  Research}\ }\textbf {\bibinfo {volume} {7}},\ \bibinfo {pages} {033049}
  (\bibinfo {year} {2025})}\BibitemShut {NoStop}%
\bibitem [{\citenamefont {Jaynes}(1957)}]{jaynes1957information}%
  \BibitemOpen
  \bibfield  {author} {\bibinfo {author} {\bibfnamefont {E.~T.}\ \bibnamefont
  {Jaynes}},\ }\href@noop {} {\bibfield  {journal} {\bibinfo  {journal}
  {Physical review}\ }\textbf {\bibinfo {volume} {106}},\ \bibinfo {pages}
  {620} (\bibinfo {year} {1957})}\BibitemShut {NoStop}%
\bibitem [{\citenamefont {Rosas}\ \emph
  {et~al.}(2024{\natexlab{a}})\citenamefont {Rosas}, \citenamefont {Geiger},
  \citenamefont {Luppi}, \citenamefont {Seth}, \citenamefont {Polani},
  \citenamefont {Gastpar},\ and\ \citenamefont {Mediano}}]{rosas2024software}%
  \BibitemOpen
  \bibfield  {author} {\bibinfo {author} {\bibfnamefont {F.~E.}\ \bibnamefont
  {Rosas}}, \bibinfo {author} {\bibfnamefont {B.~C.}\ \bibnamefont {Geiger}},
  \bibinfo {author} {\bibfnamefont {A.~I.}\ \bibnamefont {Luppi}}, \bibinfo
  {author} {\bibfnamefont {A.~K.}\ \bibnamefont {Seth}}, \bibinfo {author}
  {\bibfnamefont {D.}~\bibnamefont {Polani}}, \bibinfo {author} {\bibfnamefont
  {M.}~\bibnamefont {Gastpar}}, \ and\ \bibinfo {author} {\bibfnamefont
  {P.~A.}\ \bibnamefont {Mediano}},\ }\href@noop {} {\bibfield  {journal}
  {\bibinfo  {journal} {arXiv preprint arXiv:2402.09090}\ } (\bibinfo {year}
  {2024}{\natexlab{a}})}\BibitemShut {NoStop}%
\bibitem [{\citenamefont {Chechik}\ \emph {et~al.}(2002)\citenamefont
  {Chechik}, \citenamefont {Globerson}, \citenamefont {Anderson}, \citenamefont
  {Young}, \citenamefont {Nelken},\ and\ \citenamefont
  {Tishby}}]{Chechik2002_group}%
  \BibitemOpen
  \bibfield  {author} {\bibinfo {author} {\bibfnamefont {G.}~\bibnamefont
  {Chechik}}, \bibinfo {author} {\bibfnamefont {A.}~\bibnamefont {Globerson}},
  \bibinfo {author} {\bibfnamefont {M.}~\bibnamefont {Anderson}}, \bibinfo
  {author} {\bibfnamefont {E.}~\bibnamefont {Young}}, \bibinfo {author}
  {\bibfnamefont {I.}~\bibnamefont {Nelken}}, \ and\ \bibinfo {author}
  {\bibfnamefont {N.}~\bibnamefont {Tishby}},\ }in\ \href@noop {} {\emph
  {\bibinfo {booktitle} {Advances in Neural Information Processing Systems
  14}}},\ \bibinfo {editor} {edited by\ \bibinfo {editor} {\bibfnamefont
  {T.~G.}\ \bibnamefont {Dietterich}}, \bibinfo {editor} {\bibfnamefont
  {S.}~\bibnamefont {Becker}}, \ and\ \bibinfo {editor} {\bibfnamefont
  {Z.}~\bibnamefont {Ghahramani}}}\ (\bibinfo  {publisher} {MIT Press},\
  \bibinfo {address} {Cambridge, MA},\ \bibinfo {year} {2002})\BibitemShut
  {NoStop}%
\bibitem [{\citenamefont {Timme}\ \emph {et~al.}(2014)\citenamefont {Timme},
  \citenamefont {Alford}, \citenamefont {Flecker},\ and\ \citenamefont
  {Beggs}}]{timme2014synergy}%
  \BibitemOpen
  \bibfield  {author} {\bibinfo {author} {\bibfnamefont {N.}~\bibnamefont
  {Timme}}, \bibinfo {author} {\bibfnamefont {W.}~\bibnamefont {Alford}},
  \bibinfo {author} {\bibfnamefont {B.}~\bibnamefont {Flecker}}, \ and\
  \bibinfo {author} {\bibfnamefont {J.~M.}\ \bibnamefont {Beggs}},\ }\href@noop
  {} {\bibfield  {journal} {\bibinfo  {journal} {Journal of Computational
  Neuroscience}\ }\textbf {\bibinfo {volume} {36}},\ \bibinfo {pages} {119}
  (\bibinfo {year} {2014})}\BibitemShut {NoStop}%
\bibitem [{\citenamefont {Rosas}\ \emph
  {et~al.}(2024{\natexlab{b}})\citenamefont {Rosas}, \citenamefont {Mediano},\
  and\ \citenamefont {Gastpar}}]{rosas2024characterising}%
  \BibitemOpen
  \bibfield  {author} {\bibinfo {author} {\bibfnamefont {F.~E.}\ \bibnamefont
  {Rosas}}, \bibinfo {author} {\bibfnamefont {P.~A.~M.}\ \bibnamefont
  {Mediano}}, \ and\ \bibinfo {author} {\bibfnamefont {M.}~\bibnamefont
  {Gastpar}},\ }\href@noop {} {\bibfield  {journal} {\bibinfo  {journal} {arXiv
  preprint arXiv:2404.07140}\ } (\bibinfo {year} {2024}{\natexlab{b}})},\
  \bibinfo {note} {\url{http://arxiv.org/abs/2404.07140}}\BibitemShut {NoStop}%
\bibitem [{\citenamefont {Varley}(2024)}]{varley2024scalable}%
  \BibitemOpen
  \bibfield  {author} {\bibinfo {author} {\bibfnamefont {T.~F.}\ \bibnamefont
  {Varley}},\ }\href@noop {} {\bibfield  {journal} {\bibinfo  {journal} {npj
  Complexity}\ }\textbf {\bibinfo {volume} {1}},\ \bibinfo {pages} {9}
  (\bibinfo {year} {2024})}\BibitemShut {NoStop}%
\bibitem [{\citenamefont {Rosas}\ \emph {et~al.}(2019)\citenamefont {Rosas},
  \citenamefont {Mediano}, \citenamefont {Gastpar},\ and\ \citenamefont
  {Jensen}}]{rosas2019quantifying}%
  \BibitemOpen
  \bibfield  {author} {\bibinfo {author} {\bibfnamefont {F.~E.}\ \bibnamefont
  {Rosas}}, \bibinfo {author} {\bibfnamefont {P.~A.}\ \bibnamefont {Mediano}},
  \bibinfo {author} {\bibfnamefont {M.}~\bibnamefont {Gastpar}}, \ and\
  \bibinfo {author} {\bibfnamefont {H.~J.}\ \bibnamefont {Jensen}},\
  }\href@noop {} {\bibfield  {journal} {\bibinfo  {journal} {Physical Review
  E}\ }\textbf {\bibinfo {volume} {100}},\ \bibinfo {pages} {032305} (\bibinfo
  {year} {2019})}\BibitemShut {NoStop}%
\bibitem [{\citenamefont {Watanabe}(1960)}]{watanabe1960information}%
  \BibitemOpen
  \bibfield  {author} {\bibinfo {author} {\bibfnamefont {S.}~\bibnamefont
  {Watanabe}},\ }\href@noop {} {\bibfield  {journal} {\bibinfo  {journal} {IBM
  Journal of research and development}\ }\textbf {\bibinfo {volume} {4}},\
  \bibinfo {pages} {66} (\bibinfo {year} {1960})}\BibitemShut {NoStop}%
\bibitem [{\citenamefont {Te~Sun}(1978)}]{te1978nonnegative}%
  \BibitemOpen
  \bibfield  {author} {\bibinfo {author} {\bibfnamefont {H.}~\bibnamefont
  {Te~Sun}},\ }\href@noop {} {\bibfield  {journal} {\bibinfo  {journal}
  {Information and Control}\ }\textbf {\bibinfo {volume} {36}},\ \bibinfo
  {pages} {133} (\bibinfo {year} {1978})}\BibitemShut {NoStop}%
\bibitem [{\citenamefont {McGill}(1954)}]{mcgill1954multivariate}%
  \BibitemOpen
  \bibfield  {author} {\bibinfo {author} {\bibfnamefont {W.}~\bibnamefont
  {McGill}},\ }\href@noop {} {\bibfield  {journal} {\bibinfo  {journal}
  {Transactions of the IRE Professional Group on Information Theory}\ }\textbf
  {\bibinfo {volume} {4}},\ \bibinfo {pages} {93} (\bibinfo {year}
  {1954})}\BibitemShut {NoStop}%
\bibitem [{\citenamefont {Tononi}\ \emph {et~al.}(1994)\citenamefont {Tononi},
  \citenamefont {Sporns},\ and\ \citenamefont {Edelman}}]{tononi1994measure}%
  \BibitemOpen
  \bibfield  {author} {\bibinfo {author} {\bibfnamefont {G.}~\bibnamefont
  {Tononi}}, \bibinfo {author} {\bibfnamefont {O.}~\bibnamefont {Sporns}}, \
  and\ \bibinfo {author} {\bibfnamefont {G.~M.}\ \bibnamefont {Edelman}},\
  }\href@noop {} {\bibfield  {journal} {\bibinfo  {journal} {Proceedings of the
  National Academy of Sciences}\ }\textbf {\bibinfo {volume} {91}},\ \bibinfo
  {pages} {5033} (\bibinfo {year} {1994})}\BibitemShut {NoStop}%
\bibitem [{\citenamefont {Brenner}\ \emph {et~al.}(2000)\citenamefont
  {Brenner}, \citenamefont {Strong}, \citenamefont {Koberle}, \citenamefont
  {Bialek},\ and\ \citenamefont {Steveninck}}]{brenner2000synergy}%
  \BibitemOpen
  \bibfield  {author} {\bibinfo {author} {\bibfnamefont {N.}~\bibnamefont
  {Brenner}}, \bibinfo {author} {\bibfnamefont {S.~P.}\ \bibnamefont {Strong}},
  \bibinfo {author} {\bibfnamefont {R.}~\bibnamefont {Koberle}}, \bibinfo
  {author} {\bibfnamefont {W.}~\bibnamefont {Bialek}}, \ and\ \bibinfo {author}
  {\bibfnamefont {R.~R. d. R.~v.}\ \bibnamefont {Steveninck}},\ }\href@noop {}
  {\bibfield  {journal} {\bibinfo  {journal} {Neural computation}\ }\textbf
  {\bibinfo {volume} {12}},\ \bibinfo {pages} {1531} (\bibinfo {year}
  {2000})}\BibitemShut {NoStop}%
\bibitem [{\citenamefont {Gawne}\ and\ \citenamefont
  {Richmond}(1993)}]{gawne1993independent}%
  \BibitemOpen
  \bibfield  {author} {\bibinfo {author} {\bibfnamefont {T.~J.}\ \bibnamefont
  {Gawne}}\ and\ \bibinfo {author} {\bibfnamefont {B.~J.}\ \bibnamefont
  {Richmond}},\ }\href@noop {} {\bibfield  {journal} {\bibinfo  {journal}
  {Journal of Neuroscience}\ }\textbf {\bibinfo {volume} {13}},\ \bibinfo
  {pages} {2758} (\bibinfo {year} {1993})}\BibitemShut {NoStop}%
\bibitem [{\citenamefont {Griffith}\ and\ \citenamefont
  {Koch}(2014)}]{griffith2014quantifying}%
  \BibitemOpen
  \bibfield  {author} {\bibinfo {author} {\bibfnamefont {V.}~\bibnamefont
  {Griffith}}\ and\ \bibinfo {author} {\bibfnamefont {C.}~\bibnamefont
  {Koch}},\ }in\ \href@noop {} {\emph {\bibinfo {booktitle} {Guided
  self-organization: inception}}}\ (\bibinfo  {publisher} {Springer},\ \bibinfo
  {year} {2014})\ pp.\ \bibinfo {pages} {159--190}\BibitemShut {NoStop}%
\bibitem [{\citenamefont {Yang}\ and\ \citenamefont
  {Gu}(2004)}]{yang2004feature}%
  \BibitemOpen
  \bibfield  {author} {\bibinfo {author} {\bibfnamefont {S.}~\bibnamefont
  {Yang}}\ and\ \bibinfo {author} {\bibfnamefont {J.}~\bibnamefont {Gu}},\
  }\href@noop {} {\bibfield  {journal} {\bibinfo  {journal} {Journal of
  Zhejiang University-Science A}\ }\textbf {\bibinfo {volume} {5}},\ \bibinfo
  {pages} {1382} (\bibinfo {year} {2004})}\BibitemShut {NoStop}%
\bibitem [{\citenamefont {Mares}\ \emph {et~al.}(2022)\citenamefont {Mares},
  \citenamefont {Mares}, \citenamefont {Dobrica},\ and\ \citenamefont
  {Demetrescu}}]{mares2022selection}%
  \BibitemOpen
  \bibfield  {author} {\bibinfo {author} {\bibfnamefont {I.}~\bibnamefont
  {Mares}}, \bibinfo {author} {\bibfnamefont {C.}~\bibnamefont {Mares}},
  \bibinfo {author} {\bibfnamefont {V.}~\bibnamefont {Dobrica}}, \ and\
  \bibinfo {author} {\bibfnamefont {C.}~\bibnamefont {Demetrescu}},\
  }\href@noop {} {\bibfield  {journal} {\bibinfo  {journal} {Entropy}\ }\textbf
  {\bibinfo {volume} {24}},\ \bibinfo {pages} {1375} (\bibinfo {year}
  {2022})}\BibitemShut {NoStop}%
\bibitem [{\citenamefont {Mares}\ \emph {et~al.}(2023)\citenamefont {Mares},
  \citenamefont {Dobrica}, \citenamefont {Demetrescu},\ and\ \citenamefont
  {Mares}}]{mares2023combined}%
  \BibitemOpen
  \bibfield  {author} {\bibinfo {author} {\bibfnamefont {I.}~\bibnamefont
  {Mares}}, \bibinfo {author} {\bibfnamefont {V.}~\bibnamefont {Dobrica}},
  \bibinfo {author} {\bibfnamefont {C.}~\bibnamefont {Demetrescu}}, \ and\
  \bibinfo {author} {\bibfnamefont {C.}~\bibnamefont {Mares}},\ }\href@noop {}
  {\bibfield  {journal} {\bibinfo  {journal} {Atmosphere}\ }\textbf {\bibinfo
  {volume} {14}},\ \bibinfo {pages} {1622} (\bibinfo {year}
  {2023})}\BibitemShut {NoStop}%
\bibitem [{\citenamefont {Luecke}\ \emph {et~al.}(2024)\citenamefont {Luecke},
  \citenamefont {Guo}, \citenamefont {Sheu}, \citenamefont {Singh},
  \citenamefont {Lowe}, \citenamefont {Han}, \citenamefont {Diaz},
  \citenamefont {Lopes}, \citenamefont {Wollman},\ and\ \citenamefont
  {Hoffmann}}]{luecke2024dynamical}%
  \BibitemOpen
  \bibfield  {author} {\bibinfo {author} {\bibfnamefont {S.}~\bibnamefont
  {Luecke}}, \bibinfo {author} {\bibfnamefont {X.}~\bibnamefont {Guo}},
  \bibinfo {author} {\bibfnamefont {K.~M.}\ \bibnamefont {Sheu}}, \bibinfo
  {author} {\bibfnamefont {A.}~\bibnamefont {Singh}}, \bibinfo {author}
  {\bibfnamefont {S.~C.}\ \bibnamefont {Lowe}}, \bibinfo {author}
  {\bibfnamefont {M.}~\bibnamefont {Han}}, \bibinfo {author} {\bibfnamefont
  {J.}~\bibnamefont {Diaz}}, \bibinfo {author} {\bibfnamefont {F.}~\bibnamefont
  {Lopes}}, \bibinfo {author} {\bibfnamefont {R.}~\bibnamefont {Wollman}}, \
  and\ \bibinfo {author} {\bibfnamefont {A.}~\bibnamefont {Hoffmann}},\
  }\href@noop {} {\bibfield  {journal} {\bibinfo  {journal} {Molecular Systems
  Biology}\ }\textbf {\bibinfo {volume} {20}},\ \bibinfo {pages} {898}
  (\bibinfo {year} {2024})}\BibitemShut {NoStop}%
\bibitem [{\citenamefont {Lin}\ and\ \citenamefont
  {Costello}(2004)}]{lin2001error}%
  \BibitemOpen
  \bibfield  {author} {\bibinfo {author} {\bibfnamefont {S.}~\bibnamefont
  {Lin}}\ and\ \bibinfo {author} {\bibfnamefont {D.}~\bibnamefont {Costello}},\
  }\href@noop {} {\emph {\bibinfo {title} {{Error Control Coding: Fundamentals
  and Applications}}}},\ \bibinfo {edition} {2nd}\ ed.\ (\bibinfo  {publisher}
  {Prentice-Hall, Inc.},\ \bibinfo {year} {2004})\BibitemShut {NoStop}%
\bibitem [{\citenamefont {LeCun}\ \emph {et~al.}(1998)\citenamefont {LeCun},
  \citenamefont {Bottou}, \citenamefont {Bengio},\ and\ \citenamefont
  {Haffner}}]{lecun_gradient-based_1998}%
  \BibitemOpen
  \bibfield  {author} {\bibinfo {author} {\bibfnamefont {Y.}~\bibnamefont
  {LeCun}}, \bibinfo {author} {\bibfnamefont {L.}~\bibnamefont {Bottou}},
  \bibinfo {author} {\bibfnamefont {Y.}~\bibnamefont {Bengio}}, \ and\ \bibinfo
  {author} {\bibfnamefont {P.}~\bibnamefont {Haffner}},\ }\href@noop {}
  {\bibfield  {journal} {\bibinfo  {journal} {Proceedings of the IEEE}\
  }\textbf {\bibinfo {volume} {86}},\ \bibinfo {pages} {2278} (\bibinfo {year}
  {1998})}\BibitemShut {NoStop}%
\bibitem [{\citenamefont {Huang}\ \emph {et~al.}(2007)\citenamefont {Huang},
  \citenamefont {Ramesh}, \citenamefont {Berg},\ and\ \citenamefont
  {Learned-Miller}}]{LFWTech}%
  \BibitemOpen
  \bibfield  {author} {\bibinfo {author} {\bibfnamefont {G.~B.}\ \bibnamefont
  {Huang}}, \bibinfo {author} {\bibfnamefont {M.}~\bibnamefont {Ramesh}},
  \bibinfo {author} {\bibfnamefont {T.}~\bibnamefont {Berg}}, \ and\ \bibinfo
  {author} {\bibfnamefont {E.}~\bibnamefont {Learned-Miller}},\ }\href@noop {}
  {\emph {\bibinfo {title} {Labeled Faces in the Wild: A Database for Studying
  Face Recognition in Unconstrained Environments}}},\ \bibinfo {type} {Tech.
  Rep.}\ \bibinfo {number} {07-49}\ (\bibinfo  {institution} {University of
  Massachusetts, Amherst},\ \bibinfo {year} {2007})\BibitemShut {NoStop}%
\bibitem [{\citenamefont {Barnett}\ and\ \citenamefont
  {Seth}(2015)}]{barnett2015granger}%
  \BibitemOpen
  \bibfield  {author} {\bibinfo {author} {\bibfnamefont {L.}~\bibnamefont
  {Barnett}}\ and\ \bibinfo {author} {\bibfnamefont {A.~K.}\ \bibnamefont
  {Seth}},\ }\href@noop {} {\bibfield  {journal} {\bibinfo  {journal} {Physical
  Review E}\ }\textbf {\bibinfo {volume} {91}},\ \bibinfo {pages} {040101}
  (\bibinfo {year} {2015})}\BibitemShut {NoStop}%
\bibitem [{\citenamefont {Schreiber}(2000)}]{schreiber2000measuring}%
  \BibitemOpen
  \bibfield  {author} {\bibinfo {author} {\bibfnamefont {T.}~\bibnamefont
  {Schreiber}},\ }\href@noop {} {\bibfield  {journal} {\bibinfo  {journal}
  {Physical review letters}\ }\textbf {\bibinfo {volume} {85}},\ \bibinfo
  {pages} {461} (\bibinfo {year} {2000})}\BibitemShut {NoStop}%
\bibitem [{\citenamefont {Barnett}\ \emph {et~al.}(2009)\citenamefont
  {Barnett}, \citenamefont {Barrett},\ and\ \citenamefont
  {Seth}}]{barnett2009granger}%
  \BibitemOpen
  \bibfield  {author} {\bibinfo {author} {\bibfnamefont {L.}~\bibnamefont
  {Barnett}}, \bibinfo {author} {\bibfnamefont {A.~B.}\ \bibnamefont
  {Barrett}}, \ and\ \bibinfo {author} {\bibfnamefont {A.~K.}\ \bibnamefont
  {Seth}},\ }\href@noop {} {\bibfield  {journal} {\bibinfo  {journal} {Physical
  review letters}\ }\textbf {\bibinfo {volume} {103}},\ \bibinfo {pages}
  {238701} (\bibinfo {year} {2009})}\BibitemShut {NoStop}%
\bibitem [{\citenamefont {Wibral}\ \emph {et~al.}(2014)\citenamefont {Wibral},
  \citenamefont {Vicente},\ and\ \citenamefont {Lindner}}]{wibral2014transfer}%
  \BibitemOpen
  \bibfield  {author} {\bibinfo {author} {\bibfnamefont {M.}~\bibnamefont
  {Wibral}}, \bibinfo {author} {\bibfnamefont {R.}~\bibnamefont {Vicente}}, \
  and\ \bibinfo {author} {\bibfnamefont {M.}~\bibnamefont {Lindner}},\ }in\
  \href@noop {} {\emph {\bibinfo {booktitle} {Directed information measures in
  neuroscience}}}\ (\bibinfo  {publisher} {Springer},\ \bibinfo {year} {2014})\
  pp.\ \bibinfo {pages} {3--36}\BibitemShut {NoStop}%
\bibitem [{\citenamefont {Bossomaier}\ \emph {et~al.}(2016)\citenamefont
  {Bossomaier}, \citenamefont {Barnett}, \citenamefont {Harr{\'e}},\ and\
  \citenamefont {Lizier}}]{bossomaier2016transfer}%
  \BibitemOpen
  \bibfield  {author} {\bibinfo {author} {\bibfnamefont {T.}~\bibnamefont
  {Bossomaier}}, \bibinfo {author} {\bibfnamefont {L.}~\bibnamefont {Barnett}},
  \bibinfo {author} {\bibfnamefont {M.}~\bibnamefont {Harr{\'e}}}, \ and\
  \bibinfo {author} {\bibfnamefont {J.~T.}\ \bibnamefont {Lizier}},\ }in\
  \href@noop {} {\emph {\bibinfo {booktitle} {An introduction to transfer
  entropy: Information flow in complex systems}}}\ (\bibinfo  {publisher}
  {Springer},\ \bibinfo {year} {2016})\ pp.\ \bibinfo {pages}
  {65--95}\BibitemShut {NoStop}%
\bibitem [{\citenamefont {Nugent}\ \emph {et~al.}(2022)\citenamefont {Nugent},
  \citenamefont {Thomas}, \citenamefont {Mahoney}, \citenamefont {Gibbons},
  \citenamefont {Smith}, \citenamefont {Charles}, \citenamefont {Shaw},
  \citenamefont {Stout}, \citenamefont {Namyst}, \citenamefont {Basavaraj}
  \emph {et~al.}}]{nugent2022nimh}%
  \BibitemOpen
  \bibfield  {author} {\bibinfo {author} {\bibfnamefont {A.~C.}\ \bibnamefont
  {Nugent}}, \bibinfo {author} {\bibfnamefont {A.~G.}\ \bibnamefont {Thomas}},
  \bibinfo {author} {\bibfnamefont {M.}~\bibnamefont {Mahoney}}, \bibinfo
  {author} {\bibfnamefont {A.}~\bibnamefont {Gibbons}}, \bibinfo {author}
  {\bibfnamefont {J.~T.}\ \bibnamefont {Smith}}, \bibinfo {author}
  {\bibfnamefont {A.~J.}\ \bibnamefont {Charles}}, \bibinfo {author}
  {\bibfnamefont {J.~S.}\ \bibnamefont {Shaw}}, \bibinfo {author}
  {\bibfnamefont {J.~D.}\ \bibnamefont {Stout}}, \bibinfo {author}
  {\bibfnamefont {A.~M.}\ \bibnamefont {Namyst}}, \bibinfo {author}
  {\bibfnamefont {A.}~\bibnamefont {Basavaraj}},  \emph {et~al.},\ }\href@noop
  {} {\bibfield  {journal} {\bibinfo  {journal} {Scientific Data}\ }\textbf
  {\bibinfo {volume} {9}},\ \bibinfo {pages} {518} (\bibinfo {year}
  {2022})}\BibitemShut {NoStop}%
\bibitem [{\citenamefont {Goldfeld}\ \emph {et~al.}(2019)\citenamefont
  {Goldfeld}, \citenamefont {Berg}, \citenamefont {Greenewald}, \citenamefont
  {Melnyk}, \citenamefont {Nguyen}, \citenamefont {Kingsbury},\ and\
  \citenamefont {Polyanskiy}}]{goldfeld_estimating_2019}%
  \BibitemOpen
  \bibfield  {author} {\bibinfo {author} {\bibfnamefont {Z.}~\bibnamefont
  {Goldfeld}}, \bibinfo {author} {\bibfnamefont {E.~v.~d.}\ \bibnamefont
  {Berg}}, \bibinfo {author} {\bibfnamefont {K.}~\bibnamefont {Greenewald}},
  \bibinfo {author} {\bibfnamefont {I.}~\bibnamefont {Melnyk}}, \bibinfo
  {author} {\bibfnamefont {N.}~\bibnamefont {Nguyen}}, \bibinfo {author}
  {\bibfnamefont {B.}~\bibnamefont {Kingsbury}}, \ and\ \bibinfo {author}
  {\bibfnamefont {Y.}~\bibnamefont {Polyanskiy}},\ }\href@noop {} {\bibfield
  {journal} {\bibinfo  {journal} {International Conference on Machine
  Learning}\ } (\bibinfo {year} {2019})}\BibitemShut {NoStop}%
\bibitem [{\citenamefont {Tishby}\ and\ \citenamefont
  {Zaslavsky}(2015)}]{tishby2015deep}%
  \BibitemOpen
  \bibfield  {author} {\bibinfo {author} {\bibfnamefont {N.}~\bibnamefont
  {Tishby}}\ and\ \bibinfo {author} {\bibfnamefont {N.}~\bibnamefont
  {Zaslavsky}},\ }in\ \href@noop {} {\emph {\bibinfo {booktitle} {2015 ieee
  information theory workshop (itw)}}}\ (\bibinfo {organization} {Ieee},\
  \bibinfo {year} {2015})\ pp.\ \bibinfo {pages} {1--5}\BibitemShut {NoStop}%
\bibitem [{\citenamefont {Saxe}\ \emph {et~al.}(2019)\citenamefont {Saxe},
  \citenamefont {Bansal}, \citenamefont {Dapello}, \citenamefont {Advani},
  \citenamefont {Kolchinsky}, \citenamefont {Tracey},\ and\ \citenamefont
  {Cox}}]{saxe_information_2019}%
  \BibitemOpen
  \bibfield  {author} {\bibinfo {author} {\bibfnamefont {A.~M.}\ \bibnamefont
  {Saxe}}, \bibinfo {author} {\bibfnamefont {Y.}~\bibnamefont {Bansal}},
  \bibinfo {author} {\bibfnamefont {J.}~\bibnamefont {Dapello}}, \bibinfo
  {author} {\bibfnamefont {M.}~\bibnamefont {Advani}}, \bibinfo {author}
  {\bibfnamefont {A.}~\bibnamefont {Kolchinsky}}, \bibinfo {author}
  {\bibfnamefont {B.~D.}\ \bibnamefont {Tracey}}, \ and\ \bibinfo {author}
  {\bibfnamefont {D.~D.}\ \bibnamefont {Cox}},\ }\href@noop {} {\bibfield
  {journal} {\bibinfo  {journal} {Journal of Statistical Mechanics: Theory and
  Experiment}\ }\textbf {\bibinfo {volume} {2019}},\ \bibinfo {pages} {124020}
  (\bibinfo {year} {2019})}\BibitemShut {NoStop}%
\bibitem [{\citenamefont {Tapia}\ and\ \citenamefont
  {Est{\'e}vez}(2020)}]{tapia2020information}%
  \BibitemOpen
  \bibfield  {author} {\bibinfo {author} {\bibfnamefont {N.~I.}\ \bibnamefont
  {Tapia}}\ and\ \bibinfo {author} {\bibfnamefont {P.~A.}\ \bibnamefont
  {Est{\'e}vez}},\ }in\ \href@noop {} {\emph {\bibinfo {booktitle} {2020
  International Joint Conference on Neural Networks (IJCNN)}}}\ (\bibinfo
  {organization} {IEEE},\ \bibinfo {year} {2020})\ pp.\ \bibinfo {pages}
  {1--8}\BibitemShut {NoStop}%
\bibitem [{\citenamefont {Burgess}\ \emph {et~al.}(2018)\citenamefont
  {Burgess}, \citenamefont {Higgins}, \citenamefont {Pal}, \citenamefont
  {Matthey}, \citenamefont {Watters}, \citenamefont {Desjardins},\ and\
  \citenamefont {Lerchner}}]{burgess2018understanding}%
  \BibitemOpen
  \bibfield  {author} {\bibinfo {author} {\bibfnamefont {C.~P.}\ \bibnamefont
  {Burgess}}, \bibinfo {author} {\bibfnamefont {I.}~\bibnamefont {Higgins}},
  \bibinfo {author} {\bibfnamefont {A.}~\bibnamefont {Pal}}, \bibinfo {author}
  {\bibfnamefont {L.}~\bibnamefont {Matthey}}, \bibinfo {author} {\bibfnamefont
  {N.}~\bibnamefont {Watters}}, \bibinfo {author} {\bibfnamefont
  {G.}~\bibnamefont {Desjardins}}, \ and\ \bibinfo {author} {\bibfnamefont
  {A.}~\bibnamefont {Lerchner}},\ }\href@noop {} {\bibfield  {journal}
  {\bibinfo  {journal} {arXiv preprint arXiv:1804.03599}\ } (\bibinfo {year}
  {2018})}\BibitemShut {NoStop}%
\bibitem [{\citenamefont {Elhage}\ \emph {et~al.}(2021)\citenamefont {Elhage},
  \citenamefont {Nanda}, \citenamefont {Olsson}, \citenamefont {Henighan},
  \citenamefont {Joseph}, \citenamefont {Mann}, \citenamefont {Askell},
  \citenamefont {Bai}, \citenamefont {Chen}, \citenamefont {Conerly} \emph
  {et~al.}}]{elhage2021mathematical}%
  \BibitemOpen
  \bibfield  {author} {\bibinfo {author} {\bibfnamefont {N.}~\bibnamefont
  {Elhage}}, \bibinfo {author} {\bibfnamefont {N.}~\bibnamefont {Nanda}},
  \bibinfo {author} {\bibfnamefont {C.}~\bibnamefont {Olsson}}, \bibinfo
  {author} {\bibfnamefont {T.}~\bibnamefont {Henighan}}, \bibinfo {author}
  {\bibfnamefont {N.}~\bibnamefont {Joseph}}, \bibinfo {author} {\bibfnamefont
  {B.}~\bibnamefont {Mann}}, \bibinfo {author} {\bibfnamefont {A.}~\bibnamefont
  {Askell}}, \bibinfo {author} {\bibfnamefont {Y.}~\bibnamefont {Bai}},
  \bibinfo {author} {\bibfnamefont {A.}~\bibnamefont {Chen}}, \bibinfo {author}
  {\bibfnamefont {T.}~\bibnamefont {Conerly}},  \emph {et~al.},\ }\href@noop {}
  {\bibfield  {journal} {\bibinfo  {journal} {Transformer Circuits Thread}\
  }\textbf {\bibinfo {volume} {1}},\ \bibinfo {pages} {12} (\bibinfo {year}
  {2021})}\BibitemShut {NoStop}%
\bibitem [{\citenamefont {Ameisen}\ \emph {et~al.}(2025)\citenamefont
  {Ameisen}, \citenamefont {Lindsey}, \citenamefont {Pearce}, \citenamefont
  {Gurnee}, \citenamefont {Turner}, \citenamefont {Chen}, \citenamefont
  {Citro}, \citenamefont {Abrahams}, \citenamefont {Carter}, \citenamefont
  {Hosmer} \emph {et~al.}}]{ameisen2025circuit}%
  \BibitemOpen
  \bibfield  {author} {\bibinfo {author} {\bibfnamefont {E.}~\bibnamefont
  {Ameisen}}, \bibinfo {author} {\bibfnamefont {J.}~\bibnamefont {Lindsey}},
  \bibinfo {author} {\bibfnamefont {A.}~\bibnamefont {Pearce}}, \bibinfo
  {author} {\bibfnamefont {W.}~\bibnamefont {Gurnee}}, \bibinfo {author}
  {\bibfnamefont {N.~L.}\ \bibnamefont {Turner}}, \bibinfo {author}
  {\bibfnamefont {B.}~\bibnamefont {Chen}}, \bibinfo {author} {\bibfnamefont
  {C.}~\bibnamefont {Citro}}, \bibinfo {author} {\bibfnamefont
  {D.}~\bibnamefont {Abrahams}}, \bibinfo {author} {\bibfnamefont
  {S.}~\bibnamefont {Carter}}, \bibinfo {author} {\bibfnamefont
  {B.}~\bibnamefont {Hosmer}},  \emph {et~al.},\ }\href@noop {} {\bibfield
  {journal} {\bibinfo  {journal} {Transformer Circuits Thread}\ }\textbf
  {\bibinfo {volume} {6}},\ \bibinfo {pages} {16318} (\bibinfo {year}
  {2025})}\BibitemShut {NoStop}%
\bibitem [{\citenamefont {Geiger}\ \emph {et~al.}(2021)\citenamefont {Geiger},
  \citenamefont {Lu}, \citenamefont {Icard},\ and\ \citenamefont
  {Potts}}]{geiger2021causal}%
  \BibitemOpen
  \bibfield  {author} {\bibinfo {author} {\bibfnamefont {A.}~\bibnamefont
  {Geiger}}, \bibinfo {author} {\bibfnamefont {H.}~\bibnamefont {Lu}}, \bibinfo
  {author} {\bibfnamefont {T.}~\bibnamefont {Icard}}, \ and\ \bibinfo {author}
  {\bibfnamefont {C.}~\bibnamefont {Potts}},\ }\href@noop {} {\bibfield
  {journal} {\bibinfo  {journal} {Advances in neural information processing
  systems}\ }\textbf {\bibinfo {volume} {34}},\ \bibinfo {pages} {9574}
  (\bibinfo {year} {2021})}\BibitemShut {NoStop}%
\bibitem [{\citenamefont {Geiger}\ \emph {et~al.}(2025)\citenamefont {Geiger},
  \citenamefont {Ibeling}, \citenamefont {Zur}, \citenamefont {Chaudhary},
  \citenamefont {Chauhan}, \citenamefont {Huang}, \citenamefont {Arora},
  \citenamefont {Wu}, \citenamefont {Goodman}, \citenamefont {Potts} \emph
  {et~al.}}]{geiger2025causal}%
  \BibitemOpen
  \bibfield  {author} {\bibinfo {author} {\bibfnamefont {A.}~\bibnamefont
  {Geiger}}, \bibinfo {author} {\bibfnamefont {D.}~\bibnamefont {Ibeling}},
  \bibinfo {author} {\bibfnamefont {A.}~\bibnamefont {Zur}}, \bibinfo {author}
  {\bibfnamefont {M.}~\bibnamefont {Chaudhary}}, \bibinfo {author}
  {\bibfnamefont {S.}~\bibnamefont {Chauhan}}, \bibinfo {author} {\bibfnamefont
  {J.}~\bibnamefont {Huang}}, \bibinfo {author} {\bibfnamefont
  {A.}~\bibnamefont {Arora}}, \bibinfo {author} {\bibfnamefont
  {Z.}~\bibnamefont {Wu}}, \bibinfo {author} {\bibfnamefont {N.}~\bibnamefont
  {Goodman}}, \bibinfo {author} {\bibfnamefont {C.}~\bibnamefont {Potts}},
  \emph {et~al.},\ }\href@noop {} {\bibfield  {journal} {\bibinfo  {journal}
  {Journal of Machine Learning Research}\ }\textbf {\bibinfo {volume} {26}},\
  \bibinfo {pages} {1} (\bibinfo {year} {2025})}\BibitemShut {NoStop}%
\bibitem [{\citenamefont {Varley}\ \emph
  {et~al.}(2023{\natexlab{c}})\citenamefont {Varley}, \citenamefont {Pope},
  \citenamefont {Grazia}, \citenamefont {Joshua},\ and\ \citenamefont
  {Sporns}}]{varley2023partial}%
  \BibitemOpen
  \bibfield  {author} {\bibinfo {author} {\bibfnamefont {T.~F.}\ \bibnamefont
  {Varley}}, \bibinfo {author} {\bibfnamefont {M.}~\bibnamefont {Pope}},
  \bibinfo {author} {\bibfnamefont {M.}~\bibnamefont {Grazia}}, \bibinfo
  {author} {\bibnamefont {Joshua}}, \ and\ \bibinfo {author} {\bibfnamefont
  {O.}~\bibnamefont {Sporns}},\ }\href@noop {} {\bibfield  {journal} {\bibinfo
  {journal} {Proceedings of the National Academy of Sciences}\ }\textbf
  {\bibinfo {volume} {120}},\ \bibinfo {pages} {e2300888120} (\bibinfo {year}
  {2023}{\natexlab{c}})}\BibitemShut {NoStop}%
\bibitem [{\citenamefont {Hillebrand}\ \emph {et~al.}(2016)\citenamefont
  {Hillebrand}, \citenamefont {Tewarie}, \citenamefont {Van~Dellen},
  \citenamefont {Yu}, \citenamefont {Carbo}, \citenamefont {Douw},
  \citenamefont {Gouw}, \citenamefont {Van~Straaten},\ and\ \citenamefont
  {Stam}}]{hillebrand2016direction}%
  \BibitemOpen
  \bibfield  {author} {\bibinfo {author} {\bibfnamefont {A.}~\bibnamefont
  {Hillebrand}}, \bibinfo {author} {\bibfnamefont {P.}~\bibnamefont {Tewarie}},
  \bibinfo {author} {\bibfnamefont {E.}~\bibnamefont {Van~Dellen}}, \bibinfo
  {author} {\bibfnamefont {M.}~\bibnamefont {Yu}}, \bibinfo {author}
  {\bibfnamefont {E.~W.}\ \bibnamefont {Carbo}}, \bibinfo {author}
  {\bibfnamefont {L.}~\bibnamefont {Douw}}, \bibinfo {author} {\bibfnamefont
  {A.~A.}\ \bibnamefont {Gouw}}, \bibinfo {author} {\bibfnamefont {E.~C.}\
  \bibnamefont {Van~Straaten}}, \ and\ \bibinfo {author} {\bibfnamefont
  {C.~J.}\ \bibnamefont {Stam}},\ }\href@noop {} {\bibfield  {journal}
  {\bibinfo  {journal} {Proceedings of the National Academy of Sciences}\
  }\textbf {\bibinfo {volume} {113}},\ \bibinfo {pages} {3867} (\bibinfo {year}
  {2016})}\BibitemShut {NoStop}%
\bibitem [{\citenamefont {Stramaglia}\ \emph {et~al.}(2014)\citenamefont
  {Stramaglia}, \citenamefont {M~Cortes},\ and\ \citenamefont
  {Marinazzo}}]{stramaglia2014synergy}%
  \BibitemOpen
  \bibfield  {author} {\bibinfo {author} {\bibfnamefont {S.}~\bibnamefont
  {Stramaglia}}, \bibinfo {author} {\bibfnamefont {J.}~\bibnamefont
  {M~Cortes}}, \ and\ \bibinfo {author} {\bibfnamefont {D.}~\bibnamefont
  {Marinazzo}},\ }\href@noop {} {\bibfield  {journal} {\bibinfo  {journal} {New
  Journal of Physics}\ }\textbf {\bibinfo {volume} {16}},\ \bibinfo {pages}
  {105003} (\bibinfo {year} {2014})}\BibitemShut {NoStop}%
\bibitem [{\citenamefont {Stramaglia}\ \emph {et~al.}(2024)\citenamefont
  {Stramaglia}, \citenamefont {Faes}, \citenamefont {Cortes},\ and\
  \citenamefont {Marinazzo}}]{stramaglia2024disentangling}%
  \BibitemOpen
  \bibfield  {author} {\bibinfo {author} {\bibfnamefont {S.}~\bibnamefont
  {Stramaglia}}, \bibinfo {author} {\bibfnamefont {L.}~\bibnamefont {Faes}},
  \bibinfo {author} {\bibfnamefont {J.~M.}\ \bibnamefont {Cortes}}, \ and\
  \bibinfo {author} {\bibfnamefont {D.}~\bibnamefont {Marinazzo}},\ }\href@noop
  {} {\bibfield  {journal} {\bibinfo  {journal} {Physical Review Research}\
  }\textbf {\bibinfo {volume} {6}},\ \bibinfo {pages} {L032007} (\bibinfo
  {year} {2024})}\BibitemShut {NoStop}%
\bibitem [{\citenamefont {Faes}\ \emph {et~al.}(2017)\citenamefont {Faes},
  \citenamefont {Marinazzo},\ and\ \citenamefont
  {Stramaglia}}]{faes2017multiscale}%
  \BibitemOpen
  \bibfield  {author} {\bibinfo {author} {\bibfnamefont {L.}~\bibnamefont
  {Faes}}, \bibinfo {author} {\bibfnamefont {D.}~\bibnamefont {Marinazzo}}, \
  and\ \bibinfo {author} {\bibfnamefont {S.}~\bibnamefont {Stramaglia}},\
  }\href@noop {} {\bibfield  {journal} {\bibinfo  {journal} {Entropy}\ }\textbf
  {\bibinfo {volume} {19}},\ \bibinfo {pages} {408} (\bibinfo {year}
  {2017})}\BibitemShut {NoStop}%
\bibitem [{\citenamefont {Shine}\ \emph {et~al.}(2016)\citenamefont {Shine},
  \citenamefont {Bissett}, \citenamefont {Bell}, \citenamefont {Koyejo},
  \citenamefont {Balsters}, \citenamefont {Gorgolewski}, \citenamefont
  {Moodie},\ and\ \citenamefont {Poldrack}}]{shine2016dynamics}%
  \BibitemOpen
  \bibfield  {author} {\bibinfo {author} {\bibfnamefont {J.~M.}\ \bibnamefont
  {Shine}}, \bibinfo {author} {\bibfnamefont {P.~G.}\ \bibnamefont {Bissett}},
  \bibinfo {author} {\bibfnamefont {P.~T.}\ \bibnamefont {Bell}}, \bibinfo
  {author} {\bibfnamefont {O.}~\bibnamefont {Koyejo}}, \bibinfo {author}
  {\bibfnamefont {J.~H.}\ \bibnamefont {Balsters}}, \bibinfo {author}
  {\bibfnamefont {K.~J.}\ \bibnamefont {Gorgolewski}}, \bibinfo {author}
  {\bibfnamefont {C.~A.}\ \bibnamefont {Moodie}}, \ and\ \bibinfo {author}
  {\bibfnamefont {R.~A.}\ \bibnamefont {Poldrack}},\ }\href@noop {} {\bibfield
  {journal} {\bibinfo  {journal} {Neuron}\ }\textbf {\bibinfo {volume} {92}},\
  \bibinfo {pages} {544} (\bibinfo {year} {2016})}\BibitemShut {NoStop}%
\bibitem [{\citenamefont {Deco}\ \emph {et~al.}(2021)\citenamefont {Deco},
  \citenamefont {Vidaurre},\ and\ \citenamefont
  {Kringelbach}}]{deco2021revisiting}%
  \BibitemOpen
  \bibfield  {author} {\bibinfo {author} {\bibfnamefont {G.}~\bibnamefont
  {Deco}}, \bibinfo {author} {\bibfnamefont {D.}~\bibnamefont {Vidaurre}}, \
  and\ \bibinfo {author} {\bibfnamefont {M.~L.}\ \bibnamefont {Kringelbach}},\
  }\href@noop {} {\bibfield  {journal} {\bibinfo  {journal} {Nature human
  behaviour}\ }\textbf {\bibinfo {volume} {5}},\ \bibinfo {pages} {497}
  (\bibinfo {year} {2021})}\BibitemShut {NoStop}%
\bibitem [{\citenamefont {Rosas}\ \emph {et~al.}(2026)\citenamefont {Rosas},
  \citenamefont {Mediano}, \citenamefont {Luppi}, \citenamefont {Boyd},
  \citenamefont {Berta}, \citenamefont {Jensen}, \citenamefont {Gastpar},\ and\
  \citenamefont {Carrol}}]{rosas2026faces}%
  \BibitemOpen
  \bibfield  {author} {\bibinfo {author} {\bibfnamefont {F.}~\bibnamefont
  {Rosas}}, \bibinfo {author} {\bibfnamefont {P.}~\bibnamefont {Mediano}},
  \bibinfo {author} {\bibfnamefont {A.}~\bibnamefont {Luppi}}, \bibinfo
  {author} {\bibfnamefont {A.}~\bibnamefont {Boyd}}, \bibinfo {author}
  {\bibfnamefont {M.}~\bibnamefont {Berta}}, \bibinfo {author} {\bibfnamefont
  {H.}~\bibnamefont {Jensen}}, \bibinfo {author} {\bibfnamefont
  {M.}~\bibnamefont {Gastpar}}, \ and\ \bibinfo {author} {\bibfnamefont
  {S.}~\bibnamefont {Carrol}},\ }\href@noop {} {\bibfield  {journal} {\bibinfo
  {journal} {Nature Reviews Physics}\ } (\bibinfo {year} {2026})},\ \bibinfo
  {note} {in press}\BibitemShut {NoStop}%
\bibitem [{\citenamefont {Kim}\ \emph {et~al.}(2021)\citenamefont {Kim},
  \citenamefont {T.~Jakobsen}, \citenamefont {Natarajan},\ and\ \citenamefont
  {Won}}]{kim2021tenet}%
  \BibitemOpen
  \bibfield  {author} {\bibinfo {author} {\bibfnamefont {J.}~\bibnamefont
  {Kim}}, \bibinfo {author} {\bibfnamefont {S.}~\bibnamefont {T.~Jakobsen}},
  \bibinfo {author} {\bibfnamefont {K.~N.}\ \bibnamefont {Natarajan}}, \ and\
  \bibinfo {author} {\bibfnamefont {K.-J.}\ \bibnamefont {Won}},\ }\href@noop
  {} {\bibfield  {journal} {\bibinfo  {journal} {Nucleic acids research}\
  }\textbf {\bibinfo {volume} {49}},\ \bibinfo {pages} {e1} (\bibinfo {year}
  {2021})}\BibitemShut {NoStop}%
\bibitem [{\citenamefont {Runge}\ \emph {et~al.}(2019)\citenamefont {Runge},
  \citenamefont {Bathiany}, \citenamefont {Bollt}, \citenamefont {Camps-Valls},
  \citenamefont {Coumou}, \citenamefont {Deyle}, \citenamefont {Glymour},
  \citenamefont {Kretschmer}, \citenamefont {Mahecha}, \citenamefont
  {Mu{\~n}oz-Mar{\'\i}} \emph {et~al.}}]{runge2019inferring}%
  \BibitemOpen
  \bibfield  {author} {\bibinfo {author} {\bibfnamefont {J.}~\bibnamefont
  {Runge}}, \bibinfo {author} {\bibfnamefont {S.}~\bibnamefont {Bathiany}},
  \bibinfo {author} {\bibfnamefont {E.}~\bibnamefont {Bollt}}, \bibinfo
  {author} {\bibfnamefont {G.}~\bibnamefont {Camps-Valls}}, \bibinfo {author}
  {\bibfnamefont {D.}~\bibnamefont {Coumou}}, \bibinfo {author} {\bibfnamefont
  {E.}~\bibnamefont {Deyle}}, \bibinfo {author} {\bibfnamefont
  {C.}~\bibnamefont {Glymour}}, \bibinfo {author} {\bibfnamefont
  {M.}~\bibnamefont {Kretschmer}}, \bibinfo {author} {\bibfnamefont {M.~D.}\
  \bibnamefont {Mahecha}}, \bibinfo {author} {\bibfnamefont {J.}~\bibnamefont
  {Mu{\~n}oz-Mar{\'\i}}},  \emph {et~al.},\ }\href@noop {} {\bibfield
  {journal} {\bibinfo  {journal} {Nature communications}\ }\textbf {\bibinfo
  {volume} {10}},\ \bibinfo {pages} {2553} (\bibinfo {year}
  {2019})}\BibitemShut {NoStop}%
\bibitem [{\citenamefont {Belloli}\ \emph {et~al.}(2025)\citenamefont
  {Belloli}, \citenamefont {Mediano}, \citenamefont {Cofr{\'e}}, \citenamefont
  {Slezak},\ and\ \citenamefont {Herzog}}]{belloli2025thoi}%
  \BibitemOpen
  \bibfield  {author} {\bibinfo {author} {\bibfnamefont {L.}~\bibnamefont
  {Belloli}}, \bibinfo {author} {\bibfnamefont {P.}~\bibnamefont {Mediano}},
  \bibinfo {author} {\bibfnamefont {R.}~\bibnamefont {Cofr{\'e}}}, \bibinfo
  {author} {\bibfnamefont {D.~F.}\ \bibnamefont {Slezak}}, \ and\ \bibinfo
  {author} {\bibfnamefont {R.}~\bibnamefont {Herzog}},\ }\href@noop {}
  {\bibfield  {journal} {\bibinfo  {journal} {arXiv preprint arXiv:2501.03381}\
  } (\bibinfo {year} {2025})}\BibitemShut {NoStop}%
\bibitem [{\citenamefont {Venkatesh}\ \emph {et~al.}(2023)\citenamefont
  {Venkatesh}, \citenamefont {Bennett}, \citenamefont {Gale}, \citenamefont
  {Ramirez}, \citenamefont {Heller}, \citenamefont {Durand}, \citenamefont
  {Olsen},\ and\ \citenamefont {Mihalas}}]{venkatesh2023gaussian}%
  \BibitemOpen
  \bibfield  {author} {\bibinfo {author} {\bibfnamefont {P.}~\bibnamefont
  {Venkatesh}}, \bibinfo {author} {\bibfnamefont {C.}~\bibnamefont {Bennett}},
  \bibinfo {author} {\bibfnamefont {S.}~\bibnamefont {Gale}}, \bibinfo {author}
  {\bibfnamefont {T.}~\bibnamefont {Ramirez}}, \bibinfo {author} {\bibfnamefont
  {G.}~\bibnamefont {Heller}}, \bibinfo {author} {\bibfnamefont
  {S.}~\bibnamefont {Durand}}, \bibinfo {author} {\bibfnamefont
  {S.}~\bibnamefont {Olsen}}, \ and\ \bibinfo {author} {\bibfnamefont
  {S.}~\bibnamefont {Mihalas}},\ }\href@noop {} {\bibfield  {journal} {\bibinfo
   {journal} {Advances in Neural Information Processing Systems}\ }\textbf
  {\bibinfo {volume} {36}},\ \bibinfo {pages} {74602} (\bibinfo {year}
  {2023})}\BibitemShut {NoStop}%
\bibitem [{\citenamefont {Liardi}\ \emph {et~al.}(2025)\citenamefont {Liardi},
  \citenamefont {Rosas}, \citenamefont {Carhart-Harris}, \citenamefont
  {Blackburne}, \citenamefont {Bor},\ and\ \citenamefont
  {Mediano}}]{liardi2025null}%
  \BibitemOpen
  \bibfield  {author} {\bibinfo {author} {\bibfnamefont {A.}~\bibnamefont
  {Liardi}}, \bibinfo {author} {\bibfnamefont {F.~E.}\ \bibnamefont {Rosas}},
  \bibinfo {author} {\bibfnamefont {R.~L.}\ \bibnamefont {Carhart-Harris}},
  \bibinfo {author} {\bibfnamefont {G.}~\bibnamefont {Blackburne}}, \bibinfo
  {author} {\bibfnamefont {D.}~\bibnamefont {Bor}}, \ and\ \bibinfo {author}
  {\bibfnamefont {P.~A.}\ \bibnamefont {Mediano}},\ }\href@noop {} {\bibfield
  {journal} {\bibinfo  {journal} {PLOS Computational Biology}\ }\textbf
  {\bibinfo {volume} {21}},\ \bibinfo {pages} {e1013629} (\bibinfo {year}
  {2025})}\BibitemShut {NoStop}%
\bibitem [{\citenamefont {Cover}(1999)}]{cover1999elements}%
  \BibitemOpen
  \bibfield  {author} {\bibinfo {author} {\bibfnamefont {T.~M.}\ \bibnamefont
  {Cover}},\ }\href@noop {} {\emph {\bibinfo {title} {Elements of information
  theory}}}\ (\bibinfo  {publisher} {John Wiley \& Sons},\ \bibinfo {year}
  {1999})\BibitemShut {NoStop}%
\bibitem [{\citenamefont {Geiger}(2021)}]{geiger_information_2021}%
  \BibitemOpen
  \bibfield  {author} {\bibinfo {author} {\bibfnamefont {B.~C.}\ \bibnamefont
  {Geiger}},\ }\href@noop {} {\bibfield  {journal} {\bibinfo  {journal} {IEEE
  Transactions on Neural Networks and Learning Systems}\ } (\bibinfo {year}
  {2021})}\BibitemShut {NoStop}%
\bibitem [{\citenamefont {Yadav}\ and\ \citenamefont
  {Bottou}(2019)}]{yadav2019cold}%
  \BibitemOpen
  \bibfield  {author} {\bibinfo {author} {\bibfnamefont {C.}~\bibnamefont
  {Yadav}}\ and\ \bibinfo {author} {\bibfnamefont {L.}~\bibnamefont {Bottou}},\
  }in\ \href@noop {} {\emph {\bibinfo {booktitle} {Advances in Neural
  Information Processing Systems}}},\ Vol.~\bibinfo {volume} {32}\ (\bibinfo
  {year} {2019})\BibitemShut {NoStop}%
\bibitem [{\citenamefont {Geweke}(1982)}]{geweke1982measurement}%
  \BibitemOpen
  \bibfield  {author} {\bibinfo {author} {\bibfnamefont {J.}~\bibnamefont
  {Geweke}},\ }\href@noop {} {\bibfield  {journal} {\bibinfo  {journal}
  {Journal of the American statistical association}\ }\textbf {\bibinfo
  {volume} {77}},\ \bibinfo {pages} {304} (\bibinfo {year} {1982})}\BibitemShut
  {NoStop}%
\bibitem [{\citenamefont {Geweke}(1984)}]{geweke1984measures}%
  \BibitemOpen
  \bibfield  {author} {\bibinfo {author} {\bibfnamefont {J.~F.}\ \bibnamefont
  {Geweke}},\ }\href@noop {} {\bibfield  {journal} {\bibinfo  {journal}
  {Journal of the American Statistical Association}\ }\textbf {\bibinfo
  {volume} {79}},\ \bibinfo {pages} {907} (\bibinfo {year} {1984})}\BibitemShut
  {NoStop}%
\bibitem [{\citenamefont {Henschel}\ \emph {et~al.}(2020)\citenamefont
  {Henschel}, \citenamefont {Conjeti}, \citenamefont {Estrada}, \citenamefont
  {Diers}, \citenamefont {Fischl},\ and\ \citenamefont
  {Reuter}}]{henschel2020fastsurfer}%
  \BibitemOpen
  \bibfield  {author} {\bibinfo {author} {\bibfnamefont {L.}~\bibnamefont
  {Henschel}}, \bibinfo {author} {\bibfnamefont {S.}~\bibnamefont {Conjeti}},
  \bibinfo {author} {\bibfnamefont {S.}~\bibnamefont {Estrada}}, \bibinfo
  {author} {\bibfnamefont {K.}~\bibnamefont {Diers}}, \bibinfo {author}
  {\bibfnamefont {B.}~\bibnamefont {Fischl}}, \ and\ \bibinfo {author}
  {\bibfnamefont {M.}~\bibnamefont {Reuter}},\ }\href@noop {} {\bibfield
  {journal} {\bibinfo  {journal} {NeuroImage}\ }\textbf {\bibinfo {volume}
  {219}},\ \bibinfo {pages} {117012} (\bibinfo {year} {2020})}\BibitemShut
  {NoStop}%
\bibitem [{\citenamefont {Klein}\ and\ \citenamefont
  {Tourville}(2012)}]{klein2012101}%
  \BibitemOpen
  \bibfield  {author} {\bibinfo {author} {\bibfnamefont {A.}~\bibnamefont
  {Klein}}\ and\ \bibinfo {author} {\bibfnamefont {J.}~\bibnamefont
  {Tourville}},\ }\href@noop {} {\bibfield  {journal} {\bibinfo  {journal}
  {Frontiers in neuroscience}\ }\textbf {\bibinfo {volume} {6}},\ \bibinfo
  {pages} {33392} (\bibinfo {year} {2012})}\BibitemShut {NoStop}%
\bibitem [{\citenamefont {Gramfort}\ \emph {et~al.}(2013)\citenamefont
  {Gramfort}, \citenamefont {Luessi}, \citenamefont {Larson}, \citenamefont
  {Engemann}, \citenamefont {Strohmeier}, \citenamefont {Brodbeck},
  \citenamefont {Goj}, \citenamefont {Jas}, \citenamefont {Brooks},
  \citenamefont {Parkkonen} \emph {et~al.}}]{gramfort2013meg}%
  \BibitemOpen
  \bibfield  {author} {\bibinfo {author} {\bibfnamefont {A.}~\bibnamefont
  {Gramfort}}, \bibinfo {author} {\bibfnamefont {M.}~\bibnamefont {Luessi}},
  \bibinfo {author} {\bibfnamefont {E.}~\bibnamefont {Larson}}, \bibinfo
  {author} {\bibfnamefont {D.~A.}\ \bibnamefont {Engemann}}, \bibinfo {author}
  {\bibfnamefont {D.}~\bibnamefont {Strohmeier}}, \bibinfo {author}
  {\bibfnamefont {C.}~\bibnamefont {Brodbeck}}, \bibinfo {author}
  {\bibfnamefont {R.}~\bibnamefont {Goj}}, \bibinfo {author} {\bibfnamefont
  {M.}~\bibnamefont {Jas}}, \bibinfo {author} {\bibfnamefont {T.}~\bibnamefont
  {Brooks}}, \bibinfo {author} {\bibfnamefont {L.}~\bibnamefont {Parkkonen}},
  \emph {et~al.},\ }\href@noop {} {\bibfield  {journal} {\bibinfo  {journal}
  {Frontiers in Neuroinformatics}\ }\textbf {\bibinfo {volume} {7}},\ \bibinfo
  {pages} {267} (\bibinfo {year} {2013})}\BibitemShut {NoStop}%
\bibitem [{\citenamefont {Van~Veen}\ \emph {et~al.}(1997)\citenamefont
  {Van~Veen}, \citenamefont {Van~Drongelen}, \citenamefont {Yuchtman},\ and\
  \citenamefont {Suzuki}}]{van1997localization}%
  \BibitemOpen
  \bibfield  {author} {\bibinfo {author} {\bibfnamefont {B.~D.}\ \bibnamefont
  {Van~Veen}}, \bibinfo {author} {\bibfnamefont {W.}~\bibnamefont
  {Van~Drongelen}}, \bibinfo {author} {\bibfnamefont {M.}~\bibnamefont
  {Yuchtman}}, \ and\ \bibinfo {author} {\bibfnamefont {A.}~\bibnamefont
  {Suzuki}},\ }\href@noop {} {\bibfield  {journal} {\bibinfo  {journal} {IEEE
  Transactions on biomedical engineering}\ }\textbf {\bibinfo {volume} {44}},\
  \bibinfo {pages} {867} (\bibinfo {year} {1997})}\BibitemShut {NoStop}%
\bibitem [{\citenamefont {Barnett}\ and\ \citenamefont
  {Seth}(2014)}]{barnett2014mvgc}%
  \BibitemOpen
  \bibfield  {author} {\bibinfo {author} {\bibfnamefont {L.}~\bibnamefont
  {Barnett}}\ and\ \bibinfo {author} {\bibfnamefont {A.~K.}\ \bibnamefont
  {Seth}},\ }\href@noop {} {\bibfield  {journal} {\bibinfo  {journal} {Journal
  of neuroscience methods}\ }\textbf {\bibinfo {volume} {223}},\ \bibinfo
  {pages} {50} (\bibinfo {year} {2014})}\BibitemShut {NoStop}%
\bibitem [{\citenamefont {Mages}\ \emph {et~al.}(2024)\citenamefont {Mages},
  \citenamefont {Anastasiadi},\ and\ \citenamefont {Rohner}}]{mages2024non}%
  \BibitemOpen
  \bibfield  {author} {\bibinfo {author} {\bibfnamefont {T.}~\bibnamefont
  {Mages}}, \bibinfo {author} {\bibfnamefont {E.}~\bibnamefont {Anastasiadi}},
  \ and\ \bibinfo {author} {\bibfnamefont {C.}~\bibnamefont {Rohner}},\
  }\href@noop {} {\bibfield  {journal} {\bibinfo  {journal} {Entropy}\ }\textbf
  {\bibinfo {volume} {26}},\ \bibinfo {pages} {424} (\bibinfo {year}
  {2024})}\BibitemShut {NoStop}%
\bibitem [{\citenamefont {Breiman}(2001)}]{breiman2001random}%
  \BibitemOpen
  \bibfield  {author} {\bibinfo {author} {\bibfnamefont {L.}~\bibnamefont
  {Breiman}},\ }\href@noop {} {\bibfield  {journal} {\bibinfo  {journal}
  {Machine Learning}\ }\textbf {\bibinfo {volume} {45}},\ \bibinfo {pages} {5}
  (\bibinfo {year} {2001})}\BibitemShut {NoStop}%
\bibitem [{\citenamefont {Fano}(1961)}]{fano1961transmission}%
  \BibitemOpen
  \bibfield  {author} {\bibinfo {author} {\bibfnamefont {R.~M.}\ \bibnamefont
  {Fano}},\ }\href@noop {} {\emph {\bibinfo {title} {Transmission of
  Information: A Statistical Theory of Communications}}}\ (\bibinfo
  {publisher} {MIT Press},\ \bibinfo {address} {Cambridge, MA},\ \bibinfo
  {year} {1961})\BibitemShut {NoStop}%
\bibitem [{\citenamefont {Hellman}\ and\ \citenamefont
  {Raviv}(1970)}]{hellman1970probability}%
  \BibitemOpen
  \bibfield  {author} {\bibinfo {author} {\bibfnamefont {M.~E.}\ \bibnamefont
  {Hellman}}\ and\ \bibinfo {author} {\bibfnamefont {J.}~\bibnamefont
  {Raviv}},\ }\href@noop {} {\bibfield  {journal} {\bibinfo  {journal} {IEEE
  Transactions on Information Theory}\ }\textbf {\bibinfo {volume} {16}},\
  \bibinfo {pages} {368} (\bibinfo {year} {1970})}\BibitemShut {NoStop}%
\bibitem [{\citenamefont {Feder}\ and\ \citenamefont
  {Merhav}(1994)}]{feder1994relations}%
  \BibitemOpen
  \bibfield  {author} {\bibinfo {author} {\bibfnamefont {M.}~\bibnamefont
  {Feder}}\ and\ \bibinfo {author} {\bibfnamefont {N.}~\bibnamefont {Merhav}},\
  }\href@noop {} {\bibfield  {journal} {\bibinfo  {journal} {IEEE Transactions
  on Information Theory}\ }\textbf {\bibinfo {volume} {40}},\ \bibinfo {pages}
  {259} (\bibinfo {year} {1994})}\BibitemShut {NoStop}%
\bibitem [{\citenamefont {Sason}\ and\ \citenamefont
  {Verd{\'u}}(2018)}]{sason2018arimoto}%
  \BibitemOpen
  \bibfield  {author} {\bibinfo {author} {\bibfnamefont {I.}~\bibnamefont
  {Sason}}\ and\ \bibinfo {author} {\bibfnamefont {S.}~\bibnamefont
  {Verd{\'u}}},\ }\href@noop {} {\bibfield  {journal} {\bibinfo  {journal}
  {IEEE Transactions on Information Theory}\ }\textbf {\bibinfo {volume}
  {64}},\ \bibinfo {pages} {4} (\bibinfo {year} {2018})}\BibitemShut {NoStop}%
\end{thebibliography}%
\bibliographystyle{apsrev4-1}

\end{document}